%% file: main.tex
\documentclass[11pt]{article}
\usepackage{tikz}

\usepackage{verbatim}
\usetikzlibrary{arrows.meta,shapes}

\pgfdeclarelayer{background}
\pgfsetlayers{background,main}

\def\showauthornotes{1}

\def\showdraftbox{0}

\input{macros}

\allowdisplaybreaks[1]

\ifnum\showauthornotes=1
\newcommand{\todo}[1]{\colorbox{Mygray}{\color{red}\parbox{\textwidth}{#1}}}
\else
\newcommand{\todo}[1]{}
\fi

\newcommand{\eps}{\varepsilon}

\newcommand{\lap}{\ensuremath{\boldsymbol{L}}}

\newcommand\Otil[1]{\ensuremath{\widetilde{O}\left(#1\right)}}
\newcommand\otil[1]{\ensuremath{\widetilde{\mathit{O}}(#1)}}

\newcommand\Ccal{\mathcal{C}}

\renewcommand\AA{\boldsymbol{\mathit{A}}}
\newcommand\DD{\boldsymbol{\mathit{D}}}

\newcommand\MM{\boldsymbol{\mathit{M}}}

\newcommand\II{\boldsymbol{\mathit{I}}}
\newcommand\LL{\boldsymbol{\mathit{L}}}
\newcommand\LLtil{\widetilde{\boldsymbol{L}}}
\newcommand\PP{\boldsymbol{\mathit{P}}}
\newcommand\QQ{\boldsymbol{\mathit{Q}}}
\renewcommand\SS{\boldsymbol{\mathit{S}}}
\newcommand\TT{\boldsymbol{\mathit{T}}}

\newcommand\XX{\boldsymbol{\mathit{X}}}

\newcommand\YY{\boldsymbol{\mathit{Y}}}

\newcommand\cchi{\boldsymbol{\chi}}

\newcommand\nhat{\widehat{n}}
\newcommand\mhat{\widehat{m}}

\newcommand\xhat{\widehat{x}}

\newcommand\Hhat{\widehat{H}}

\newcommand\Gtil{\widetilde{G}}

\newcommand\Kcal{\mathcal{K}}

\newcommand\Shat{\widehat{S}}

\newcommand\epshat{\widehat{\eps}}

\newcommand\dd{\boldsymbol{\mathit{d}}}
\newcommand\ww{\boldsymbol{\mathit{w}}}
\newcommand\xx{\boldsymbol{\mathit{x}}}
\newcommand\yy{\boldsymbol{\mathit{y}}}

\newcommand\rr{\boldsymbol{\mathit{r}}}

\newcommand\PPi{\boldsymbol{\mathit{\Pi}}}
\newcommand\one{\vec{1}}

\newcommand{\cupdot}{\mathbin{\mathaccent\cdot\cup}}

\newcommand{\DegPreserveSparsify}{{\textsc{DegreePreservingSparsify}}}
\newcommand{\SpectralSketch}{{{\textsc{SpectralSketch}}}}
\newcommand{\MoveEdges}{{\textsc{MoveEdges}}}
\newcommand{\MoveEdgesExpander}{{\textsc{MoveEdgesExpander}}}
\newcommand{\ExtractBoundedDegreeGraph}{{\textsc{ExtractBoundedDegreeGraph}}}
\newcommand{\NaiveCycleDecomposition}{{\textsc{NaiveCycleDecomposition}}}
\newcommand{\ShortCycleAlgo}{{\textsc{ShortCycleDecomposition}}}
\newcommand{\CycleDecomposition}{{\textsc{CycleDecomposition}}}
\newcommand{\ExpanderDecompose}{{\textsc{ExpanderDecompose}}}
\newcommand{\NSExpanderDecompose}{{\textsc{NSExpanderDecompose}}}
\newcommand{\SparsifyOnce}{{\textsc{SparsifyOnce}}}
\newcommand{\SampleMatchings}{{\textsc{SampleMatchings}}}
\newcommand{\ImplicitPartitionAndSample}{{\textsc{ImplicitPartitionAndSample}}}
\newcommand{\SC}{{\textsc{SC}}}
\newcommand{\MakeBalanced}{{\textsc{MakeBalanced}}}
\newcommand{\DecomposeAndSample}{{\textsc{DecomposeAndSample}}}
\newcommand{\SampleBiCliques}{{\textsc{SampleBiCliques}}}

\newcommand{\EulerianSparsify}{{\textsc{EulerianSparsify}}}
\newcommand{\DirectedSparsifyOnce}{{\textsc{DirectedSparsifyOnce}}}
\newcommand{\ImplicitSketchUnweightedBiCliques}{{\textsc{ImplicitSketchUnweightedBiCliques}}}

\newcommand{\SchurSparse}{{\textsc{SchurSparse}}}

\newcommand{\vol}[1]{\operatorname{vol}\left({#1}\right)}
\newcommand{\reff}{\mathop{R}_{\textrm{eff}}}

\newcommand{\dir}[1]{{\vec{#1}}}

\newcommand\expec[2]{\operatorname*{\mathbb{E}}_{#1}\left[ #2 \right]}
\newcommand\prob[2]{\operatorname*{\mathbb{P}}_{#1}\left[ #2 \right]}

\newcommand{\etal}{\emph{et al.}}

\newcommand{\TCycle}{\textsc{T}_{\textsc{CycleDecomp}}}

\newenvironment{tight_enumerate}{
\begin{enumerate}
  \setlength{\itemsep}{2pt}
  \setlength{\parskip}{1pt}
  \setlength{\partopsep}{1pt}
}{\end{enumerate}}
\newenvironment{tight_itemize}{
\begin{itemize}
  \setlength{\itemsep}{2pt}
  \setlength{\parskip}{1pt}
  \setlength{\partopsep}{1pt}
}{\end{itemize}}

\title{%
  Graph Sparsification, Spectral Sketches, and Faster Resistance
  Computation, via Short Cycle Decompositions
}%

\author{
  Timothy Chu\\
  Carnegie Mellon \\
  \texttt{tzchu@andrew.cmu.edu} 
  \and
  Yu Gao\\
  Georgia Tech \\
  \texttt{ygao380@gatech.edu}
  \and
  Richard Peng\\
  Georgia Tech \\
  \texttt{rpeng@cc.gatech.edu} 
  \and
  Sushant Sachdeva\\
  University of Toronto \\
  \texttt{sachdeva@cs.toronto.edu}
  \and
  Saurabh Sawlani\\
  Georgia Tech \\
  \texttt{sawlani@gatech.edu} 
  \and
  Junxing Wang\\
  Carnegie Mellon \\
  \texttt{junxingw@andrew.cmu.edu}
  }

\date{\today}
\begin{document}

\maketitle
\thispagestyle{empty}

\begin{abstract}
  We develop a framework for graph sparsification and sketching, based
  on a new tool, short cycle decomposition -- a decomposition of an
  unweighted graph into an edge-disjoint collection of short cycles,
  plus a small number of extra edges. A simple observation gives that
  every graph $G$ on $n$ vertices with $m$ edges can be decomposed in
  $O(mn)$ time into cycles of length at most $2 \log n,$ and at most
  $2n$ extra edges. We give an $m^{1+o(1)}$ time algorithm for
  constructing a short cycle decomposition, with cycles of length
  $n^{o(1)},$ and $n^{1+o(1)}$ extra edges.  Both the existential and
  algorithmic variants of this decomposition enable us to make
  progress on several open problems in randomized graph algorithms.
  \begin{enumerate}
  \item We present an algorithm that runs in time $m^{1+o(1)}\eps^{-1.5}$
    and returns $(1\pm\eps)$-approximations to effective resistances
    of all edges, improving over the previous best of
    $\otil{\min\{m\eps^{-2}, n^{2} \eps^{-1}\}}$.
    This routine in turn gives an algorithm to approximate the determinant
    of a graph Laplacian up to a factor of $(1\pm \eps)$ in
    $m^{1 + o(1)} + n^{\nfrac{15}{8}+o(1)}\eps^{-\nfrac{7}{4}}$ time.
  \item We show existence and efficient algorithms for constructing
    graphical spectral sketches -- a distribution over sparse graphs $H$
    such that for a fixed vector $\xx$, we have
    $\xx^{\top} \lap_H \xx = (1\pm\eps) \xx^{\top} \lap_G \xx$ and
    $\xx^{\top} \lap^{+}_H \xx = (1\pm\eps) \xx^{\top} \lap^{+}_G \xx$
    with high probability, where $\lap$ is the graph Laplacian
    and $\lap^{+}$ is its pseudoinverse.
      This implies the existence of
      resistance-sparsifiers with about $n \eps^{-1}$
    edges that preserve the effective resistances between 
    every pair of vertices up to $(1\pm\eps).$
  \item By combining short cycle decompositions with known tools in
  graph sparsification,
    we show the existence of nearly-linear sized degree-preserving
    spectral sparsifiers, as well as significantly sparser approximations
    of directed graphs. The latter is critical to recent
    breakthroughs on faster algorithms for solving linear systems in directed Laplacians.
  \end{enumerate}
  The running time and output qualities of our spectral sketch and
  degree-preserving (directed) sparsification algorithms are limited
  by the efficiency of our routines for constructing short cycle
  decompositions.  Improved algorithms for short cycle decompositions
  will lead to improvements for each of these algorithms.
\end{abstract}

\newpage

\input{intro}
\input{prelims}
\input{overview}
\input{deg-preserving}
\input{eulerian}
\input{resistance-sparsifiers}
\input{resistance-computation}
\input{construction}

\subsection*{Acknowledgments}
We thank John Peebles for many insightful discussions,
David Durfee for pointing out several key details in the
interaction with determinant estimation from Appendix~\ref{sec:Determinant},
and Di Wang for explaining to us the state-of-the-art
for efficiently generating expander decompositions.

{\small
\bibliographystyle{alpha}
\bibliography{refs}}
\begin{appendix}
\input{unit-weight}
\input{determinant}
\end{appendix}
\end{document}

%% file: macros.tex
\usepackage{geometry}
\usepackage[pagebackref]{hyperref}
\usepackage{amsmath,amssymb,amsthm,amstext,amsfonts,bbm,algorithm,algorithmicx,graphicx,xspace,nicefrac}
\usepackage{color,stmaryrd,enumerate,latexsym,bm,amsfonts,wrapfig,verbatim,tabularx,textcomp,
subfig}
\usepackage[square,sort,comma,numbers]{natbib}

\usepackage{amsfonts}
\usepackage{comment} 
\usepackage{epsfig} 
\usepackage{latexsym,nicefrac,bbm}
\usepackage{xspace}
\usepackage{color,fancybox,graphicx,url}
\usepackage{enumitem, fullpage}
\usepackage{booktabs}
\usepackage{commath}
\usepackage{mdframed}
\usepackage{pdfsync}

\usepackage{fancybox}
\newenvironment{fminipage}%
  {\begin{Sbox}\begin{minipage}}%
  {\end{minipage}\end{Sbox}\fbox{\TheSbox}}

\newcommand{\defeq}{\stackrel{\textup{def}}{=}}

\newtheorem{theorem}{Theorem}[section]

\newtheorem{lemma}[theorem]{Lemma}
\newtheorem{definition}[theorem]{Definition}
\newtheorem{corollary}[theorem]{Corollary}

\newtheorem{fact}[theorem]{Fact}

\DeclareMathOperator*{\var}{\bf Var}

\newcommand{\nfrac}[2]{\nicefrac{#1}{#2}}
\def\abs#1{\left| #1 \right|}
\renewcommand{\norm}[1]{\ensuremath{\left\lVert #1 \right\rVert}}

\newcommand{\ceil}[1]{\left\lceil\, {#1}\,\right\rceil}

\newcommand\rea{\mathbb R}

\newcommand{\marginlabel}[1]%
{\mbox{}\marginpar{\it{\raggedleft\hspace{0pt}#1}}}
\newcommand\poly{{\textrm{poly}}}  %usage \poly(n)

\newcommand\calG{\mathcal{G}}
\newcommand\calH{\mathcal{H}}

\newcommand\calK{\mathcal{K}}

\definecolor{Mygray}{gray}{0.8}

 \ifcsname ifcommentflag\endcsname\else
  \expandafter\let\csname ifcommentflag\expandafter\endcsname
                  \csname iffalse\endcsname
\fi

\ifnum\showauthornotes=1
\newcommand{\Authornote}[2]{{\sf\small\color{red}{[#1: #2]}}}
\newcommand{\Authoredit}[2]{{\sf\small\color{red}{[#1]}\color{blue}{#2}}}
\newcommand{\Authorcomment}[2]{{\sf \small\color{gray}{[#1: #2]}}}
\newcommand{\Authorfnote}[2]{\footnote{\color{red}{#1: #2}}}
\newcommand{\Authorfixme}[1]{\Authornote{#1}{\textbf{??}}}
\newcommand{\Authormarginmark}[1]{\marginpar{\textcolor{red}{\fbox{%\Large
#1:!}}}}
\else
\newcommand{\Authornote}[2]{}
\newcommand{\Authoredit}[2]{}
\newcommand{\Authorcomment}[2]{}
\newcommand{\Authorfnote}[2]{}
\newcommand{\Authorfixme}[1]{}
\newcommand{\Authormarginmark}[1]{}
\fi

\newlength{\pgmtab}  %  \pgmtab is the width of each tab in the
 {
	\begin{enumerate}}{\end{enumerate}}

\def\qedsketch{\ifmmode\Box\else{\unskip\nobreak\hfil
\penalty50\hskip1em\null\nobreak\hfil$\Box$
\parfillskip=0pt\finalhyphendemerits=0\endgraf}\fi}

\newlength{\tpush}
\newcommand{\handout}[5]{
   \noindent
   \begin{center}
   \framebox{ \vbox{ \hbox to \textwidth { {\bf \coursenum\ :\  \coursename} \hfill #5 }
       \vspace{3mm}
       \hbox to \textwidth { {\Large \hfill #2  \hfill} }
       \vspace{1mm}
       \hbox to \textwidth { {\it #3 \hfill #4} }
     }
   }
   \end{center}
   \vspace*{4mm}
   \newcommand{\lecturenum}{#1}
   \addcontentsline{toc}{chapter}{Lecture #1 -- #2}
}

\ifnum\showdraftbox=1

\else

\fi

\newcommand{\Mcal}{\mathcal{M}}
\newcommand{\chat}{\widehat{c}}
\newcommand{\fhat}{\widehat{f}}
\newcommand{\khat}{\widehat{k}}
\newcommand{\lhat}{\widehat{l}}
\newcommand{\gns}{\gamma_{\textsc{NS}}}
\newcommand{\gst}{\gamma_{\textsc{ST}}}

%% file: intro.tex
%!TEX root = main.tex

\section{Introduction}
\label{sec:Introduction}

Graph sparsification is a procedure that, given a graph $G,$ returns
another graph $H$, typically with much fewer edges, that approximately
preserves some characteristics of $G$. Graph sparsification originated
from the study of combinatorial graph algorithms related to cuts and
flows~\cite{BenczurK96, EppsteinGIN97}. Many different notions of
graph sparsification have been extensively studied, for instance,
spanners~\cite{Chew86} approximately preserve pairwise distances,
whereas cut-sparsification approximately preserves the sizes of all
cuts~\cite{BenczurK96}.  Spielman and Teng~\cite{SpielmanTengSolver:journal,
  SpielmanT11:journal} defined spectral sparsification, a notion
that's strictly stronger than a cut-sparsification.

Spectral sparsifiers have found numerous applications to graph
algorithms.  They are key to fast solvers for Laplacian linear
systems~\cite{SpielmanTengSolver:journal, SpielmanT11:journal,
  KoutisMP10, KoutisMP11}.  Recently they have been used as the
\emph{sole} graph theoretic primitive in graph algorithms including
solving linear systems~\cite{PengS14, KyngLPSS16}, sampling random
spanning trees~\cite{DurfeeKPRS17, DurfeeKPRS17}, measuring edge
centrality~\cite{LiZ18,LiPSYZ18:arxiv}, etc.

For an undirected, weighted graph $G = (V, E_G, w_G),$ we recall that
the Laplacian of $G,$ $\lap_G$ is the unique symmetric $V \times V$
matrix such that for all $x \in \rea^{V},$ we have
\[
  x^{\top}\lap_G x = \sum_{(u,v) \in E_G} w_G(u,v)(x_u - x_v)^{2}.
\]
For two positive scalars $a, b,$ we write $a \approx_{\eps} b$ if
$e^{-\eps} a \le b \le e^{\eps} a.$ We say the graph
$H = (V, E_H, w_H)$ is an $\eps$-spectral sparsifier of $G$ if,
\begin{equation}
  \label{eq:intro:forall-sparsifier}
  \forall x \in \rea^V, \qquad  x^{\top}\lap_Gx \approx_{\eps} x^{\top} \lap_H x.
\end{equation}
Restricting the above definition only to vectors
$x \in \{\pm 1\}^{V},$ one obtains cut sparsifiers. For a graph $G$
with $n$ vertices and $m$ edges, Spielman and Teng gave the first
algorithm for constructing spectral sparsifiers with
$\otil{n\eps^{-2}}$ edges\footnote{The $\otil{\cdot}$ notation hides
  $\poly(\log n)$ factors.}. Spielman and Srivastava~\cite{SpielmanS08:journal}
proved that one could construct a sparsifier for $G$ by independently
sampling $O(n \eps^{-2} \log n)$ edges with probabilities proportional
to their \emph{leverage scores} in $G.$ Finally, Batson, Spielman, and
Srivastava~\cite{BatsonSS12} proved that one could construct
sparsifiers with $O(n\eps^{-2})$ edges, and that this is optimal even
for constructing sparsifiers for the complete graph. Recently,
Carlson \etal~\cite{CarlsonKNT17:arxiv} have proved a more general lower
bound, proving that one needs $\Omega(n\eps^{-2} \log n)$ bits to
store any data structure that can approximately compute the sizes of
all cuts in $G.$

Given the tight upper and lower bounds, it is natural to guess at this
point that our understanding of graph sparsification is essentially
complete. However, numerous recent works have surprisingly brought to
attention several aspects that we do not seem to understand as yet.
\begin{tight_enumerate}
\item Are our bounds tight if we relax the requirement in
  Equation~\eqref{eq:intro:forall-sparsifier} to hold only for a fixed
  unknown $x$ with high probability?  Andoni
  \etal~\cite{AndoniCKQWZ16} define such an object to be a
  \emph{spectral sketch}. They also construct a \emph{data structure}
  (not a graph) with $\otil{n\eps^{-1}}$ space that is a spectral
  sketch for $x \in \{\pm 1\}^{V},$ even though $\Omega(n\eps^{-2})$
  is a lower bound if one must answer correctly for all
  $x \in \{\pm 1\}^{V}.$
  Building on
  their work, Jambulapati and Sidford~\cite{JambulapatiS18} showed how
  to construct such data structures that can answer queries for any
  $x$ with high probability. A natural question remains open: whether
  there exist \emph{graphs} that are spectral sketches with
  $\otil{n\eps^{-1}}$ edges?

\item What if we only want to preserve the effective
  resistance\footnote{The effective resistance between a pair $u,v$ is
    the voltage difference between $u, v$ if we consider the graph as
    an electrical network with every edge of weight $w_e$ as a
    resistor of resistance $\frac{1}{w_e},$ and we send one unit of
    current from $u$ to $v.$ } between all pairs of vertices?  Dinitz,
  Krauthgamer, and Wagner~\cite{DinitzKW15} define such a graph $H$ as a
  \emph{resistance sparsifier} of $G$, and show their existence for
  regular expanders with degree $\Omega(n).$ They conjecture that
  every graph admits an $\eps$-resistance sparsifier with
  $\otil{n\eps^{-1}}$ edges.
  
\item An $\eps$-spectral sparsifier preserves weighted vertex degrees
  up to $(1\pm\eps).$ Do there exist spectral sparsifiers that exactly
  preserve weighted degrees? Dinitz~\etal~\cite{DinitzKW15} also
  explicitly pose a related question -- does every dense regular
  expander contain a sparse regular expander?
\item What about sparsification for directed graphs? The above
  sparsification notions, and algorithms are difficult to generalize
  to directed graphs. Cohen~\etal~\cite{CohenKPPSV16} developed a
  notion of sparsification for Eulerian directed graphs (directed
  graphs with all vertices having in-degree equal to out-degree), and
  gave the first almost-linear time algorithms\footnote{An algorithm
    is said to be almost-linear time if it runs in $m^{1+o(1)}$ time
    on graphs with $m$ edges.} for building such sparsifiers. However,
  their algorithm is based on expander decomposition, and isn't as
  versatile as the importance sampling based sparsification of
  undirected graphs~\cite{SpielmanS08:journal}. Is there an easier
  approach to sparsifying Eulerian directed graphs?
\item There is an ever-growing body of work on the algorithmic
  applications of graph sparsification~\cite{Spielman10,
    Teng10:survey, BatsonSST13, Teng16:book}. Could the above improved
  guarantees lead to even faster algorithms for some of these
  problems? Two problems of significant interest include estimating
  determinants~\cite{DurfeeKPRS17} and sampling random spanning
  trees~\cite{DurfeeKPRS17, DurfeePPR17, Schild17:arxiv}.
\end{tight_enumerate}

\subsection{Our Contributions}
\label{subsec:contributions}
In this paper, we develop a framework for graph sparsification based
on a new graph-theoretic tool we call \emph{short cycle
  decomposition}. Informally, a short cycle decomposition of a graph
$G$ is a decomposition into a sparse graph, and several cycles of
short length.  We use our framework to give affirmative answers to all
the challenges in graph sparsification discussed in the previous
section. Specifically:
\begin{tight_enumerate}
\item We show that every graph $G$ has a graph $H$ with
  $\otil{n\eps^{-1}}$ edges that is an $\eps$-spectral-sketch for $G.$
  The existence of such
  graphic spectral-sketches was not known before.
  Moreover, we give an algorithm to construct an
  $\eps$-spectral-sketch with $n^{1+o(1)} \eps^{-1}$ edges in
  $m^{1+o(1)}$ time. In addition,  $H$ is also a spectral-sketch for $\LL_{G}^{+}.$

\item We show every graph $G$ has an $\eps$-resistance sparsifier with
  $\otil{n\eps^{-1}}$ edges, affirmatively answering the question
    raised by Dinitz~\etal~\cite{DinitzKW15}. We also give an
    algorithm to construct $\eps$-resistance sparsifiers with
    $n^{1+o(1)}\eps^{-1}$ edges in $m^{1+o(1)}$ time.

\item We show that every graph has an $\eps$-spectral sparsifier with
  $\otil{n\eps^{-2}}$ edges that exactly preserves the
  weighted-degrees of all vertices.  It follows that every dense
  regular expander contains a sparse (weighted) regular expander.
  Before our work, it was not known if there exist sparse
  degree-preserving sparsifiers (even for cut sparsifiers).

\item We show that short cycle decompositions can be used for
  constructing sparse spectral approximations for Eulerian directed
  graphs under the notion of spectral approximation given by
  Cohen~\etal~\cite{CohenKPPSV16} for Eulerian directed graphs
  (see~\ref{sec:overview:eulerian} for definition). We show that
  short-cycle decompositions are sufficient for sparsifying Eulerian
  directed graphs, and prove that every directed Eulerian graph has a
  spectral approximation with $O(n\eps^{-2}\log^{4} n)$ edges.
 
\item We build on our spectral-sketches, to give an algorithm for
  estimating the effective resistances of all edges up to a factor of
  $(1\pm \eps)$ in $m^{1+o(1)}\eps^{-1.5}$ time. The previous best
  results for this algorithm were
  $\otil{m\eps^{-2}}$~\cite{SpielmanS08:journal} and
  $\otil{n^{2}\eps^{-1}}$~\cite{JambulapatiS18}.

  Incorporating this result into the work of
  Durfee~\etal~\cite{DurfeePPR17} gives an
  $m^{1+o(1)} + n^{\nfrac{15}{8} + o(1)}\eps^{-\nfrac{7}{4}}$ time algorithm for
  approximating the determinant of a $\LL$ (rather, $\LL$ after
  deleting the last row and column, which is the number of spanning
  trees in a graph), up to a factor of $(1\pm\eps).$ The previous best
  algorithm for this problem ran in time
  $\otil{n^{2}\eps^{-2}}$~\cite{DurfeeKPRS17}.
\end{tight_enumerate}

As a key component of all our results, we present efficient algorithms
for constructing short cycle decompositions. From a bird's eye view,
the key advantage provided by short cycle decompositions for all the
above results, is that they allow us to sample edges in a coordinated
manner so to preserve weighted vertex degrees exactly.
\begin{definition}
\label{def:ShortCycleDecomposition}
An $(\mhat, L)$-short cycle decomposition of an unweighted undirected
graph $G$, decomposes $G$ into several edge-disjoint cycles, each of
length at most $L,$ and at most $\mhat$ edges not in these cycles.
\end{definition}
The existence of such a decomposition  with $\mhat(m, n) \leq 2n$
and $L(m, n) \leq 2\log{n}$ is a simple observation.
We repeatedly remove vertices of degree at most 2 from the graph,
along with their incident edges (removing at most $2n$ edges in total).
If the remaining graph has no cycle of length at most $2 \log n,$
a breadth-first search tree
of depth $\log n$ starting from any remaining vertex will contain more
than $n$ vertices, a contradiction. This can be implemented as an
$O(mn)$ time algorithm to find a $(2n, 2\log n)$-short cycle
decomposition, which in turn implies a similar running time
for all the existential results above.
Finding such decompositions faster is a core component of this paper:
we give an algorithm that constructs an $(n^{1+o(1)},n^{o(1)})$-short
cycle decomposition of a graph in $m^{1+o(1)}$ time.

\medskip
\noindent \textbf{Organization.}  To keep this section brief, we defer
the formal definitions and theorem statements to the overview of the
work (Section~\ref{sec:overview}), after defining a few necessary
preliminaries in Section~\ref{sec:prelims}. We start with
degree-preserving spectral sparsifiers in
Section~\ref{sec:Degree-Preserving}, and then give the algorithm for
sparsification of Eulerian directed graphs
(Section~\ref{sec:eulerian}). Next, we present the construction of
spectral-sketches and resistance sparsifiers in
Section~\ref{sec:Resistance-Sparsifiers}, followed by our algorithm
for estimating effective resistances for all edges in
Section~\ref{sec:Computing-Resistances}. Finally, we give our
almost-linear time algorithm for constructing a short cycle
decomposition in Section~\ref{sec:construction}.

%%% Local Variables:
%%% mode: latex
%%% TeX-master: "main"
%%% End:

%% file: prelims.tex
%!TEX root = main.tex

\section{Preliminaries}
\label{sec:prelims}
A square symmetric $n \times n$ matrix $\MM$ is positive semi-definite
(PSD), denoted $\MM \succeq 0,$ if for all $\xx \in \rea^{n},$ we have
$\xx^{\top} \MM \xx \ge 0.$ For two matrices $\MM_{1}, \MM_{2},$ we
write $\MM_{1} \succeq \MM_{2}$ if for all
$\xx \in \rea^{n}, \xx^{\top} \MM_{1} \xx \ge \xx^{\top} \MM_{2} \xx,$
or equivalently $\MM_{1} - \MM_{2} \succeq 0.$

For $\eps \ge 0,$ and two positive real numbers $a, b,$ we write
$a \approx_{\eps} b$ to express $e^{-\eps} a \le b \le e^{\eps} a.$
Observe that $a \approx_{\eps} b$ if $b \approx_{\eps} a.$ For two PSD
matrices $\MM_1, \MM_2,$ we write $\MM_1 \approx_{\eps} \MM_2$ if for
all
$\xx \in \rea^{n}, \xx^{\top} \MM_{1} \xx \approx_{\eps} \xx^{\top}
\MM_{2} \xx.$
\begin{fact}
  \label{fact:spectral-error-composition}
  For any PSD $\MM_{1}, \MM_{2}, \MM_{3}$ and
  $\eps_{1}, \eps_{2} \ge 0,$ if we have
  $\MM_{1} \approx_{\eps_1} \MM_{2}$ and
  $\MM_{2} \approx_{\eps_2} \MM_{3},$ then
  $\MM_{1} \approx_{\eps_1 + \eps_2} \MM_{3}.$
\end{fact}
For two graphs $G_1, G_2,$ we often abuse notation to
write $G_1 \succeq G_2$ to mean $\lap_{G_1} \succeq \lap_{G_2}$ and
$G_1 \approx_{\eps} G_2$ to mean $\lap_{G_1} \approx_{\eps} \lap_{G_2}.$

For any PSD matrix $\MM,$ we let $\MM^{+}$ denote the Moore-Penrose
pseudoinverse of $\MM.$ Thus, if $\MM$ has an eigenvalues
$0 \le \lambda_1 \le \lambda_{2} \le \ldots \le \lambda_{n},$ with
unit-norm eigenvectors $v_1, v_{2},\ldots, v_{n}$ respectively, we
have $\MM = \sum_{i} \lambda_{i} v_{i} v_{i}^{\top},$ and
$\MM^{+} = \sum_{\lambda_i > 0} \frac{1}{\lambda_{i}} v_{i}
v_{i}^{\top}.$ Similarly, we have
$\MM^{\nfrac{1}{2}} = \sum_{i} \sqrt{\lambda_{i}} v_{i} v_{i}^{\top},$
and
$\MM^{\nfrac{+}{2}} = \sum_{\lambda_i > 0}
\frac{1}{\sqrt{\lambda_{i}}} v_{i} v_{i}^{\top}.$

Our notion of approximation is preserved under inverses:
\begin{fact}
  \label{fact:spectral-error-invert}
  For any PSD $\MM_{1}$ and $\MM_{2}$,
  and any error $\eps > 0$, we have
  $\MM_{1} \approx_{\eps} \MM_{2}$
  if and only if $\MM_{1}^{+} \approx_{\eps} \MM^{+}_{2}$.
\end{fact}

For any $u,$ we let $\chi_{u}$ denote the vector such that the
$u^{\textrm{th}}$ coordinate is 1, and all other coordinates are 0. We
let $\chi_{uv} = \chi_{u} - \chi_{v}.$
For any edge $e=(u,v)$ in a connected graph $G,$ the effective
resistance of $e$ is defined as $\reff(e) = \chi_{uv}^{\top}
\LL_{G}^{+} \chi_{uv}.$ For a directed graph $\dir{G}$, its directed
Laplacian $\LL_{\dir{G}}$, can be defined as 
$\sum_{e=u\rightarrow v} \chi_{uv}^\top\chi_{u}.$

All logarithms throughout the paper are with base 2. Unless mentioned,
we assume that our input graph $G$ has $m$ edges and $n$
vertices. Throughout the paper, we consider graphs with positive
integral weights on the edges. Whenever we say the weights are poly
bounded, we assume they are bounded by $n^{O(1)}.$ The expression with
high-probability means with probability larger than $1-\frac{1}{n^{\Omega(1)}}.$

%%% Local Variables:
%%% mode: latex
%%% TeX-master: "main"
%%% End:

%% file: overview.tex
%!TEX root = main.tex

\section{Overview}
\label{sec:overview}

There are 4 major approaches to date towards graph sparsification:
expander
partitioning~\cite{SpielmanT11:journal,AndoniCKQWZ16,JambulapatiS18},
importance sampling~\cite{BenczurK96,SpielmanS08:journal,KoutisLP12},
potential function based~\cite{BatsonSS12,ZhuLO15,LeeS15,LeeS17}, and
spanners based, which use sampling via matrix
concentration~\cite{KapralovP12,Koutis14,KyngPPS17}. A survey of these
approaches can be found in~\cite{BatsonSST13}.

We present a framework for graph sparsification built on short cycle
decomposition that merges several ideas from the importance-sampling
and spanners based approaches.  Before giving an overview of the
results in our paper, we first present an alternative algorithm for
the classic graph sparsification result of Spielman and
Srivastava~\cite{SpielmanS08:journal}. This will be quite useful since
our algorithms for constructing degree-preserving sparsifiers and
sparsifying Eulerian directed graphs are immediately built on the
following algorithm, and degree-preserving sparsification is a key
idea underlying all our remaining results.

Say we have a graph $G(V, E, w)$ with $m$ edges and $n$ edges.  We
start by expressing $\LL_{G} = \sum_{e \in E} w_{e} \LL_{e},$ where
for edge $e = (u,v),$
$\LL_{e} = \chi_{uv}\chi_{uv}^{\top}.$ We can
re-write this as
\[ \PPi = \sum_{e} w_{e}\LL_{G}^{\nfrac{+}{2}}
  \LL_{e}\LL_{G}^{\nfrac{+}{2}},\] where $\PPi$ is the projection
orthogonal to the all ones vector. Given a subset of edges
$E^{\prime} \subseteq E,$ we draw a random graph $H$ as follows,
independently for every edge $e \in E^{\prime}$, we include it in $H$
with probability $\nfrac{1}{2}$ and weight $2w_{e}.$ Otherwise, we
delete the edge $e$. All edges $e \in E\setminus E^{\prime}$ are
included in $H$ with weight $w_{e}.$ Observe that the expectation of
$\LL_{H}$ is $\LL_{G}.$

It follows from standard concentration results for sums of matrix
random variables that if for each edge $e$ in $E^{\prime},$ the norm
$\norm{w_{e}\LL_{G}^{\nfrac{+}{2}} \LL_{e}\LL_{G}^{\nfrac{+}{2}}}$ is
bounded (up to constants) by $\frac{\eps^{2}}{\log n},$ then with high
probability, $\LL_{H} \approx_{\eps} \LL_{G}.$

Now, observe that
$\norm{w_{uv}\LL_{G}^{\nfrac{+}{2}} \LL_{uv}\LL_{G}^{\nfrac{+}{2}}} =
w_{e} \chi_{uv} \LL_{G}^{+} \chi_{uv} = w_{uv} \reff(u,v)$ (this is
defined as the leverage score of the edge $uv$). A simple trace
argument implies $\sum_{e} w_{e} \reff(e) = \frac{n-1}{m},$ and hence
at least half the edges satisfy $w_{e} \reff(e) \le \frac{2n}{m}.$
Letting these edges with low leverage score be the set of edges
$E^\prime$ we toss random coins for, we obtain that
$\LL_{H} \approx_{\sqrt{\frac{2n}{m} \log n}} \LL_{G}.$ Moreover, in
expectation, $H$ has at most $\frac{3}{4}m$ edges.

We can repeat the above sparsification 
roughly $O(\log n)$ 
times to go down to $O(n\eps^{-2}\log n)$ edges,
at each step sparsifying the graph just obtained.
By Fact~\ref{fact:spectral-error-composition}, the final
approximation error is given by the sum of the error at each
sparsification step. Since the number of edges is going down
geometrically, the error is increasing geometrically, and hence is
dominated by the error at the last step, yielding that the final graph is an
$O\left(\sqrt{\frac{2n \log n}{n\eps^{-2}\log n}}\right) = O(\eps)$
spectral-sparsifier for $G.$

In order to implement this algorithm efficiently, we need to estimate
effective resistances for the edges. For the above algorithm, constant
factor estimates of the effective resistances suffice (at the cost of
changing the constants involved). Spielman and
Srivastava~\cite{SpielmanS08:journal} showed that one can obtain
constant factor estimates for all the edges together in $\otil{m}$
time, resulting in a complete running time of $\otil{m}$ for the above
sparsification algorithm.

\subsection{Degree Preserving Spectral Sparsification}
Now, we adapt the above algorithm to leverage a short-cycle
decomposition of the graph. Short cycles permit us to sample
correlated edges in the graph while still keeping each random sample
small in spectral norm. We first use this approach to construct
degree-preserving spectral sparsifiers.

We first formally define a degree-preserving sparsifier.
\begin{definition}[Degree-Preserving Sparsifier]
	\label{def:deg-preserving-sparsifier}
	A graph $H(V, E')$ is said to be a degree-preserving
	$\eps$-sparsifier of $G(V, E)$ if
	\begin{tight_enumerate}
		\item \label{item:sparsifier}for every $x \in \rea^V,$ we have,
		$x^{\top}\lap_G x \approx_{\eps} x^{\top}\lap_H x ,$ and
		\item \label{item:deg-preserving} every vertex $u \in V$ has the
		same weighted degree in $G$ and $H,$ i.e.,
		$\sum_v w^{G}_{u,v} = \sum_v w^{H}_{u, v}.$
	\end{tight_enumerate}
\end{definition}

Given the above algorithm for usual graph sparsification, the main
obstacle is that at each sparsification step, the weighted degrees are
not preserved. This is where we require our key tool,
a short cycle decomposition, which we now formally define.
\begin{definition}
\label{dfn:CycleDecompose}
For an undirected unweighted graph $G(V,E),$ we say that
$\{C_1, C_2, \ldots\}$ is a an $(\mhat, L)$-short cycle
decomposition, if $C_1, C_2,\ldots$ are edge-disjoint cycles in $G,$
each $C_i$ is a cycle of length at most $L,$ and $\left|E \setminus
\bigcup_{i} C_i \right| \le \mhat.$  
\end{definition}
Assuming that we have an
efficient algorithm for constructing an $(\mhat, L)$-short cycle
decomposition of any given graph, we show the following theorem.

\begin{theorem}
	\label{thm:intro:deg-preserving-sparsifier}
	Given $\eps \in (0,1],$ every undirected graph $G$ with
        poly-bounded weights has a degree-preserving $\eps$-sparsifier
        with $O(n\eps^{-2} \log^{2} n)$ edges. The algorithm
        \DegPreserveSparsify, given our short cycle decomposition
        algorithm, takes in a graph $G$ and runs in time $m^{1+o(1)}$
        and returns a degree-preserving $\eps$-sparsifier of $G$ with
        $n^{1+o(1)} \eps^{-2}$ edges.
\end{theorem}

The following is a brief
description of our degree-preserving sparsification algorithm.

Assume first that our graph $G$ is an unweighted graph that has been
given to us as the union of disjoint cycles of even length. We sample
a random graph $H$ as follows. For each cycle independently, we index
the edges in order start from an arbitrary vertex, and perform the
following correlated sampling procedure: with probability
$\nfrac{1}{2}$, we keep only the even indexed edges with weight 2, and
with probability $\nfrac{1}{2},$ we keep only the odd indexed edges with
weight 2 (see Figure~\ref{fig:deg-sparsify}). Observe that $H$ has
half as many edges as $G,$ and has exactly the same weighted degrees
as $G$. In order to apply matrix concentration, we need to ensure that
for each cycle $C,$ the norm
$\norm{\LL_{G}^{\nfrac{+}{2}} \LL_{C}\LL_{G}^{\nfrac{+}{2}}}$ is at
most $\frac{\eps^{2}}{\log n}$, where $\LL_{C}$ is the Laplacian of
the cycle $C.$ This norm is easily upper bounded by
$\sum_{e \in C} \reff(e).$

If instead, $G$ was any arbitrary unweighted graph, we move all the
edges with $\reff \ge \frac{2n}{m}$ to $H.$ Again, by averaging, we
still have at least $\nfrac{m}{2}$ edges remaining. Now, we greedily
pick a bi-partition of the vertices of $G$ such that at least half the
remaining edges are crossing the cut. We add all the non-crossing
edges to $H.$ Now, we utilize an $(\mhat, L)$-short cycle
decomposition of $G.$ Thus, all but $\mhat$ edges of $G$ are
partitioned into cycles of length at most $L$. Observe that all the
cycle edges crossing the bi-partition, at least must now be in even
cycles, each with total $\reff$ bounded by $\frac{2nL}{m}.$ Now,
independently for each cycle, we pick even or odd edges with
probability $\nfrac{1}{2},$ and add them to $H$ with weight
2. Assuming $m \ge 8\mhat,$ $H$ has at most $\frac{15}{16}m$ edges,
the same weighted degree as $G,$ and with high probability
$H \approx_{\sqrt{\frac{2nL}{m}\log n}} G.$

Note that re-framing original sparsification into an algorithm for
reducing the edges by a constant fraction is crucial for this
algorithm. We are only able to reduce the edges in a cycle by
half. Further, the cycle decomposition of the graph will necessarily
change with every iteration.

For starting with a weighted graph with poly-bounded weights, we can
use the binary representation of the edge weight to split each edge
into $O(\log n)$ edges, each with a weight that's a power of 2. Now,
repeating the above procedure as before, we can construct a
degree-preserving $\eps$-sparsifier for $G$ with roughly
$\mhat \log n + n\eps^{-2}L\log n$ edges. Using the $O(\log n)$ length
short-cycle decomposition, this gives roughly $n\eps^{-2}\log^{2} n$
edges.

\tikzset{mynode/.style={inner sep=2pt,fill,outer sep=0,circle}}

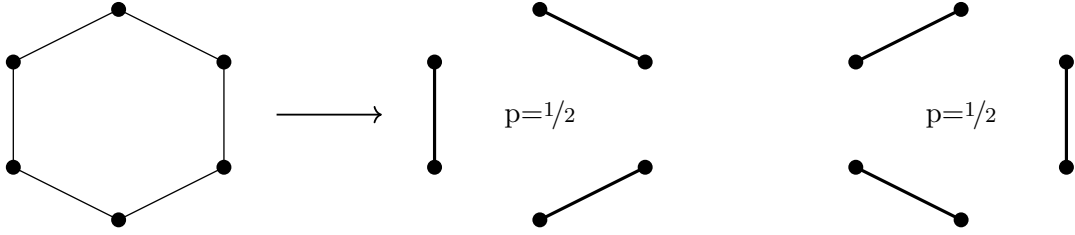
\begin{figure}
  \centering
  \begin{tikzpicture}[scale=0.7]
    \coordinate (1) at (2,0);
    \coordinate (2) at (4,1);
    \coordinate (3) at (4,3);
    \coordinate (4) at (2,4);
    \coordinate (5) at (0,3);
    \coordinate (6) at (0,1);

    %% vertices
    \foreach \x in {1,2,...,6}{
      \node[mynode] at (\x) {};
    }
    %%% edges    
    \draw[line width=0.5pt] (1) -- (2) -- (3) -- (4) -- (5) -- (6) -- cycle;
    \draw[->, line width = 0.7pt] (5,2) -- (7,2);

    \node at (10,2) {p=\nfrac{1}{2}};
    \coordinate (11) at (10,0);
    \coordinate (12) at (12,1);
    \coordinate (13) at (12,3);
    \coordinate (14) at (10,4);
    \coordinate (15) at (8,3);
    \coordinate (16) at (8,1);
    %% vertices
    \foreach \x in {11,12,...,16}{
      \node[mynode] at (\x) {};
    }
    \draw[line width=1.2pt] (11) -- (12) (13) -- (14)  (15) -- (16) ;

    \node at (18,2) {p=\nfrac{1}{2}};
    \coordinate (21) at (18,0);
    \coordinate (22) at (20,1);
    \coordinate (23) at (20,3);
    \coordinate (24) at (18,4);
    \coordinate (25) at (16,3);
    \coordinate (26) at (16,1);
    %% vertices
    \foreach \x in {21,22,...,26}{
      \node[mynode] at (\x) {};
    }
    \draw[line width=1.2pt] (22) -- (23) (24) -- (25)  (26) -- (21) ;
  \end{tikzpicture}
  \caption{Sampling alternate edges in short cycles for
    degree-preserving sparsification. The thick edges are double the
    weights of the thin edges.}
  \label{fig:deg-sparsify}
\end{figure}

\subsection{Sparsification of Eulerian Directed Graphs}
\label{sec:overview:eulerian}
Now, we can take a very similar approach to sparsifying Eulerian
directed graphs. This is a primitive introduced
in~\cite{CohenKPPSV16}, and is at the core of recent developments in
fast solvers for linear systems in directed
Laplacians~\cite{CohenKPPSV16, CohenKPPRSV17,Kyng17}.  In contrast to
undirected graphs, it has been significantly more challenging to give
an appropriate notion of approximation for directed graphs (see
Section~\ref{sec:eulerian} for the definition of Laplacian $\LL_\dir{G}$
of  a directed graph $G$). Cohen~\etal~\cite{CohenKPPRSV17} showed
that for the purpose of solving linear systems in Eulerian directed
graphs, one such useful notion is to say $\dir{H}$ $\eps$-approximates
$\dir{G}$ if
\[\norm{\LL_{G}^{\nfrac{+}{2}} (\LL_{\dir{H}} - \LL_{\dir{G}}) \LL_{G}^{\nfrac{+}{2}}} \le
  \eps,\] where $G$ is the undirectification of $G,$ \emph{i.e.}, the
underlying undirected graph of $\dir{G}$ with edge-weights halved.  In
the case where $\dir{G}$ is Eulerian,
$\LL_{G} = \frac{1}{2}(\LL_{\dir{G}} + \LL_{\dir{G}}^{\top}).$

The key obstacle in sparsifying Eulerian directed graphs is to sample
directed subgraphs $\dir{H}$ that are Eulerian since independent
sampling cannot provide us with such precise control on the
degrees. The work of Cohen~\etal~\cite{CohenKPPRSV17} fixed this
locally by modifying the diagonal in $\LL_{\dir{H}}$ in order to make
the sampled graph Eulerian. This approach induces an error in
$\LL_{\dir{H}}$ of the order of $\eps \DD_{\dir{G}}$ where $\DD$ is
the diagonal out-degree matrix for $\dir{G}.$ In order for this error
to be small relative to $\LL_{\dir{G}},$ $\dir{G}$ must be an
expander. Hence, the need of expander partitioning in their approach.

However, as we saw above, a short cycle decomposition allows us to
perform correlated sampling on edges with precise control on the
degrees. For sampling directed graphs, consider a single cycle where
the edges may have arbitrary direction (see
Figure~\ref{fig:eulerian-sparsify}). With probability $\nfrac{1}{2},$
we sample the edges in clockwise-direction, and with probability
$\nfrac{1}{2},$ we sample the edges in the anti-clockwise
direction. In either case, we double the weights of the sampled
edges. Observe that for each vertex, the difference between the
outgoing and incoming degrees is preserved exactly. Hence, if we
started with an Eulerian directed graph, we end up with an Eulerian
directed graph. Moreover, in expectation, we only keep half the edges
of the cycle.

We can now basically follow the algorithm for degree-preserving
sparsification. We treat the graph as undirected for all steps except
for sampling edges from a cycle. In particular, the cycle
decomposition is found in the corresponding undirected graph. Using
the above approach for sampling edges from each cycle, we can sample
an Eulerian directed graph, that has a constant fraction fewer edges
in expectation. Since the matrices involved are no longer symmetric,
we invoke concentration bounds for rectangular matrices to obtain
\[\norm{\LL_{G}^{\nfrac{+}{2}} (\LL_{\dir{H}} - \LL_{\dir{G}}) \LL_{G}^{\nfrac{+}{2}}} \le
  O\left(\sqrt{\frac{nL^{3}\log n}{m}}\right).\]

Now, repeating this sparsification procedure, and observing that this
notion of approximation error also composes, we obtain an Eulerian
directed graph $\dir{H}$ that $\eps$-approximates $\dir{G}$ with
roughly $\mhat\log n + n\eps^{-2}L^{3} \log n$ edges. Again, using the
naive cycle decomposition, this is $O(n\eps^{-2} \log^{4}n)$ edges.
\begin{figure}
  \centering
  \begin{tikzpicture}[scale=0.7]
    \tikzset{edge/.style = {-{Latex[length=2mm,width=2mm]}}}
    \coordinate (1) at (2,0);
    \coordinate (2) at (4,1);
    \coordinate (3) at (4,3);
    \coordinate (4) at (2,4);
    \coordinate (5) at (0,3);
    \coordinate (6) at (0,1);
    %% vertices
    \foreach \x in {1,2,...,6}{
      \node[mynode] at (\x) {};
    }
    %%% edges    
    \draw[edge] (1) to (2);
    \draw[edge] (2) to (3);
    \draw[edge] (4) to (5);
    \draw[edge] (4) to (3);
    \draw[edge] (6) to (1);
    \draw[edge] (6) to (5);
    \draw[->, line width = 0.7pt] (5,2) -- (7,2);

    \node at (10,2) {p=\nfrac{1}{2}};
    \coordinate (11) at (10,0);
    \coordinate (12) at (12,1);
    \coordinate (13) at (12,3);
    \coordinate (14) at (10,4);
    \coordinate (15) at (8,3);
    \coordinate (16) at (8,1);
    %% vertices
    \foreach \x in {11,12,...,16}{
      \node[mynode] at (\x) {};
    }
    \draw[edge, line width=1.2pt] (11) to (12);
    \draw[edge, line width=1.2pt] (12) to (13);
    \draw[edge, line width=1.2pt] (14) to (15);
    \draw[edge, line width=1.2pt] (16) to (11);

    \node at (18,2) {p=\nfrac{1}{2}};
    \coordinate (21) at (18,0);
    \coordinate (22) at (20,1);
    \coordinate (23) at (20,3);
    \coordinate (24) at (18,4);
    \coordinate (25) at (16,3);
    \coordinate (26) at (16,1);
    %% vertices
    \foreach \x in {21,22,...,26}{
      \node[mynode] at (\x) {};
    }
    \draw[edge, line width=1.2pt] (24) to (23);
    \draw[edge, line width=1.2pt] (26) to (25);

  \end{tikzpicture}
  \caption{Sampling edges along a random direction in short cycles for
    sparsification of Eulerian directed graphs. The thick edges are
    double the weights of the thin edges.}
  \label{fig:eulerian-sparsify}
\end{figure}
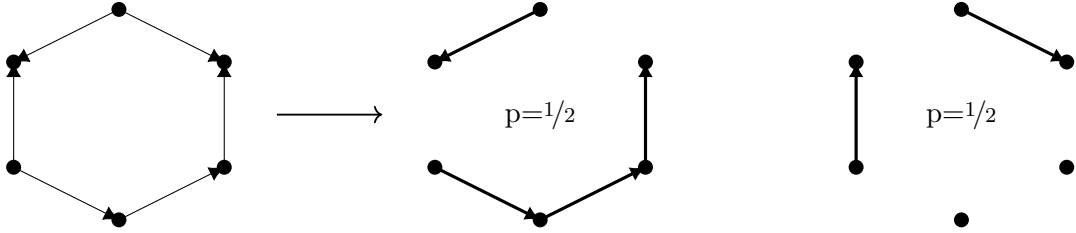

\begin{theorem}
	\label{thm:intro:euleriean-sparsify}
	Given $\eps \in (0,1],$ for every Eulerian directed graph
        $\dir{G},$ we can find in $\otil{mn}$ time an Eulerian
        directed graph $\dir{H}$ with $O(n\eps^{-2} \log^{4} n)$
        edges, that $\eps$-approximates $\dir{G}.$
 \end{theorem}
 
This shows the existence of sparsifiers for Eulerian graphs
with fewer edges than the nearly-linear sized ones constructed
in Cohen~\etal~\cite{CohenKPPRSV17}.
More importantly, it shows that approaches based on importance
sampling, which work well on undirected graphs, can work in the
more general directed settings as well.
However, the high costs of computing short cycle decompositions in this
paper means this does not lead to faster asymptotic running times
in the applications -- we believe this is an interesting direction for future work.

\subsection{Graphical Spectral Sketches and Resistance Sparsifiers}
\label{subset:overview:resistance-sparsify}
We define a graphical spectral sketch as follows:
\begin{definition}[Graphical Spectral Sketch]
  \label{def:sketch}
  Given a graph $G(V,E),$ a distribution $\calH$ over random graphs
  $H(V, E')$ is said to be a graphical $\eps$-spectral sketch for $G$,
  if for any fixed $\xx \in \rea^{V},$ with high probability, over the
  sample $H \sim \calH,$ we have
  $ \xx^{\top}\lap_G \xx \approx_{\eps} \xx^{\top}\lap_H \xx .$
\end{definition}

For constructing graphical spectral sketches, we closely follow the
approach of Jambulapati and Sidford~\cite{JambulapatiS18} and
Andoni~\etal~\cite{AndoniCKQWZ16}.
However, to construct sketches which are graphical,
we use an approach similar to the degree-preserving sparsification algorithm.
Our result is as follows:
\begin{theorem}
	\label{thm:overview:spectral-sketches}
	Given $\eps \in (0,1],$ every undirected graph $G$ with $n$
        vertices and $m$ edges has a graphical $\eps$-spectral sketch
        of $G$ with $\otil{n\eps^{-1}}$ edges. The algorithm
        \SpectralSketch, given $G,$ runs in time $m^{1+o(1)},$ and
        with high probability returns a graphical $\eps$-spectral
        sketch of $G$ with $n^{1+o(1)}\eps^{-1}$ edges. In addition,
        both these graphical sketches satisfy\footnote{$\lap^{+}$
          denotes the Moore-Penrose pseudo-inverse of $\lap.$} that
        for any fixed $\xx \in \rea^{n},$ with high probability over
        the sample $H \sim \calH,$ we have
        $ \xx^{\top}\lap^{+}_G \xx \approx_{\eps} \xx^{\top}\lap^{+}_H
        \xx.$
\end{theorem}

The key idea in \cite{JambulapatiS18} and \cite{AndoniCKQWZ16} is to focus on an
expander $G$, and for each vertex $u$ with degree
$d_u,$ sample $\eps^{-1}$ edges incident at $u$ and add them to $H$ after
scaling its weight by $d_u\eps$ (if $d_u \le \eps^{-1},$ we add all the
edges of $u$ to $H$), for a total of $n\eps^{-1}$ edges. Firstly,
observe that this means that we will have vertices where the degree
changes by $\sqrt{\eps^{-1}}.$ This is not good enough to preserve
$\xx^{\top} \LL_{G} \xx$ up to $(1\pm\eps)$ even for the vectors
$\xx \in \{\cchi_u\}_{u \in V}.$
They get around this by explicitly storing the
diagonal degree matrix $\DD_{G}$ of $G$ using $O(n)$ extra space. For
a fixed vector $\xx,$ they consider the estimator
$\xx^{\top}\DD_{G} \xx - \xx^{\top} \AA_{H} \xx.$ Its expectation is
easily seen to be $\xx^{\top} \LL_{G} \xx.$ They prove that its standard
deviation is bounded by $\eps \cdot O(\xx^{\top} \DD_{G} \xx).$ For an
expander with conductance $\phi,$ Cheeger's inequality (see
Lemma~\ref{lem:Cheeger}) gives that
$\eps \cdot \xx^{\top} \DD_{G} \xx = \eps\cdot O(\phi^{-2} \xx^{\top}
\LL_{G} \xx).$

In order to construct an estimator for general graphs, they invoke
expander partitioning~\cite{SpielmanT11:journal}, which guarantees
that in any graph $G,$ we can find disjoint vertex induced pieces $G,$
such that each piece is contained in an expander
(a well-connected subgraph, formally defined in Section~\ref{subsec:SpectralSketches}),
 and at least
half the edges are contained within such pieces. Applying this
$O(\log n)$ times recursively, and combining the above estimators for
each piece, Jambulapati and Sidford~\cite{JambulapatiS18} obtain an estimator with standard
deviation $\eps \cdot \otil{\xx^{\top} \LL_{G} \xx}.$ 

The above sketch is not a graph since sampling edges does not preserve
degrees exactly. Hence, our degree-preserving sparsification algorithm
presents a natural approach to convert it into a graphical sketch. We
aim to reduce the edge count by a constant factor without incurring
too much variance in the quadratic form (and then repeat this process
$O(\log n)$ times). We apply expander decomposition, and within each
piece, add all edges incident on vertices with degree at most
$\otil{\eps^{-1}L}$ to $H.$ On the remaining graph, as before,
we find a bi-partition, a cycle decomposition, and independently pick
odd/even edges in the cycles with double the weight. This reduces the
number of edges by a constant factor. Since we preserve weighted
degrees \emph{exactly}, an analysis similar to the above gives that
for a fixed vector $\xx$ the standard deviation in $\xx^{\top}\LL_{H}\xx$ is
bounded by $\eps\cdot\otil{\xx^{\top} \LL_{G} \xx}.$ Repeating this
process $O(\log n)$ times gives us a graph $H$ with
$\otil{\mhat + n\eps^{-1}L}$ edges. Averaging $\otil{1}$ such
sketches, and applying concentration bounds, we obtain a graphical
$\eps$-spectral sketch of $G.$

The fact that we have a graph allows us to reason about the quadratic
form of its inverse $\xx^{\top} \LL_{H}^{+} \xx$.
We first argue that $H$ is an $\sqrt{\eps}$-spectral sparsifier
of $G$ by showing that the probabilities that we sample edges
to form $H$ are good upper bounds of (appropriate rescalings of)
effective resistances.
This follows because any edge $e$ incident to vertices
with degrees at least $\eps^{-1}$
that are contained in an expander with expansion at least $\phi$
has effective resistance at most $O(\phi^{-2}\eps)$.

A simple, but somewhat surprising argument
(Lemma~\ref{lem:Invert}) gives that if $H$ is a graphical
$\eps$-spectral sketch, and a $\sqrt{\eps}$-spectral sparsifier,
then for any fixed vector $x,$ with high probability, it also preserves the inverse quadratic
form of $G$, \emph{i.e.},
$\xx^{\top} \LL_{H}^{+} \xx \approx_{O(\eps)} \xx^{\top} \LL_{G}^{+} \xx.$

Picking $\xx \in \{\chi_{uv}|u,v \in V\},$ and taking union bound, we
obtain that with high probability, for all $u,v \in V,$
$\reff^G(u,v) \approx_{O(\eps)} \reff^{H}(u,v).$ This means that $H$
is a resistance-sparsifier for $G$ with high probability. Again, the naive cycle
decomposition gives the existence of resistance-sparsifiers with
$\otil{n\eps^{-1}}$ edges.
\begin{corollary}[Resistance Sparsifiers]
	\label{cor:resistance-sparsifiers}
	For $\eps \in (0,1],$ every undirected graph $G$ on $n$
        vertices has a resistance sparsifier with $\otil{n\eps^{-1}}$
        edges. The algorithm \SpectralSketch, given $G,$ runs in time
        $m^{1+o(1)},$ and with high probability returns a resistance
        sparsifier of $G$ with $n^{1+o(1)}\eps^{-1}$ edges.
\end{corollary}

\subsection{Estimating Effective-Resistances}

The effective resistance of an edge is a fundamental quantity.
It and its extensions have a variety of connections in the analysis
of networks~\cite{SarkarM07,Sarkar10:thesis},
combinatorics~\cite{Lovasz93,DurfeeFGX18}
and the design of better graph
algorithms~\cite{ChristianoKMST10,MadryST15,Schild17:arxiv}.

While the effective resistance of an edge $uv$ can be computed
to high accuracy using linear system solvers, doing so for all edges
leads to a quadratic running time.
On the other hand, the many algorithmic applications of resistances
have motivated studies on efficient algorithms for estimating
all resistances.
There have been two main approaches for estimating effective
resistances to date: random projections~\cite{SpielmanS08:journal,KoutisLP12}
or recursive invocations of sparsified Gaussian
 elimination~\cite{DurfeeKPRS17}.
Both of them lead to running times of $\otil{m\eps^{-2}}$
for producing $1 \pm \eps$ estimates of the resistances
of all $m$ edges of a graph.

A recent result by Musco et al.~\cite{MuscoNSUW18} demonstrated
the unlikelihood of high accuracy algorithms
(with $\eps^{-c}$ dependency for some small $c$)
for estimating the resistances of all edges.
On the other hand, the running time of a determinant estimation
algorithm for Laplacians by Durfee et al.~\cite{DurfeePPR17}
hinges on this $\eps$ dependency.
The running time bottleneck of this algorithm is the estimation
of effective resistances of $\otil{n^{1.5}}$ edges,
but to an multiplicative error of $\eps = n^{-0.25}$.
Both methods for estimating resistances described
above~\cite{SpielmanS08:journal,DurfeeKPRS17} give running times
of $\otil{n^2}$ in this setting.
Practical studies involving the random projection method for 
estimating resistances~\cite{Sarkar10:thesis,MavroforakisGKT15}
also demonstrate the $\log{n} \eps^{-2}$ factor in the runtime of
such methods translates to solving $10^3$ linear systems for a $10\%$ error.
Such high overhead has been a major limitation in applying effective
resistances to analyzing networks.

A key advantage of our graph sketches and resistance sparsifiers
is that because the resulting objects remain as graphs, they
can be substituted into the intermediate states of the sparsified
Gaussian elimination approach for computing 
graph sparsifiers~\cite{DurfeeKPRS17}.
They give a reduction from computing effective resistances to
computing approximate Schur complements,
which are partial states of Gaussian elimination.
Incorporating our spectral sketches in place of generic graph
sparsification algorithms with $\eps^{-2}$ dependencies gives
our main algorithmic result.
\begin{theorem}
\label{thm:MainER}
Given any undirected graph $G$ with $m$ vertices, and $n$ edges,
and any $t$ vertex pairs and error $\eps > 0$, we can with high
probability compute $\eps$-approximations to the effective resistances
between all $t$ of these pairs in
$O(m^{1 + o(1)} + (n + t) n^{o(1)} \eps^{-1.5})$ time.
\end{theorem} 

This is the first routine for estimating effective resistances
on sparse graphs that obtain an $\eps$ dependence better
than $\eps^{-2}$.
In the dense case an $\otil{n^{2} \eps^{-1}}$ result was shown
by Jambulapati and Sidford~\cite{JambulapatiS18}, but it relies
on $\Omega(n)$ linear systems solves, one per column of the matrix.

We obtain this result via two key reductions:
\begin{tight_enumerate}
\item The recursive approximate Gaussian elimination approach
  from~\cite{DurfeeKPRS17} utilizes the fact that effective
  resistances are preserved under Gaussian eliminations.  As this
  recursion has depth $O(\log{n})$, our guarantees for
  $\eps$-spectral sketches imply that it suffices to work with
  sketches of such graphs produced by Gaussian elimination.  However,
  Schur complement of very sparse graphs such as the degree $n$ star
  may have $\Omega(n^{2})$ edges.  Even if we eliminate an independent
  set of size $\Theta(n),$ each with roughly average degrees in our
  spectral sketches with $n\eps^{-1}$ edges, we will end up with
  at least $n\eps^{-2}$ edges.  Thus, we need to directly compute
  spectral sketches of Schur complements without first constructing the
  dense graph explicitly.
\item The work of Kyng et al.~\cite{KyngLPSS16} builds fast solvers
  for Laplacian systems via approximate Cholesky factorization. As a
  key step, they reduce computing approximating Schur complements to
  implicitly sparsifying a sum of product weighted cliques \footnote{A
    product weighted clique has a weight vector $\ww$ with the $u,v$
    edge having weight $\ww_{u}\ww_{v}.$}. Assuming we start with a
  spectral-sketch, we know that the graph has total degree
  $ n^{o(1)}\eps^{-1},$ this implies that the total number of vertices
  involves in these product weighted cliques is $\eps^{-1}n.$
  Thus, our goal becomes designing an algorithm for implicitly building
  spectral sketches of product-weighted cliques with a total of
  $n^{1 + o(1)} \eps^{-1}$ vertices that run in time
  $n^{1 + o(1)} \eps^{-(2 - c)}$ for some constant $c > 0$.
\end{tight_enumerate}

Our algorithm works with these
weighted cliques in time dependent on their representation, which is
the total number of vertices, rather than the number of edges.  We do
so by working with \emph{bi-cliques} as the basic unit, instead of
edges.  Our algorithm then follows the expander-partitioning based scheme
for producing spectral sketches, as in previous
works on graph sketches with $\eps^{-1}$ type
dependencies~\cite{AndoniCKQWZ16,JambulapatiS18}.  This requires
showing that this representation as bi-cliques interacts well with
both weights and graph partitions.  Then on each well-connected piece,
we sample $ n^{o(1)} \eps^{-0.5}$ matchings from each
bi-clique.

This results in each vertex in
the bi-clique representation contributing 
$n^{o(1)} \eps^{-0.5}$ edges to the intermediate sketch.
As we are running such routines on the output of spectral sketches, the
total number of vertices in these cliques is  $n^{o(1)} \eps^{-1}$,
giving a total edge count of $n^{1+o(1)}\eps^{-1.5}.$
On this graph, we can  now explicitly compute another
spectral sketch of size $n^{1 + o(1)}\eps^{-1}$.

An additional complication is computing an expander decomposition
using Lemma~\ref{lem:ExpanderDecomposition} requires examining all the
edges of a graph, which in our case is cost-prohibitive.
We resolve this by computing these decompositions on a constant
error sparse approximation of this graph instead.

Incorporating this spectral sketch of Schur complements
back into~\cite{DurfeePPR17} gives the first
sub-quadratic time algorithm for estimating the determinants of a
graph Laplacian with the last row and column removed.
This value has a variety of natural interpretations including the number
of random spanning trees in the graph. 
Note that while the determinant may be exponentially large, the
result in~\cite{DurfeePPR17} is stable with variable-precision
floating point numbers.
\begin{corollary}
\label{cor:ITSOVER90PAGES}
Given any graph Laplacian $\LL$ on $n$ vertices and $m$ edges,
and any error $0 < \eps < 1/2$, we can produce an $1 \pm \eps$
estimate to $\det(\LL_{-n})$, the determinant of $\LL$ with
the last row/column removed,
in time
$m^{1 + o(1)} + n^{\nfrac{15}{8} + o(1)} \eps^{-\nfrac{7}{4}}$.
\end{corollary}
Note that the removal of the last row / column is necessary and
standard due to $\LL$ being low rank.
Details on this algorithm, and the specific connections with~\cite{DurfeePPR17}
are in Appendix~\ref{sec:Determinant}.
We remark that this algorithm
however does not speed up the random spanning tree generation portion
of~\cite{DurfeePPR17} due to it relying on finer variance bounds that
require sampling $\otil{n^2}$ edges.  That spanning tree sampling
algorithm however, is superseded by the recent breakthrough result by
Schild~\cite{Schild17:arxiv}.

\subsection{Almost-Linear Time Short Cycle Decomposition}
The bottleneck in the performances of all algorithms outlined above is
the computation of short cycle decompositions
(Definition~\ref{def:ShortCycleDecomposition}). The simple existence
proof from Section~\ref{subsec:contributions} can be implemented to
find a short cycle decomposition in $O(mn)$ time (see
Section~\ref{sec:construction} for pseudo-code and proof).
\begin{theorem}
	\label{thm:NaiveCycleDecomposition}
	The algorithm \NaiveCycleDecomposition,
	given an undirected	unweighted graph $G$,
	returns a
	$(2n, 2\log n)$-short cycle decomposition of $G$ in $O(mn)$ time.
\end{theorem}

While the above algorithm gives us near-optimal 
length and number of remaining edges\footnote{Consider the wheel graph with $\frac{n-1}{\log n}$
	spokes, and replace each spoke with a path of length $\log n.$ This
	graph has $n$ vertices, $(n-1)\left(1+\frac{1}{\log n}\right)$
	edges, and girth of $2 \log n + 1.$ Lubotzky, Philip, and
	Sarnak~\cite{LubotzkyPS88} constructed explicit Ramanujan graphs
	that are 4-regular (and hence have 2n edges) and girth
	$\frac{4}{3}\log_3 n \ge 0.84 \log_2 n.$},
we were unable to obtain an
almost-linear time algorithm using shortest-path trees. The main
obstacle is that updating shortest-path trees is expensive under edge
deletions.

\medskip
\noindent \textbf{Possible Approaches via Spanners.} Another approach is to try spanners. The existence of a short cycle
decomposition is a direct consequence of spanners. A key result by
Awerbuch and Peleg~\cite{AwerbuchP90} for spanners states that every
unweighted graph $G$ has a subgraph $H$ with $O(n)$ edges such that
for every edge in $H \setminus G,$ its end points are within distance
$O(\log{n})$ in $H$. Thus, every edge in $H \setminus G$ is in a cycle
of length $O(\log n).$ We can remove this cycle and repeat.

Thus, another approach for generating this decomposition is by
dynamic, or even decremental, spanners~\cite{BaswanaS08, BernsteinR11,
  BaswanaKS12}. While these data structures allow for $\poly(\log n)$
time per update, they are randomized, and crucially, only work against
oblivious adversaries.  Thus, the update sequence needs to fixed
before the data structure samples its randomness. To the best of our
understanding, in each of these result, the choice of cycle edges
depends upon the randomness. Thus, their guarantees cannot be used for
constructing short cycle decompositions.  The only deterministic
dynamic spanner algorithm we're aware of is the work of Bodwin and
Krinninger~\cite{BodwinK16}.  However, it has overheads of at least
$\sqrt{n}$ in the spanner size / running time.

\medskip
\noindent \textbf{Possible Approaches via Oblivious Routings.}
Another possible approach of finding short cycles is via oblivious
routings: a routing of an edge $e$ (that doesn't use $e$) immediately
gives a cycle involving $e.$ Since there exist oblivious routings for
routing all edges of $G$ in $G$ with small congestion, the average
length of a cycle cannot be too large.

Recent works, especially those related to approximate maximum flow,
have given several efficient constructions of oblivious routing
schemes~\cite{Racke08, RackeST14, Madry10, Sherman13, KelnerLOS14,
  Peng16}.  However, such routings only allow us to route a single
commodity in nearly-linear time. Using current techniques, routing
$\Omega(n)$ arbitrary demands on an expander with $\poly(\log n)$
congestion seems to requires $n^{1.5}$ time.
On the other hand, on more limited topologies,
it is known how to route each demand in sub-linear time~\cite{Valiant82}.
Such a requirement of only using local information to route
have been studied as myopic routing~\cite{GaoSY17}, but we are not
aware of such results with provable guarantees.

\medskip
\noindent \textbf{Our Construction.}
As an important technical part of this paper, we
give an almost-linear-time algorithm for constructing a short cycle
decomposition of a graph.
\begin{theorem}
	\label{thm:ShortCycleDecomposition}
	The algorithm \ShortCycleAlgo, given an undirected unweighted graph
	$G$ with $n$ vertices and $m$ edges, returns a
	$(n^{1+o(1)}, n^{o(1)})$-short cycle decomposition of $G$ in $m^{1+o(1)}$ time.
\end{theorem}

Our construction of short cycle decomposition is inspired by oblivious
routings, and uses the properties of random walks on expanders.  This
can be viewed as extending previous works that utilize behaviors of
electrical flows~\cite{KelnerM11,KelnerLOS14}, but we take advantage
of the much more local nature of random walks. This use of random walks
to port graphs to fewer vertices is in part motivated by their use in
the construction of data structures for dynamically maintaining
effective resistances, involving a subset of the
authors~\cite{DurfeeGGP18:arxiv}.  It also has similarities with the
leader election algorithm for connectivity on well-connected graphs in
a recent independent work by Assadi\etal~\cite{AssadiSW18:arxiv}.

Say we have an expander graph $G$ with conductance $\phi.$ We know
random walks of length $\phi^{-2} \log n$ mix well in $G.$ Choose a
parameter $k$ say $n^{\nfrac{1}{10}},$ and pick the set $S$ of $n/k$
vertices of largest degree (with total degree at least
$\nfrac{2m}{k}$). For every edge $e$ leaving $S,$ starting from its
other end point $u \notin S,$ we take a $O(\phi^{-2}\log n)$ step
random walk. This random walk hits $S$ again with probability
$\Omega(\nfrac{1}{k}).$ Thus, if we pick $k\log n$ random walks, at least
one of them will hit $S$ again with high probability. This is a short
cycle in $G_{\slash S}$ ($G$ with $S$ contracted to a single
vertex). Since these are independent random walks, Chernoff bounds
imply that the maximum congestion is $\otil{k\phi^{-2}}.$ Thus, we can
greedily pick a set of $\widetilde{\Omega}(m\phi^{4}k^{-1})$ cycles of
length $\otil{\phi^{-2}}$ in $G_{\slash S}$ that are disjoint.

Now, we just need to connect these cycles within $S.$ We define a
graph on the vertices of $S,$ with one edge for every cycle in
$G_{\slash S}$ connecting the two end points in $S,$ and recurse on
$S.$ With $10$ levels of recursion (since $k = n^{\nfrac{1}{10}}$),
and using the naive cycle-decomposition for the base case, we find a
short cycle decomposition in this graph, and then can expand it to a
cycle decomposition in $G$ using the cycles in $G_{\slash S}.$ This
should give cycles of length $(\phi^{-2} \log n)^{10}.$

There is a key obstacle here: this approach really needs expanders,
not pieces contained in expanders, as in the expander decomposition
from Spielman and Teng~\cite{SpielmanT11}. Instead, we use a recent result of
Nanongkai and Saranurak~\cite{NanongkaiS17} that guarantees the pieces
are expanders, at the cost of achieving $\phi = n^{-o(1)},$ and a
running time of $m^{1+o(1)}.$ A careful trade-off of parameters allows
us to recurse for $n^{o(1)}$ iterations, resulting in an
$(n^{1+o(1)}, n^{o(1)})$-short cycle decomposition in $m^{1+o(1)}$
time.

\subsection*{Cycle Decomposition algorithm in the following sections}
In the following sections, we assume $\CycleDecomposition$ is an
algorithm that takes as input an unweighted graph with $n$ vertices
and $m$ edges, runs in time at most $\TCycle(m,n) \ge m,$ and returns a
$(\mhat(n),L(n))$-short cycle decomposition. Further, we assume that
$T$ satisfies
\begin{align}
  \label{assumption:super-additive-T}
  \sum_{i} \TCycle(m_i, n) \leq \TCycle\left(\sum_i m_i , n \right)
\end{align}
for all $m_i \geq n$. We also assume $\TCycle(m,n) \le
\TCycle(m^\prime, n^\prime),$ for any $m \le m^{\prime}, n \le n^{\prime}.$ 
Since $n$ will remain the same throughout these sections, we will simply
write $\mhat$ and $L$ instead of $\mhat(n)$ and $L(n).$

%%% Local Variables:
%%% mode: latex
%%% TeX-master: "main"
%%% End:

%% file: deg-preserving.tex
%!TEX root = main.tex

\section{Degree-Preserving Spectral Sparsification}
\label{sec:Degree-Preserving}
In this section, we describe an efficient algorithm for constructing
degree-preserving spectral sparsifiers, proving
Theorem~\ref{thm:intro:deg-preserving-sparsifier}. 

The algorithm will use a short cycle decomposition, and sparsify each cycle $C$
with the distribution
\begin{align}
  \widetilde{C}
:=
2
\cdot
\begin{cases}
\text{all odd edges of $C$}
&
\qquad \text{w.p. $\nfrac{1}{2}$}
\\
\text{all even edges of $C$}
&
\qquad \text{w.p. $\nfrac{1}{2}$}.
\end{cases}
\label{eq:SparsifyCycle}
\end{align}
We will bound the error in this distribution via matrix Chernoff bounds
\cite{Tropp12}, and recursively apply this sparsification procedure until our
graph achieves low edge count.

The following theorem is the main result of this section.
\begin{theorem}
  \label{thm:DegreePreservingSparsify}
  Given a graph $G$ with integer
  poly bounded edge weights, an error parameter $\eps,$ and a cycle
  decomposition routine $\CycleDecomposition$, the algorithm
  ${\DegPreserveSparsify}$ (described in
  Algorithm~\ref{alg:DegreePreservingSparsify}) returns a graph $H$
  with at most $O\left(\mhat \log{n} + n L \eps^{-2} \log n \right)$
  edges such that all vertices have the same weighted degrees in $G$
  and $H$, and with high probability,
  $\LL_{G} \approx_{\eps} \LL_{H}.$

The algorithm ${\DegPreserveSparsify}$ runs in time
\begin{align}
O\left(m \log^2{n} \right) + \TCycle\left(O\left(m \log n\right), n\right)
\label{eq:TimeBound}
\end{align}
\end{theorem}

%%% DEGREE PRESERVING SPARSIFY ALGORITHM %%%

\begin{algorithm}

\caption{${\DegPreserveSparsify}(G, \eps, \CycleDecomposition)$}
\textbf{Input}: Graph $G$ with poly bounded	edge weights.

\begin{tight_enumerate}
\item Decompose each edge of $G$ by its binary representation.
Now edge weights of $G$ are powers of $2$, and are at most $m \log n$ in number.
\item Compute $\rr$, a $1.5$-approximate estimate of effective resistances in $G$.
\item While 
$|E(G)| \geq \Omega(\mhat \log n + n L  \epsilon^{-2}\log n) $:
\label{alg:step:Loop}

\begin{tight_enumerate}
\item $G \leftarrow \SparsifyOnce(G, \rr, \CycleDecomposition)$.
\end{tight_enumerate}

\item Return $G$.

\end{tight_enumerate}

\label{alg:DegreePreservingSparsify}
\end{algorithm}

We first prove Theorem~\ref{thm:intro:deg-preserving-sparsifier} by
plugging in {\NaiveCycleDecomposition} and {\ShortCycleAlgo} into
Theorem~\ref{thm:DegreePreservingSparsify} and evaluating ${T}$ on
those routines. It is easy to check that their runtimes satisfy
assumption~\ref{assumption:super-additive-T}.
\begin{proof}[Proof of
  Theorem~\ref{thm:intro:deg-preserving-sparsifier}.]
  Note that ${\DegPreserveSparsify}$ always returns a graph $H$ with
  the same weighted degrees as $G$, such that
  $\LL_H \approx_{\eps} \LL_{G}$ with high probability. Using either
  {\NaiveCycleDecomposition} or {\ShortCycleAlgo} as the algorithm
  {\CycleDecomposition}, we obtain the following guarantees:
  \begin{enumerate}
  \item Using {\NaiveCycleDecomposition}: ${\DegPreserveSparsify}$
    runs in $O(mn \log n)$ time, and returns an $H$ with
    $O\left(n \eps^{-2} \log^2 n\right)$ edges.
  \item Using {\ShortCycleAlgo}: ${\DegPreserveSparsify}$ runs in
    $m^{1+o(1)}$ time, and returns an $H$ with
    $n^{1+o(1)} \eps^{-2}$  edges.
  \end{enumerate}
  Thus we have our theorem.
\end{proof}

In order to prove Theorem~\ref{thm:DegreePreservingSparsify},
we first prove the following lemma about the effect of sampling the
cycles independently.
It is a direct consequence of matrix concentration bounds.
\begin{lemma}
  \label{lem:GraphSampling}
Let $\mathcal{G}_1 \ldots \mathcal{G}_{t}$ be independent distributions
over graphs containing at most $L$ edges, and let their expectations be
\[
  G_{i} \defeq  \expec{\widetilde{G}_{i} \sim \calG_{i}}{\widetilde{G}_i},
\]
and define their sum to be $G \defeq \sum_{1 \leq i \leq t} G_i$.
For any graph $H$ with
\[
\LL_{G} \preceq \LL_{G}
\]
such that the maximum leverage score of any edge with respect to $H$ in
bounded above by $\rho$, the random graph
\[
  \widetilde{G} = \widetilde{G}_1 + \ldots + \widetilde{G}_{t}
\]
with $\widetilde{G}_i \sim \mathcal{G}_i$ satisfies
with high probability
\[
\LL_{H}
\approx_{O\left(\sqrt{L\rho \log n}\right)}
\LL_{\widetilde{G} + H - G}.
\]
\end{lemma}
%% PROOF OF GRAPH SAMPING LEMMA.

\begin{proof}
This is a corollary of Matrix Chernoff bounds from \cite{Tropp12},
which state that
  for a sequence of independent random $d \times d$ PSD matrices
  $\{\XX_i\}$ such that $\expec{}{\sum_i \XX_i} \preceq \TT$,
we have for $\delta \leq 1,$
\[
\prob{}{\expec{}{\sum \XX_i} - \sum \XX_i \preceq \delta \TT} \geq 1 - d \cdot
e^{\frac{-\delta^2}{3R}},
\]
and
\[
\prob{}{\expec{}{\sum \XX_i} - \sum \XX_i \succeq -\delta \TT} \leq 1-d \cdot
e^{\frac{-\delta^2}{3R}}.
\]
where $R$ is such that for each $i,$ we have
$R \le \norm{\TT^{\nfrac{+}{2}} \XX_i \TT^{\nfrac{+}{2}}}$ almost surely.

To prove Lemma~\ref{lem:GraphSampling}, we set
  $\XX_i = \LL_{\widetilde{G}_i}$, $\SS = \LL_{\widetilde{G}}$, and
  $\TT =
\LL_{H}$. Then,
\[ \norm{\TT^{\nfrac{+}{2}}\XX_i \TT^{\nfrac{+}{2}}}_2 \leq
\text{Tr}\left(\TT^{\nfrac{+}{2}}\XX_i \TT^{\nfrac{+}{2}}\right) = \sum_{e \in
  \widetilde{G}_i} \text{Tr} \left( \LL_H^{\nfrac{+}{2}}\LL_e
\LL_H^{\dag/2} \right) \leq \rho L,
\]
where the last inequality
follows since the number of edges in $\widetilde{G}_i$ is at most
$L$, and $\text{Tr}\left( \LL_H^{\nfrac{+}{2}} \LL_e \LL_H^{\nfrac{+}{2}} \right)$
is the definition of the leverage score of any edge $e$ in the
support of $\mathcal{G}$ w.r.t. $H$. Note that our edge leverage scores
are bounded above by $\rho$.
Now, we
can set $\delta = O\left(\sqrt{\rho L \log n}\right)$ to get
\[
\prob{}{\expec{}{\sum \LL_{\widetilde{G}_i}} - \LL_G \preceq O(
  \sqrt{L \rho \log n
		})\LL_H} \geq
1 - \frac{1}{n^{O(1)}}
\]
or equivalently,
\[
\prob{}{\LL_{\widetilde{G} + H - G}
\preceq
\left(1+O(\sqrt{L\rho \log n})\right)\LL_H}
\geq 1 - \frac{1}{n^{O(1)}}
\]
Similarly, we bound the other direction, and by the union bound, we
obtain with high probability:
\[
\exp\left(-O(\sqrt{L \rho \log n})\right) \LL_H
\preceq
\LL_{\tilde{G} + H - G}
\preceq
\exp\left(O(\sqrt{L \rho \log n})\right) \LL_H.
\]
\end{proof}

%%% SPARSIFY ONCE ALGORITHM

\begin{algorithm}
  \caption{$\SparsifyOnce(G, \rr_e, \CycleDecomposition)$}
  \textbf{Input:} A multi-graph $G$,
  with edge weights that are poly bounded
  and powers of $2$,  with no two edges of the same weight spanning the same two
  vertices,
  
  A $2$-approximate estimate $\rr_e$ of the effective resistances of
  $G$,
  
  A cycle decomposition algorithm \CycleDecomposition.

  \textbf{Output:} A degree-preserving sparsifier $H$.

  \begin{tight_enumerate}
    \item Let $E_{highER}$ denote all the edges in $G$ with:
      $ \ww_e \rr_e \geq \frac{4n}{m}.$
    \label{alg:step:FindHighER}
    \item $H \leftarrow E_{highER}$ and $G \leftarrow G \setminus E_{highER}$.
      \label{alg:step:KeepHighER}
    \item Greedily find a bi-partition of $G,$
      say $B,$ containing at least half the edges in
      $G$.
      \label{alg:step:ComputeBipartition} 
    \item $H \leftarrow (G \setminus B)$: add all edges not in the
      bipartition to $H$.
      \label{alg:step:Bipartition}
    \item Partition the edges of $B$ into unit weight graphs times a
      power of $2$, and denote the graph with edge-weights $2^i$ as
      $G_i$. The output of this step will be $G_1, \ldots G_s$ where
      $s \leq O(\log n)$.
      \label{alg:step:UnweightedDecompose}
    \item For each $G_i$:
      \begin{tight_enumerate}
        \item $\{C_{i1}, C_{i2}, \ldots C_{it}\}
          \leftarrow \CycleDecomposition(G_{i})$,
          where $G_i$ is treated as an unweighted graph because all edges
          in it have the same edge weights.
          \label{alg:step:CycleDecompose}
        \item $G \leftarrow G \setminus \left(\cup_{j=1}^t C_{ij} \right)$.
          \label{alg:step:AddNonCycles}
      \item For each cycle $C_{ij}$:
        \begin{tight_enumerate}
          \item With probability $\nfrac{1}{2}$, add all the odd indexed edges
          of $C_{ij}$ into $H$ with weights doubled,
          otherwise add all the even indexed edges into $H$ with weights doubled.
          \label{alg:step:AddCycles}
        \end{tight_enumerate}
      \end{tight_enumerate}
    \item For every pair of vertices in $H$ with at least one edge between
      them, iteratively combine edges of the same weight until no two edges
      between that pair of vertices has the same weight.
    \item Return $H$.
  \end{tight_enumerate}
  \label{alg:SparsifyOnce}
\end{algorithm}

%% END OF ALGORITHMS.

The correctness of algorithm {\DegPreserveSparsify} hinges
on the following lemma about \SparsifyOnce.
%
%% LEMMA: SPARSIFY ONCE QUALITY AND TIME
\begin{lemma}
  \label{lem:SparsifyOnce}
  Given a graph $G$ with edge weights that are polynomially bounded
  and powers of $2$, a 2-approximate estimate $\rr$ for effective
  resistances in $G$, and a cycle decomposition routine
  \CycleDecomposition, the algorithm \SparsifyOnce returns
  a graph $H$ on the same vertex set that preserves all weighted
  vertex degrees, and with high probability, \[\LL_H \approx_{O\left(\sqrt{ \frac{n L \log n}{m} } \right)} \LL_G.\]
  Further, when $m \geq \Omega (\mhat \log n)$, $H$ has at most
  $(\nfrac{15}{16})m$ edges with high probability.
\end{lemma}
\begin{proof}
First, we show that \SparsifyOnce returns a graph that is a spectral
sparsifier of $G$ with the parameters specified above.
The difference between $G$ and $H$ arises from sampling the cycles
 $\{C_{ij}\}_{ij}$ produced by \CycleDecomposition. Note that each
 cycle $C_{ij}$ must be even-length since its a cycle in the bipartite graph $G_{i}.$ 

For an even length cycle $C_{ij}$, recall the algorithm samples either
twice the odd indexed edges, or twice the even indexed edges of $C_{ij}$,
each with probability $\nfrac{1}{2}$. Let distribution $\widetilde{C}_{ij}$ denote
the distribution that returns twice the even edges of $C_{ij}$ with
$\nfrac{1}{2}$ probability and twice the odd edges with $\nfrac{1}{2}$
probability.  Also, note that all edges in $C$ have leverage score at
most $8n/m$.

Note that $H$ is obtained by removing all cycles $C_{ij}$ from $G$,
and then adding the sampled odd/even edges $\widetilde{C}_{ij}$ with
twice the weights. Thus,
\[
  H = \left(G \setminus \sum_{ij} C_{ij} \right) + \sum_{ij} \widetilde{C}_{ij},
\]
Observing that each cycle $C_{ij}$ has at most $L$ edges, along with
$ \expec{}{\widetilde{C}_{ij}} = C_{ij},$ and we can apply the
concentration bound from Lemma~\ref{lem:GraphSampling}, to obtain
\[
\LL_H \approx_{O\left( \sqrt{ \frac{n L \log n}{m} } \right)} \LL_G.
\]

Next, we show that $H$ has at most $(\nfrac{15}{16})m$ edges. The
edges removed are:
\begin{tight_enumerate}
 	\item $\leq m/2$ edges with low leverage score.
 	
 	Since $\rr$ is a 2-approximate estimate of effective resistance,
 	by Foster's theorem, the average value of $\ww_e\rr_e$ of all edges is at most $2(n - 1)/m$.
 	So at most half of the edges in $G$ have $\ww_e\rr_e\ge \frac{4n}{m}$.
 	
 	\item $\leq m/4$ edges not in the bipartition.
 	
 	At most half the high ER edges are not in the bipartition.
 	
 	\item $\leq m/8$ edges not in cycles.
 	
 	At most $O(\mhat\log n)$ edges are not in the cycles since $i \leq O(\log n)$.
 	If $m=\Omega(\mhat\log n)$, this number does not exceed $m/8$.
 	This leaves at least $m/8$ edges in cycles.
 	
 	\item We add exactly half the cycle edges to $H$. 
\end{tight_enumerate}

This completes our proof.
\end{proof}

\begin{lemma}
\label{lem:SparsifyOnceRunTime}
  Given a multi-graph $G$ with $m'$ edges satisfying the input assumptions of
  \SparsifyOnce,
a vector $\rr$ that are $2$-approximations to the effective resistances
of $G$, and a cycle decomposition routine \CycleDecomposition,
the algorithm \SparsifyOnce runs in time at most
\[
  O\left(m'\right) + \TCycle\left(m', n\right).
\]
\end{lemma}
\begin{proof}
  It takes $O(m')$ time to greedily find a bipartition in
  Step~\ref{alg:step:ComputeBipartition}. For the second term, we
  analyze the time for the cycle decomposition steps.  The total time
  taken by the cycle decomposition steps is at most
  \[
  \sum_{i} \TCycle\left(\abs{E\left(G_i\right)}, n\right).
  \]
  Here $G_i$ are the graphs from
  step~\ref{alg:step:UnweightedDecompose} in
  Algorithm~\ref{alg:SparsifyOnce}. 
  By our assumption on $\TCycle$ in
  Equation~\ref{assumption:super-additive-T}, this term is bounded
  above by $\TCycle(\sum_i |E(G_i)|, n) = \TCycle(m', n).$
\end{proof}

%% LEMMA: GRAPH SAMPLING GIVES THE ORIGINAL GRAPH.

%% PROOF OF THEOREM: DEGREE PRESERVING SPARSIFY.
Finally, we are ready to prove Theorem~\ref{thm:DegreePreservingSparsify}.
\begin{proof}[Proof of Theorem~\ref{thm:DegreePreservingSparsify}]
  First observe that the {\DegPreserveSparsify} must return a
  graph with $O(\mhat \log n + n L \eps^{-2} \log n)$ edges.

  Lemma~\ref{lem:SparsifyOnce} tells us that the weighted degree of
  each vertex is preserved throughout the algorithm.  Therefore, the
  graph resulting from
  ${\DegPreserveSparsify}(G, \eps, \CycleDecomposition)$ is
  guaranteed to have the same weighted degrees as the original graph.

  Let $H_i$ be the value of $G$ after the $i^{th}$ iteration of the
  loop in {\DegPreserveSparsify} (Line~\ref{alg:step:Loop}), and let $m_i$ be the number of edges in $H_i$.
  Lemma~\ref{lem:SparsifyOnce} tells us that with high probability,
  $H_{i+1}$ is a
  $O\left( \sqrt{\frac{n \log n}{m_i}} \right)$-sparsifier of
  $H_{i},$ and $m_{i+1} \le \frac{15}{16}m_i.$  Suppose the loop runs for $t$ iterations before it
  terminates. Then $H_t$ is the final output of
  ${\DegPreserveSparsify}$.

By construction, the error is bounded by

\[
\sum_{i=0}^{t-1} O\left(\sqrt{\frac{n\log{n}}{m_i}}\right).
\]
Since $m_i$ decreases geometrically, this can be bounded in terms of the largest term
in the product, which is:
\[
O\left(\sqrt{\frac{n L \log{n}}{m_{t-1}}}\right)
\leq \eps.
\]
The last inequality holds because
$m_{t-1} \geq \Omega(nL\eps^{-2}\log{n})$.  Therefore,
${\DegPreserveSparsify}$ returns a low-error sparsifier
with high probability.

Now we analyze the time bound. Note that to compute $1.5$-approximate
effective resistances, it takes $O(m \log^2 n)$ time
\cite{SpielmanS08:journal, KoutisLP12}.  We perform this computation
just once, since our cumulative error is bounded, and a
$2$-approximation to effective resistances suffices for us.

During each call to \SparsifyOnce in
Algorithm~\ref{alg:DegreePreservingSparsify}, the number of
edges in the input is $m_i$, and the runtime of a single loop is bounded
above by $\TCycle(m_i,n)$ due to
Lemma~\ref{lem:SparsifyOnceRunTime}. Here, $m_i$ once again refers to
$|E(H_i)|$.  Therefore, the total runtime
is at most
\[
O\left(m \log^2 n\right) + \sum_{i=0}^t \TCycle\left(m_i,n\right).
\]

Note that $m_0 \leq m \log n$.
Since $m_i$ decreases more than geometrically
due to Lemma~\ref{lem:SparsifyOnce},
we can use Lemma~\ref{lem:SparsifyOnceRunTime}
and Equation~\ref{assumption:super-additive-T}
to bound the second term in the sum above by
\[
O\left(m \log^2 n\right) +
\TCycle\left(O\left(m \log n\right), n \right),
\]
giving us our claimed runtime.
\end{proof}

%%% Local Variables:
%%% mode: latex
%%% TeX-master: "main"
%%% End:

%% file: eulerian.tex
%!TEX root = main.tex

\section{Sparsification of Eulerian Directed Graphs}
\label{sec:eulerian}
In this section, we show how we can use short cycle decompositions to
sparsify Eulerian directed graphs.
For a directed graph $\dir{G}$, its directed Laplacian,
$\LL_{\dir{G}}$, can be defined as
\[
\LL_{\dir{G}}(u,v)
:=
\begin{cases}
\text{out-degree of $u$}
&
\qquad
\text{if $u = v$,}
\\
-\text{(weight of edge $v \rightarrow u$)}
&
\qquad
\text{if $u \neq v$ and $v \rightarrow u$ is an edge.}
\end{cases}
\]
Given a directed graph $\dir{G}$,
its ``undirectification", $G$,
is defined as the graph formed by replacing every
edge $\dir{e}$ in $\dir{G}$
with an undirected edge $e$ of half the weight, with the same endpoints.

A particularly important class of directed Laplacians is the set of
Laplacian matrices that correspond to Eulerian directed graphs, that
is, directed graphs in which the in-degree of each vertex is the same
as its out-degree.  This case is shown to be complete for solving
directed linear systems~\cite{CohenKPPSV16}.  Furthermore, 
 $\dir{G}$ being Eulerian implies that its Laplacian is directly
related to the Laplacian of its undirectification,
\begin{align}
  \LL_{G}
  =
  \frac{1}{2}
  \left(
  \LL_{\dir{G}}
  +
  \LL_{\dir{G}}^{\top}
  \right).
\label{eq:Undirectification}
\end{align}
$\LL_{G}$ is a graph Laplacian, and hence positive-semidefinite.
Our main result is an efficient algorithm for sparsifying
$\LL_{\dir{G}}$.  Because $\LL_\dir{G}$ may be asymmetric, we will
bound the deviation of our sparsifier w.r.t. the norm defined by
$\LL_{G}$.  The pseudocode is given in
Algorithm~\ref{alg:EulerianSparsifyFull}.

The sparsification algorithm will once again work with cycles.
Given a cycle $\dir{C}$ with directed edges,
we sparsify it with the distribution
\begin{align}
\widetilde{C}
:=
2
\cdot
\begin{cases}
\text{all clockwise edges of $\dir{C}$}
&
\qquad \text{w.p. $\nfrac{1}{2}$}
\\
\text{all counterclockwise edges of $\dir{C}$}
&
\qquad \text{w.p. $\nfrac{1}{2}$}
\end{cases}
\label{eq:SparsifyCycleDirected}
\end{align}
Note, that in the case of directed graphs, we do not need our cycles
to have even length, and hence can get rid of the bipartition step
required for degree-preserving sparsifiers
(Section~\ref{sec:Degree-Preserving}).
We will bound the error caused by this sampling procedure
via rectangular matrix Chernoff bounds~\cite{Tropp12}.
It also appears as Theorem A.1. in Appendix A
of~\cite{CohenKPPRSV17}.

\begin{algorithm}

  \caption{$\EulerianSparsify(\dir{G}, \epsilon, \CycleDecomposition)$}
  \textbf{Input}: Eulerian directed graph $\dir{G}$ with integer, polynomially bounded
  edge weights.

  \begin{tight_enumerate}
  	 \item Decompose each edge of $\dir{G}$ by its binary representation.
  	 Now edge weights of $\dir{G}$ are powers of $2$, and are at most $m \log n$ in number.
	 \label{Line:EulerianSparsifyDecompose}
    \item Compute $\rr$, a $1.5$-approximate estimate of effective resistances in $G$.
    \item While 
      $|E(\dir{G})| \geq 8\mhat \log n
       + O\left(n L^3 \epsilon^{-2}
       \log{n} \right) $

    \begin{tight_enumerate}
      \item $\dir{G} \leftarrow \DirectedSparsifyOnce(\dir{G}, \rr,
        \CycleDecomposition)$.

    \end{tight_enumerate}

    \item Return $\dir{G}$.

  \end{tight_enumerate}

  \label{alg:EulerianSparsifyFull}
\end{algorithm}

\begin{theorem}
\label{thm:EulerianSparsifyFull}
Given an Eulerian directed graph $\dir{G}$ with poly bounded edge weights
\footnote{We also assume that $\dir{G}$ has no parallel edges.},
and a cycle decomposition routine $\CycleDecomposition$,
the algorithm \EulerianSparsify returns an Eulerian directed graph $\dir{H}$ with at most
$ O\left(
\mhat \log{n}
+n L^3 \epsilon^{-2}
\log n
\right)
$
edges such that
with high probability,
\[
\left\|
\LL_{G}^{\nfrac{+}{2}}
\left( \LL_{\dir{G}} - \LL_{\dir{H}} \right)
\LL_{G}^{\nfrac{+}{2}}
\right\|_2
\leq \epsilon.
\]
The algorithm \EulerianSparsify runs in time
\[
O\left(m \log^2{n} \right) + \TCycle \left(O\left( m\log n\right), n\right),
\]
and its pseudocode is in Algorithm~\ref{alg:EulerianSparsifyFull}.
\end{theorem}

Plugging in the guarantees of {\NaiveCycleDecomposition}
from Theorem~\ref{thm:NaiveCycleDecomposition}
then gives the existence of smaller sparsifiers for Eulerian
Laplacians as stated in Theorem~\ref{thm:intro:euleriean-sparsify}.
Note that in Theorem \ref{thm:EulerianSparsifyFull} we assume the edge weights are integers. To deal with general weights, we only need to split each edge by the binary representation of its weight and ignore insignificant bits. Details are described in Appendix \ref{sec:ReductionToUnit}.

\begin{lemma}
\label{lem:RectangleMatrixChernoff}
(Theorem 1.6 from~\cite{Tropp12})
\label{lem:RectangularMatrixChernoff}
Let $\XX_{1}, \XX_{2}, \ldots, \XX_{k}$
be  $n \times n$
random matrices such that
\begin{tight_enumerate}
\item they are $0$-matrices in expectation
\[
\expec{}{\XX_{i}} = 0,
\]
\item they have small norm:
\[
\left\|
\XX_{i}
\right\|_2
\leq
O\left( \sqrt{\frac{nL^3}{m\log n}}\right),
\]
\item their expected outer and inner products are small (in summation),
\[
\left\| \sum_{1 \leq i \leq k} \expec{}{\XX_{i} \XX_{i}^{\top}} \right\|_2,
\left\| \sum_{1 \leq i \leq k} \expec{}{\XX_{i}^{\top} \XX_{i}} \right\|_2
\leq
O\left(\frac{nL^3}{m}\right).
\]
\end{tight_enumerate}
Then w.h.p., the $2$-norm of the sum is small:
\[
\prob{}{\left\| \sum_{1 \leq i \leq k} \XX_{i} \right\|_2
> \sqrt{\frac{nL^3\log n}{m}}}
\leq n^{-\Omega(1)}.
\]
\end{lemma}

Specifically, this can be obtained from Theorem 1.6
of~\cite{Tropp12} with the parameters
\begin{align*}
&d_1 = d_2 = n,
&t = \sqrt{\frac{nL^3\log n}{m}},\\
&R = O\left( \sqrt{\frac{nL^3}{m\log n}}\right),
&\sigma^2 = O\left(\frac{nL^3}{m}\right).
\end{align*}

This motivates us to give such bounds for each cycle
whose total effective resistance is small.
This will be by decomposing the terms, for which
we need the following lemma.
It is a direct consequence of the Cauchy-Schwarz inequality.

\begin{lemma}
\label{lem:OuterTriangle}
Let $\AA_{1} \ldots \AA_{\ell}$ be matrices with sum
$
\AA
:=
\sum_{i = 1}^{\ell} \AA_{i},$
then we have
\[
\AA \AA^{\top}
\preceq
\ell \cdot \sum_{i} \AA_{i} \AA_i^{\top}.
\]
\end{lemma}

\begin{proof}
By the Cauchy-Schwarz inequality, for any set of values
$a_1 \ldots a_{\ell}$, we have
\[
\left(\sum_{i = 1}^{\ell} a_{i} \right)^2
\leq
\left(\sum_{i = 1}^{\ell} a_{i}^2 \right) \left( \sum_{i = 1}^{\ell} 1\right).
\]
Applying this entry-wise to vectors
$\yy_{1} \ldots \yy_{\ell}$ gives
\[
\norm{\sum_{i = 1}^{\ell} \yy_{i}}_2^2
\leq
\ell \cdot \sum_{i = 1}^{\ell} \norm{\yy_{i}}_2^2.
\]
Then substituting in
$
\yy_{i} \leftarrow \AA \xx_{i}
$
gives
\[
\norm{\sum_{i = 1}^{\ell} \AA_{i} \xx}_2^2
\leq
\ell \cdot \sum_{i = 1}^{\ell} \norm{\AA_{i} \xx}_2^2,
\]
or in matrix form:
\[
\xx^\top \AA^\top \AA \xx
\leq
\ell \cdot \sum_{i} \xx^\top \AA_{i}^{\top} \AA_{i} \xx.
\]
As this holds for any $\xx$, we get the desired condition on the matrices.
\end{proof}

Our key statement is as follows:
\begin{lemma}
\label{lem:DirectedSize}
Let $\dir{C}$ be a equal weighted directed cycle of length $L$
contained in a graph $\dir{G}$
where each edge $\dir{e} \in \dir{C}$ satisfies
\[
ER_{G}(e) \leq \rho/\ww_e
\] where $\ww_e$ is the weight of $e$ (and all edges in $\dir{C}$).
Let $C$ be the undirectification of $\dir{C}$,
and let $\dir{C}^{\textsc{cw}}$ represent only the clockwise 
edges in $\dir{C}$.
Let $\LLtil$ be $2\LL_{\dir{C}^{\textsc{cw}}}-\LL_{\dir C}$.   
Then
\[
\LL_{G}^{\nfrac{+}{2}} \cdot
\left(\LLtil^\top
\LL_{G}^{+}
\LLtil\right) \cdot
\LL_{G}^{\nfrac{+}{2}}
\preceq
O\left(L^3 \rho\right)
\cdot
\LL_{G}^{\nfrac{+}{2}}
\LL_{C}
\LL_{G}^{\nfrac{+}{2}},
\]
and
\[
\LL_{G}^{\nfrac{+}{2}} \cdot
\left(\LLtil
\LL_{G}^{+}
\LLtil^\top\right) \cdot
\LL_{G}^{\nfrac{+}{2}}
\preceq
O\left(L^3 \rho\right)
\cdot
\LL_{G}^{\nfrac{+}{2}}
\LL_{C} 
\LL_{G}^{\nfrac{+}{2}}.
\]
\end{lemma}

Before proving Lemma~\ref{lem:DirectedSize},
we will define some useful notation and state
a couple of auxiliary lemmas.

For a directed edge $\dir{e} \in \dir{G}$,
we denote by $rev(\dir{e})$ its reversed version.
Also, recall that its undirected counterpart in
$G$ (with half the weight) is denoted by $e$.
We also extend this notation to sets of edges.
	
We first reduce things to a directed cycle with only
clockwise edges.
Note that the directed Laplacian elements corresponding
to an edge $\dir{e}$ in both directions sum up to
twice the undirected Laplacian for $e$.

\[
\LL_{\dir{e}} + \LL_{rev(\dir{e})} = 
\ww_e\left[
\begin{array}{cc}
0 & -1\\
0 & 1
\end{array}
\right]
+
\ww_e\left[
\begin{array}{cc}
1 & 0\\
-1 & 0
\end{array}
\right]
=
\ww_e\left[
\begin{array}{cc}
1 & -1\\
-1 & 1
\end{array}
\right]
= 2\LL_{e}.
\]
Using this in the definition of $\LLtil$ gives 
\[
\LLtil=2\LL_{\dir{C}^{\textsc{cw}}}-\LL_{\dir{C}}= \LL_{\dir{F}} - 2\LL_{S},
\]
where $\dir{F}$ is the version of $\dir{C}$ where every edge is clockwise,
and $S$ is the set of undirected edges in $G$ corresponding to
counterclockwise edges in $\dir{C}$.
\begin{align*}
& \LL_\dir{F} = \LL_{\dir{C}^{\textsc{cw}}} + \LL_{rev(\dir{C} \setminus \dir{C}^{\textsc{cw}})}, \\
& \LL_S = \sum_{\dir{e} \in \dir{C} \setminus \dir{C}^{\textsc{cw}}} \LL_{e}.
\end{align*}
The triangle inequality from Lemma~\ref{lem:OuterTriangle}
implies that it suffices to bound these two terms:
\[
\LL_{G}^{\nfrac{+}{2}} \cdot
\left(\LL_{S}
\LL_{G}^{+}
\LL_{S}\right) \cdot
\LL_{G}^{\nfrac{+}{2}}
,\]
and
\[
\LL_{G}^{\nfrac{+}{2}} \cdot
\left(\LL_{\dir{C}^{\textsc{cw}}}^\top
\LL_{G}^{+}
\LL_{\dir{C}^{\textsc{cw}}}\right) \cdot
\LL_{G}^{\nfrac{+}{2}}.
\]
We use the following two lemmas to derive the bounds:
\begin{lemma}
\label{lem:DirSizeFirst}
Let $S$ be the set of undirected edges in $G$ corresponding to
counterclockwise edges in a cycle $\dir{C}$ in $\dir{G}$
of length at most $L$, where the effective resistances
of edges in $C$ are at most $\rho$. Then,
\[
\LL_{S}
\preceq
L \cdot \rho \cdot \LL_{G}.
\]
\end{lemma}
\begin{proof}
\[
\LL_{S} = \sum_{\dir{e} \in \dir{C} \setminus \dir{C}^{\textsc{cw}}} \LL_{e}.
\]
For every edge $\dir{e} \in \dir{C}$, we know that $ER_G(e) \leq \rho$.
This, along with the definition of effective resistance in $G$, we get
\[
\LL_e \preceq \rho \cdot \LL_{G}.
\]
Using the above equations and bounding the length of $\dir{C}$ by $L$,
we get the desired bound.
\end{proof}

\begin{lemma}
\label{lem:DirSizeSecond}
Let $\dir{F}$ the graph that's a length $L$ cycle (with unit weights)
such that the effective resistance between adjacent vertices in $\dir{F}$
in $G$, the undirectification of $\dir{G}$ is at most $\rho$.
Then
\[
\LL_{\dir{F}}^{\top}
\LL_{G}^{+}
\LL_{\dir{F}}
\preceq
O\left( L^3\cdot\rho \right)
\LL_{C}.
\]
\end{lemma}
\begin{proof}
Let $F$ be the undirectification of $\dir{F}$.
Explicit computations of the entries of $\LL_{F}^{+}$ gives
\[
\LL_{\dir{F}}^{\top}
\LL_{F}^{+}
\LL_{\dir{F}}
=
\frac{2}{L}\LL_{K\left(L\right)},
\]
where $K(L)$ is a unit weighted clique on $L$ vertices.
By Cauchy-Schwarz inequality, 
\[
\sum_{i<j}(x_i-x_j)^2\le \sum_{i<j}\left((j-i)\sum_{k=i+1}^{j} (x_{k-1}-x_{k})^2\right)= O(L^3)\sum_{k=1}^L (x_{k-1}-x_{k})^2
\] where $x_0$ references to $x_L$. Thus 
\begin{align*}
\LL_{\dir{F}}^{\top}
\LL_{C}^{+}
\LL_{\dir{F}}
&=\frac{2}{L}\LL_{K_L}
\preceq
O(L^2)\LL_{C}.
\label{eq:LinverseboundedbyL}
\end{align*}
Similar to the proof for Lemma \ref{lem:DirSizeFirst}, by the bounds on effective resistances, we have 
\[
\LL_{C} = \sum_{e\in C} \LL_{e}
\preceq \sum_{e\in F} \rho \cdot \LL_{G}
= L \cdot \rho \cdot \LL_{G}.
\]
Thus,
$\LL_{G}^{+} \preceq L \cdot \rho \cdot \LL_{C}^{+}.$ 
Substituting for $\LL_{C}^{+}$ gives us
the claimed result.
\end{proof}

We are now equipped to prove Lemma \ref{lem:DirectedSize}.

\begin{proof}[Proof of Lemma~\ref{lem:DirectedSize}]
Rearranging the terms from Lemma~\ref{lem:DirSizeFirst}, we get
\[
\LL_{S}^{\nfrac{1}{2}}
\LL_{G}^{+}
\LL_{S}^{\nfrac{1}{2}}
\preceq
L\cdot \rho \cdot \II.
\]
Multiplying on the left and right by the appropriate terms,
and using the fact that $\LL_{S} \preceq \LL_{C}$, we can get
\[
\LL_{G}^{\nfrac{+}{2}} \cdot
\left(\LL_{S}
\LL_{G}^{+}
\LL_{S}\right) \cdot
\LL_{G}^{\nfrac{+}{2}}
\preceq
L\cdot \rho \cdot 
\LL_{G}^{\nfrac{+}{2}}
\LL_{C}
\LL_{G}^{\nfrac{+}{2}}.
\]
From Lemma~\ref{lem:DirSizeSecond},
\[
\LL_{G}^{\nfrac{+}{2}} \cdot
\left(\LL_{\dir{F}}^{\top}
\LL_{G}^{+}
\LL_{\dir{F}}\right) \cdot
\LL_{G}^{\nfrac{+}{2}}
\preceq
O\left( L^3\cdot\rho \right)
\LL_{G}^{\nfrac{+}{2}}
\LL_{C}
\LL_{G}^{\nfrac{+}{2}}.
\]
Combining these two bounds by Lemma \ref{lem:OuterTriangle} gives
\[
\LL_{G}^{\nfrac{+}{2}} \cdot
\left(\LLtil^\top
\LL_{G}^{+}
\LLtil\right) \cdot
\LL_{G}^{\nfrac{+}{2}}
\preceq
O\left(L^3\cdot \rho\right)
\cdot
\LL_{G}^{\nfrac{+}{2}}
\LL_{C}
\LL_{G}^{\nfrac{+}{2}}.
\]
The second inequality with the transpose on the other side
also follows similarly.
\end{proof}

\begin{algorithm}
	\caption{$\DirectedSparsifyOnce(\dir{G},\rr, \CycleDecomposition)$}
	
	\textbf{Input}:
	Eulerian directed graph $\dir{G}$,
	where the edge weights are integral powers of $2$,
	
	$\rr$:  2-approximate estimates of effective resistances in $G$,
	
	$\CycleDecomposition$: short cycle decomposition routine.
	
	\textbf{Output}:
	Eulerian sparsifier $\dir{H}$.
	
	\begin{tight_enumerate}
		\item Construct a set $\dir{E}_{highER}$ with every edge $e \in \dir{G}$ that satisfies $\ww_e\rr_e \geq 4n/m$.
		
		\item $\dir{H} \leftarrow \dir{E}_{highER}$ and $\dir{G} \leftarrow \dir{G} \setminus \dir{E}_{highER}$.
		
		\item Partition the  edges of $\dir{G}$ into unit weight graphs times a power of $2$, and denote the graph with edge-weights $2^i$ as $\dir{G_i}$. The output of this step will be $\dir{G_1},\ldots,\dir{G_s}$ where $s=O(\log n)$.
		
		\item For each $\dir{G_i}$:
		\begin{tight_enumerate}
			\item $\{C_{i,1},\ldots,C_{i,t}\} \leftarrow \CycleDecomposition(G_i)$.
			Note that $C_{ij}$ is not a directed cycle, but corresponds to a
			directed cycle $\dir{C}_{ij}$.
			\item $\dir{H}\leftarrow \dir{H}+\dir{G}_i \setminus \left(\cup_{j=1}^t \dir{C}_{ij}\right)$.
			\item For each cycle $\dir{C}_{ij}$:
			\begin{tight_enumerate}
				\item With probability $\nfrac{1}{2}$, add all its clockwise edges with twice their weight to $\dir{H}$;
				and with probability $\nfrac{1}{2}$, add all its counterclockwise edges with twice their weight to $\dir{H}$. \label{line:samplecycle}
			\end{tight_enumerate}
		\end{tight_enumerate}
		\item Return $\dir{H}$.
	\end{tight_enumerate}

	\label{alg:EulerianSparsify}
\end{algorithm}

Now, using Lemma \ref{lem:RectangularMatrixChernoff},
we can construct a sparsifier of $\dir{G}$ with at most $(\nfrac{15}{16})m$ edges:
\begin{lemma}
\label{lem:EulerianSparsify}
Given an Eulerian directed graph $\dir{G}$ with edge weights being integral powers of 2,
a $2$-approximate estimate $\rr$ of effective resistances in $G$,
and a cycle decomposition routine $\CycleDecomposition$, 
the algorithm $\DirectedSparsifyOnce$
(with pseudocode shown in Algorithm \ref{alg:EulerianSparsify}) outputs
in $O(m) + \TCycle(m, n)$ time
a directed Eulerian graph $\dir{H}$ with edges weights still being powers
of $2$ such that with high probability,
\[
\left\|
\LL_{G}^{\nfrac{+}{2}}
\left( \LL_{\dir{G}} - \LL_{\dir{H}} \right)
\LL_{G}^{\nfrac{+}{2}}
\right\|_2
\leq \epsilon,
\]
where $\epsilon=\sqrt{\frac{nL^3\log n}{m}}$.  Furthermore, if
$m=\Omega(\mhat \log n)$, the expected number of edges in $\dir{H}$ is at
most $(\nfrac{15}{16})m$ with high probability.
\end{lemma}
\begin{proof}
First, note that whether a cycle $\dir{C}$ is sampled as clockwise or
counterclockwise, the difference between the in-degrees and
out-degrees of the vertices does not change. 
Hence, the graph remains Eulerian.
Also, doubling a power of $2$ still gives a power of $2$, so we get that
$\dir{H}$ is still an Eulerian graph with edge weights being powers of $2$.

Next, we prove that the resulting graph $\dir{H}$ is a good approximation of $\dir{G}$. The difference between these two graphs comes from the cycles produced by $\CycleDecomposition$. Let these cycles be $\{\dir{C}_i\}$. We consider fitting the terms into the
requirements of the matrix concentration bound
from Lemma~\ref{lem:RectangularMatrixChernoff}.

For a cycle $\dir{C}_i$, recall that the algorithm samples either
$2\dir{C}_{i}^{\textsc{cw}}$ or $2\left(\dir{C}_i\backslash \dir{C}_{i}^{\textsc{cw}}\right)$,
each with probability $\nfrac{1}{2}$.
Let $\ww_i$ be the edge weight of $\dir{C_i}$. (Note that we run $\CycleDecomposition$ on equal weighted graphs. $\ww_i$ here must be a power of $2$ but we will not use that fact.)
Let $\rho/\ww_i$ be an upper bound for the effective resistance of every edge in $\dir{C_i}$.

Now, let $\YY_i$ be the deviation on $i$\textsuperscript{th} cycle:
\[
\YY_{i}
:=
\LL_{G}^{\nfrac{+}{2}}
\left( \LL_{\dir{D}_i} - \LL_{\dir{C}_i} \right)
\LL_{G}^{\nfrac{+}{2}}
\] where 
\[
\dir{D}_i=
\begin{cases}
2\dir{C}_{i}^{\textsc{cw}} & \text{w.p. $1/2$},\\
2\left(\dir{C}_i\backslash \dir{C}_{i}^{\textsc{cw}}\right) & \text{w.p. $1/2$.}
\end{cases}
\]
We use $\XX_{i}^{\textsc{cw}}$ to denote the deviation
on $\dir{C}_i$, when it is sampled as clockwise,
i.e., $\dir{D_i} = 2 \dir{C}_i$:
\[
\XX_{i}^{\textsc{cw}}
:=
\LL_{G}^{\nfrac{+}{2}}
\left(
2\LL_{\dir{C}_{i}^{\textsc{cw}}}-\LL_{\dir{C_i}}
\right)
\LL_{G}^{\nfrac{+}{2}}.
\]
Equivalently, we use $\XX_{i}^{\textsc{ccw}}$ to denote the deviation
on $\dir{C}_i$, when it is sampled as counterclockwise.
Naturally,
\[
\YY_i=
\begin{cases}
\XX_{i}^{\textsc{cw}} & \text{w.p. $\nfrac{1}{2}$},\\
\XX_{i}^{\textsc{ccw}} & \text{w.p. $\nfrac{1}{2}$.}
\end{cases}
\]
Using our bounds from Lemma \ref{lem:DirectedSize}, we have
\[
\norm{\XX_{i}^{\textsc{cw}}}_2
\leq
\sqrt{
	O \left( L^3\cdot \rho\right)
	\cdot
	\lambda_{\max}
	\left(
	\LL_{G}^{\nfrac{+}{2}}
	\LL_{{C_i}}
	\LL_{G}^{\nfrac{+}{2}}
	\right)
}.
\]
Since every edge in $C_i$ has a leverage score of at most $\rho$,
and using the fact that $C$ has at most $L$ edges,
\[
\norm{\XX_{i}^{\textsc{cw}}}_2
\leq
O \left(
\sqrt{
	L^4 \rho^2
}
\right)
= O \left( L^2 \rho\right)
\] where $\ww_i$ is the edge weight of $\dir{C_i}$
We can get the same bound for $\XX_{i}^{\textsc{ccw}}$
by a symmetric version of Lemma \ref{lem:DirectedSize}.

To bound the variance term
\[
\norm{\sum_{i}\expec{}{\YY_{i}^{\top} \YY_{i}}}_2,
\]
by Lemma~\ref{lem:DirectedSize}, both
$(\XX_{i}^{\textsc{cw}})^{\top} \XX_{i}^{\textsc{cw}}$
and
$(\XX_{i}^{\textsc{ccw}})^{\top} \XX_{i}^{\textsc{ccw}}$
are bounded above by
\[
w_i\cdot O \left(L^3 \cdot \rho\right) \cdot
\LL_{G}^{\nfrac{+}{2}}
\LL_{C_i}
\LL_{G}^{\nfrac{+}{2}}.
\]
From the definition of $\YY_i$, and summing over all $i$,
\[
\sum_{i} \expec{}{\YY_{i}^{\top} \YY_{i}}
\leq
O\left(L^3 \cdot \rho\right) \cdot \sum_{i}
\LL_{G}^{\nfrac{+}{2}}
\LL_{C_i}
\LL_{G}^{\nfrac{+}{2}}.
\]
Since the cycles are edge disjoint, we have
\[
\sum_{i} \LL_{{C_i}} \preceq \LL_{G}.
\]
Composing on either side with $\LL_{G}^{\nfrac{+}{2}}$ gives
\[
\sum_{i}
\LL_{G}^{\nfrac{+}{2}}
\LL_{C_i}
\LL_{G}^{\nfrac{+}{2}}
\preceq \II,
\]
and hence,
\[
\norm{\sum_{i}\expec{}{\YY_{i}^{\top} \YY_{i}}}_2
\leq
O\left(L^3 \cdot \rho\right).
\]
Similarly, 
\[
\norm{\sum_{i}\expec{}{\YY_{i} \YY_{i}^{\top}}}_2
\leq
O\left(L^3 \cdot \rho\right).
\]
Now, by taking 
\[
\rho=\frac{4n}{m},
\]
Lemma~\ref{lem:RectangularMatrixChernoff} gives concentration with high probability.

Next we bound the number of edges in $\dir{H}$:
\begin{tight_enumerate}
\item $\leq m/2$ edges with low leverage score.

By the same reason as the undirected case (Section \ref{sec:Degree-Preserving}), since $\rr_e$ is a 2-approximate estimate of effective resistance, by Foster's theorem, the average value of $\ww_e\rr_e$ of all edges is at most $2(n - 1)/m$. So at most half of the edges in $\dir{G}$ have $\ww_e\rr_e\ge \frac{4n}{m}$.

\item $\leq m/4$ edges not in cycles.

At most $O(k\log n)$ edges are not in the cycles since $i \leq O(\log n)$.
If $m=\Omega(k\log n)$, this number does not exceed $m/4$.

This leaves at least $m/4$ edges in cycles.

\item The expected fraction of cycle edges that are added to $\dir{H}$ is $\nfrac{1}{2}$ in the last step. As there are at least $m/4L$ cycles and the length of each of the cycles is bounded by $L$, as long as $L=n^{o(1)}$ and $m=\Omega(n)$, by a Chernoff bound, with high probability, at most $\nfrac{3}{4}$ of the cycle edges are added to $\dir{H}$. 
\end{tight_enumerate}
We conclude that with high probability, the number of edges in $\dir{H}$ is at most $(\nfrac{15}{16})m$.

The running time is dominated by the cycle decomposition step.
Because these graphs are edge-disjoint, the super-additivity property
of $\TCycle(\cdot, n)$ from Equation~\ref{assumption:super-additive-T} gives
that the total cost of obtaining these decompositions is $\TCycle(m, n)$.
Also, there are $O(\log{n}) \leq O(m)$ such graphs, the overhead from handling
the different copies is a lower order term.
\end{proof}

Finally, by repeating $\DirectedSparsifyOnce$ a number of times, we can reduce the number of edges to almost linear in the number of vertices. 

\begin{proof}[Proof of Theorem~\ref{thm:EulerianSparsifyFull}]
First we specify the condition under which our guarantees on (1) number of edges, (2) running time and (3) approximation hold. 

Because the number of edges reduces by a constant factor with high probability, we know that with high probability, the number of edges reduces by a constant factor in each of the first $O(\log \frac{m_0}{n})$ rounds. ($m_0$ is the number of edges initially in $\dir{G}$.)

By Lemma \ref{lem:EulerianSparsify}, the resulting graph of \DirectedSparsifyOnce approximates the given graph with high probability.
Thus, with high probability, each of the $O(\log \frac{m_0}{n})$ rounds of \DirectedSparsifyOnce produces a valid approximation of the result of its previous round. 

Next we consider the case where each of the $O(\log \frac{m_0}{n})$ rounds of \DirectedSparsifyOnce reduces the number of edges by a factor of $\nfrac{1}{16}$, and also produces a valid approximation. We have proven that this happens with high probability.

As the number of edges reduces geometrically, the total error, which
is bounded by the sum of the errors in all iterations, is bounded up
to a constant factor by the error in the last round:
\[
O\left(\sqrt{\frac{nL^3\log n}{m}}\right)
\]
where $m$ is the number of edges in the last round.
Since the algorithm stops with
\[ m = \Omega \left( nL^3 \eps^{-2} \log n\right),\] picking
appropriate constants, this implies that the final error is bounded
by $\eps.$

Since our effective resistance estimates depend on $L_G$,
this small error proves that our initial $1.5$-approximate
estimate effective resistances remain $2$-approximate throughout the algorithm.

The time bound follows from repeating \DirectedSparsifyOnce,
and an added $O(m \log^2 n)$ time for computing effective resistance estimates \cite{SpielmanS08:journal, KoutisLP12}.
Note that the number of edges increases by a factor of $O(\log{n})$
due to the initial splitting into powers of $2$
on Line~\ref{Line:EulerianSparsifyDecompose}
of Algorithm~\ref{alg:EulerianSparsifyFull}.
Then the edge counts are geometrically decreasing,
and we bound the overall cost by:
\begin{align*}
& O\left(  m \log^2{n} \right)
+ \sum_{i = 0}^{O\left( \log{n} \right)}
 \TCycle\left( O\left( \left( \frac{15}{16} \right)^{i} m \log{n}
   \right), n \right) \\
& \qquad \leq
O\left( m \log^{2}n \right)
+
\TCycle\left( O\left( \sum_{i = 0}^{O(\log n)} \left( \frac{15}{16} \right)^{i} m \log{n} \right), n \right)\\
& \qquad \leq
O\left( m \log^2{n} \right) + \TCycle\left( O\left( m \log{n} \right), n \right).
\end{align*}
where the inequality once again follows from the super-additivity assumption
of $\TCycle(\cdot, n)$ from Equation~\ref{assumption:super-additive-T}.
\end{proof}

%%% Local Variables:
%%% mode: latex
%%% TeX-master: "main"
%%% End:

%% file: resistance-sparsifiers.tex
%!TEX root = main.tex

\section{Graphical Spectral Sketches and Resistance Sparsifiers}
\label{sec:Resistance-Sparsifiers}
In this section, we show that every graph has a sparse spectral-sketch
with about $n^{1 + o(1)} \epsilon^{-1}$
that preserves the quadratic form of the graph Laplacian and its
inverse to $1 \pm \epsilon$ with high probability for a fixed vector.
The main result in this section is:
\begin{theorem}
\label{thm:SpectralSketches}
Given an undirected weighted graph $G$, a parameter $\eps$,
and a cycle decomposition routine \CycleDecomposition,
  the algorithm {\SpectralSketch} (Algorithm~\ref{alg:SpectralSketch})
  returns in
  \[\Otil{m} + \TCycle(O(m \log n), n)\]
  time a graph $H$ with
  $
  \otil{\mhat + n L \eps^{-1}}
$
edges, such that with high probability,
\begin{tight_enumerate}
\item $H$ is a $\sqrt{\epsilon}$-sparsifier
of $G$, $\LL_{G} \approx_{\sqrt{\epsilon}} \LL_{H}$,
\label{part:Sparsify}
\item for any fixed vector $\xx$,
the quadratic form in $\xx$ is approximately preserved,
$\xx^{\top} \LL_{H}\xx
\approx_{\eps}
\xx^{\top} \LL_{G} \xx,$ and
\label{part:Forward}
\item for any fixed vector $\xx$, the inverse quadratic form in $\xx$
  is approximately preserved,
  $\xx^{\top} \LL_{H}^{+} \xx \approx_{\eps} \xx^{\top} \LL_{G}^{+}
  \xx.$
\label{part:Backward}
\end{tight_enumerate}
\end{theorem}

Note that part~\ref{part:Backward} of
Theorem~\ref{thm:SpectralSketches} implies
Corollary~\ref{cor:resistance-sparsifiers}.
Combining Theorem~\ref{thm:SpectralSketches} with our two cycle
decomposition algorithms leads to our main result on graphical
sketches (Theorem~\ref{thm:overview:spectral-sketches}).
The guarantees of Theorem~\ref{thm:SpectralSketches} imply
  Theorem~\ref{thm:overview:spectral-sketches} when we use {\ShortCycleAlgo} as
  the {\CycleDecomposition} algorithm in Algorithm {\SpectralSketch}. 
\begin{proof}[Proof of Theorem~\ref{thm:overview:spectral-sketches}, assuming
  Theorem~\ref{thm:SpectralSketches}.]
Using either {\NaiveCycleDecomposition} or {\ShortCycleAlgo}
  as the algorithm {\CycleDecomposition}, in Algorithm $\SpectralSketch$, the runtime/sketch size
tradeoffs are:
  \begin{enumerate}
  \item Using {\NaiveCycleDecomposition}: ${\SpectralSketch}$
    runs in $\otil{mn}$ time, and returns an $H$ with
    $\widetilde{O}(n\epsilon^{-1})$ edges.
  \item Using {\ShortCycleAlgo}: ${\SpectralSketch}$ runs in
    $m^{1+o(1)}$ time, and returns an $H$ with
    $n^{1+o(1)} \eps^{-1}$  edges.
  \end{enumerate}
  Using {\ShortCycleAlgo} gives us the result of
  Theorem~\ref{thm:overview:spectral-sketches}.
\end{proof}

A natural approach towards this result is a better analysis of the
degree-preserving sparsification from
Section~\ref{sec:Degree-Preserving}, but showing a better error
dependency of about $n / m$ (rather than $\sqrt{n/m}$).  However, we
describe a counter example to this approach.

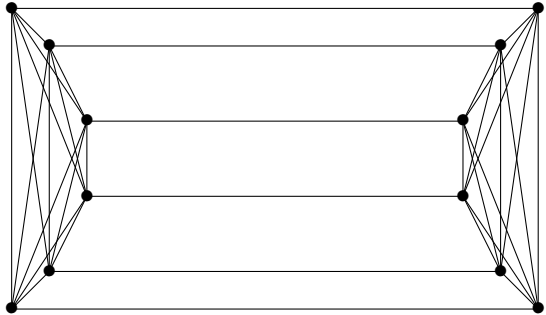
\begin{figure}[H]
  \centering
\begin{tikzpicture}
\foreach \Point in {(0,2),(0.5,1.5),(1,0.5),(1,-0.5),(0.5,-1.5),(0,-2)}{
  \node at \Point {\textbullet};
}
\foreach \Point in {(7,2),(6.5,1.5),(6,0.5),(6,-0.5),(6.5,-1.5),(7,-2)}{
  \node at \Point {\textbullet};
}
  \draw (0,2) -- (7,2);
  \draw (0.5,1.5) -- (6.5,1.5);
  \draw (1,0.5) -- (6,0.5);
  \draw (0,-2) -- (7,-2);
  \draw (0.5,-1.5) -- (6.5,-1.5);
  \draw (1,-0.5) -- (6,-0.5);
  \draw (0,2) -- (0.5,1.5);
  \draw (0,2) -- (1,0.5);
  \draw (0,2) -- (1,-0.5);
  \draw (0,2) -- (0.5,-1.5);
  \draw (0,2) -- (0,-2);
  \draw (0.5,1.5) -- (1,0.5);
  \draw (0.5,1.5) -- (1,-0.5);
  \draw (0.5,1.5) -- (0.5,-1.5);
  \draw (0.5,1.5) -- (0,-2);
  \draw (1,0.5) -- (1,-0.5);
  \draw (1,0.5) -- (0.5,-1.5);
  \draw (1,0.5) -- (0,-2);
  \draw (1,-0.5) -- (0.5,-1.5);
  \draw (1,-0.5) -- (0,-2);
  \draw (0.5,-1.5) -- (0,-2);
  \draw (7,2) -- (6.5,1.5);
  \draw (7,2) -- (6,0.5);
  \draw (7,2) -- (6,-0.5);
  \draw (7,2) -- (6.5,-1.5);
  \draw (7,2) -- (7,-2);
  \draw (6.5,1.5) -- (6,0.5);
  \draw (6.5,1.5) -- (6,-0.5);
  \draw (6.5,1.5) -- (6.5,-1.5);
  \draw (6.5,1.5) -- (7,-2);
  \draw (6,0.5) -- (6,-0.5);
  \draw (6,0.5) -- (6.5,-1.5);
  \draw (6,0.5) -- (7,-2);
  \draw (6,-0.5) -- (6.5,-1.5);
  \draw (6,-0.5) -- (7,-2);
  \draw (6.5,-1.5) -- (7,-2);
\end{tikzpicture}
  \caption{An instance of our counter-example with $n=6$}
    \label{fig:counter-example}
\end{figure}

Consider two cliques of size $n$ connected by a matching of size $n$
(see Figure~\ref{fig:counter-example}).  The effective resistance of
each edge, including the edges of the matching are at most $3 / n$.
This is because each matching edge has at least $n$ edge-disjoint
paths of length $3$ connecting its two end points: one path through
each edge of the matching using the clique edges.  We can also couple
these edges into short cycles of length $4$, consisting of two
matching edges and two clique edges.  Now consider a test vector that
is $0$ on the first clique, and $1$ on the second clique.  Randomly
sampling the matching edges incurs a variance which is about
$\sqrt{n}$, so even though all the edges that we sample have
resistance at most $3 / n$, the accuracy of our sampling process is
limited at $1 / \sqrt{n}$.

Instead, we obtain the better dependency on $\eps$
using the expander based sketching ideas from~\cite{JambulapatiS18}.
In Section~\ref{subsec:LessError}
we give a bound for a new degree-preserving
scheme whose error depends on the degrees instead of the
original quadratic form, and has a $\eps$ dependence.
We then incorporate expander decompositions in the same
manner as~\cite{JambulapatiS18} to give sketches for arbitrary
graphs in Section~\ref{subsec:SpectralSketches}.
Then in Section~\ref{subsec:SketchBackward}, we show that
guarantees~\ref{part:Sparsify} and~\ref{part:Forward} of
Theorem~\ref{thm:SpectralSketches} imply approximation of resistances, proving
Corollary~\ref{cor:resistance-sparsifiers}.

\subsection{Bounding Variance in Terms of Degrees}
\label{subsec:LessError}

We first analyze the errors of an analog of the degree-preserving
sparsification routine in the same manner as
Jambulapati and Sidford~\cite{JambulapatiS18}.
Specifically, we obtain bound on variance,
and in turn error, in terms of the degrees of the vertices. 
This dependency on degrees leads to a slight modification
of the algorithm, in that it no longer samples edges incident
to low degree vertices.
Pseudocode of this routine is in Algorithm~\ref{alg:DecomposeAndSample}.

\begin{algorithm}

\caption{$\DecomposeAndSample
(G, \alpha, \CycleDecomposition)$}

\textbf{Input}:
An undirected simple graph $G$ with all edges having the same
weight (WOLOG $1$),

Degree threshold $\alpha$,

Cycle decomposition routine \CycleDecomposition.

\textbf{Output}:
Approximation $H$ with edge weights that are
either $1$ or $2$.
\begin{tight_enumerate}
\item Let $V_{big}$ be the vertices with degree at least
$\alpha$ in $G$.

\item Let $\hat{G}$ be a bipartition of $G[V_{big}]$ with
at least half the edges of $G[V_{big}]$.

\item $\{C_{1} \ldots C_{t}\}
\leftarrow \CycleDecomposition(\hat{G})$.

\item Initialize $H \leftarrow G \setminus
\{ C_1 \cup C_2 \ldots \cup C_{t}\}$.

\item For each cycle $C_{i}$:
\begin{tight_enumerate}
   \item With probability $\nfrac{1}{2}$, add all the odd indexed edges
    of $C_{ij}$ into $H$ with weights doubled,
    otherwise add all the even indexed edges into $H$ with weights doubled.
\end{tight_enumerate}

\item Return $H$.

\end{tight_enumerate}

\label{alg:DecomposeAndSample}
\end{algorithm}

\begin{lemma}
\label{lem:BetterError}
Given a unit weighted simple graph $G$,
a parameter $\alpha$,
and a cycle decomposition routine $\CycleDecomposition$,
the algorithm $\DecomposeAndSample$
returns a graph $H$ with at most
\[
\mhat + n \alpha + (\nfrac{3}{4}) m
\]
edges in $O(m + \TCycle(m, n))$ time, such that for any vector $\xx$,
we have with high probability
\[
\abs{
\xx^{\top} \LL_{H} \xx
-
\xx^{\top} \LL_{G} \xx
}
\leq
\alpha^{-1}
\cdot
O\left( \log{n} \right)
\cdot
  \sqrt{L}
\cdot
\sum_{u} \dd_{G, u} \left( \xx_{u} - \xhat \right)^2
\]
where $\dd_{G, u}$ is the degree of vertex $u$ in $G$,
and $\xhat$ is an arbitrary scalar.
\end{lemma}

Our proof utilizes concentration bounds for sums of scalars,
specifically Bernstein bounds. We state a slight modification of Bernstein's
Inequality:
\begin{lemma}[Bernstein's Inequality]
\label{lem:Bernstein}
Let $X_1 \ldots X_{t}$ be independent random variables such that
$\expec{}{X_i} = 0$, and for some parameter $\theta,$ we have
$\abs{ X_{i}} \leq \theta$, and
$\sum_{i = 1}^{t} \expec{}{X_{i}^2} \leq \theta^2$.
Then with high probability we have
$\left|
  \sum_{i = 1}^{t} X_i
\right|
\leq
O\left(\log{n}\right) \cdot \theta.$
\end{lemma}

The above formulation comes from~\cite{ChungL06},
which is obtained by combining Theorems 8.8 and 8.9 in that
document with
$E(X_i) = 0,
\lambda = O\left(\log{n}\right) \cdot \theta,$
and $\norm{X} = M = \theta$,
for some parameter $\theta$.

We utilize Bernstein's inequality to prove the main lemma
for this section (Lemma~\ref{lem:BetterError}), by creating one
$X_{i}$ per cycle.
\begin{proof}(of Lemma~\ref{lem:BetterError})
For each cycle $C_{i}$,
let $\widetilde{C}_{i}$ denote the result of sampling
either the odd or even indexed edges.
We define the random variable $X_{i}$ as
\[
X_{i}
:=
\xx^{\top} \LL_{C_{i}} \xx
-
\xx^{\top} \LL_{\widetilde{C}_{i}} \xx.
\]
This is a mean $0$ random variable because 
the expected weight of each edge in $\widetilde{C}$ is 1.
Now, suppose the vertices along $C_{i}$
are $u_{i,1} \ldots u_{i,n_i}$ for some even $n_i = |C_i|$.
Note that for an edge $u_{i} u_{i + 1}$,
we can write its corresponding term in the quadratic form,
$( \xx_{u_{i, j}} - \xx_{u_{i, j + 1}} )^2$
as
\[
\left( \xx_{u_{i,j}} - \xhat \right)^2
+
\left( \xx_{u_{i,j + 1}} - \xhat \right)^2
-
2 \left( \xx_{u_{i,j}} - \xhat \right)
\cdot \left( \xx_{u_{i, j + 1}} - \xhat \right)
\]
We use this identity to express $X_i.$ Since the degree of each vertex
in $C_i$ is preserved when we sample $\widetilde{C}_i,$ the first two
terms in the above expression cancel out for $\LL_{C_i}$ and
$\LL_{\widetilde{C}_{i}},$ and only the third term remains in $X_i.$ In the case where we $\widetilde{C}_{i}$ is
all the odd indexed edges with double their weights, we get
\[
X_i = -2\sum_{\substack{j = 1 \\ j \sim \text{odd}}}^{n_{i}}
 \left( \xx_{u_{i, j}} - \xhat \right)
\cdot \left( \xx_{u_{i, j + 1}} - \xhat \right)
+
2\sum_{\substack{j = 2 \\ j \sim \text{even}}}^{n_{i}}
\left( \xx_{u_{i, j}} - \xhat \right)
\cdot \left( \xx_{u_{i, j + 1}} - \xhat \right).
\]
The case where we pick only the even edges is identical, but
with signs flipped.  Thus, in either case, we have
\[
|X_{i}|
\leq
2 \sum_{j = 1}^{n_i}
\left|
  \left( \xx_{u_{i, j}} - \xhat \right)
    \cdot \left( \xx_{u_{i, j + 1}} - \xhat \right)
\right|.
\]

Using these bounds, we invoke Bernstein's inequality
from Lemma~\ref{lem:Bernstein} on these variables with
\[
\theta
=
2
\cdot
\alpha^{-1}
\cdot
\sqrt{L}
\cdot
\sum_{u \in C_i} \dd_{G, u}
  \left( \xx_{u} - \xhat \right)^2,
\]
for which we need to check the bounds
on $|X_i|$ and $\sum_{i} \expec{}{X_i^2}.$
For the bound on $|X_i|$, the arithmetic-geometric mean (AM-GM)
inequality gives
\[
2\left|
\left( \xx_{i, j} - \xhat \right)
    \cdot \left( \xx_{u_{i, j + 1}} - \xhat \right)
\right|
\leq
\left( \xx_{i, j} - \xhat \right)^2
+ \left( \xx_{u_{i, j + 1}} - \xhat \right)^2,
\]
so the overall sum is at most
\[
2 \sum_{u \in C_i}
\left( \xx_{u} - \xhat \right)^2
\leq
\theta,
\]
since all the degrees are at least $\alpha$.

For the variance term, the Cauchy-Schwarz inequality gives for each
$X_i$,
\[
\expec{}{X_i^2}
\leq 
\left(
\sum_{j = 1}^{n_i}
2
\left|
\left( \xx_{u_{i, j}} - \xhat \right)
\cdot \left( \xx_{u_{i, j + 1}} - \xhat \right)
\right|
\right)^2 
\leq
4 n_i
\cdot
\sum_{j = 1}^{n_i}
\left( \xx_{u_{i, j}} - \xhat \right)^2
\left( \xx_{u_{i, j + 1}} - \xhat \right)^2.
\]
Since $G$ is a simple graph, the sum over all these terms
can in turn be upper bounded by the sum over all possible pairs
of vertices, and factorized:
\[
\sum_{i = 1}^{t} \expec{}{X_i^2}
\leq
4 L
\cdot
\sum_{u \in V_{big} } \sum_{v \in V_{big}}
\left( \xx_{u} - \xhat \right)^2
\left( \xx_{v} - \xhat \right)^2
=
4 L
\cdot
\left( \sum_{u \in V_{big} } \left( \xx_{u} - \xhat \right)^2 \right)^2
\leq
\theta^2,
\]
where the last inequality follows from all vertices in $V_{big}$
having degree at least $\alpha$.

The result then follows from Lemma~\ref{lem:Bernstein}.
\end{proof}

We remark again that the proof,
specifically the factorization of variance,
crucially depends on $G$ being a simple graph.

\subsection{Incorporating Expander Partitioning}
\label{subsec:SpectralSketches}
We now turn the error guarantees from Section~\ref{subsec:LessError},
specifically Lemma~\ref{lem:BetterError} to one that works for general
graphs.  A key observation in~\cite{JambulapatiS18} is that this sum
of vertex-degree terms can be upper bounded by $\xx^\top \LL \xx$
via Cheeger's inequality.  However, they explicitly renormalize the
vector $\xx$ against the degrees of each piece of the expander.  We
show instead that these guarantees extend to Laplacians directly with
a loss depending on the conductance of the graph.
\begin{definition}
\label{def:Conductance}
Given an unweighted undirected graph $G$ and a subset of vertices $S$,
the conductance of $S$ is
\[
\min_{\Shat \subseteq S}
\frac{\abs{E\left( \Shat, S \setminus \Shat \right)}}
{\min\left\{
  \sum_{u \in \Shat} \dd_u,
  \sum_{u \in S \setminus \Shat} \dd_u
\right\}},
\]
where $\dd_u$ is the degree of $u$ in $G,$ and $E(A, B)$ denotes the
set of edges with one end point in $A$ and another in $B.$
\end{definition}

Note that the conductance of a subgraph $G(S)$ is defined not w.r.t
the degrees in $G(S)$, but still the degree of $G$.  This is crucial for
combining our routine with graph sparsification in
Section~\ref{subsec:kekeke}.  Cheeger's inequality can also be stated
with respect to this subgraph case, giving:
\begin{lemma}[Cheeger's inequality,~\cite{AlonM85}]
	\label{lem:Cheeger}
In any graph $G$, for any subset $S$ with conductance $\phi$, we have
\[
\lambda_{2} \left( \DD_{S}^{-\nfrac{1}{2}} \LL_{G\left[ S \right]} \DD_{S}^{-\nfrac{1}{2}} \right)
\geq
\frac{1}{2} \phi^2,
\]
where $\LL_{G[S]}$ is the Laplacian matrix of the subgraph of $G$
induced by $S$, and $\DD_{S}$ is the minor of $\DD$ restricted to the
vertices in $S$.
\end{lemma}

This formulation is tailored towards expander partitions as they
were used in graph sparsification by Spielman and
Teng~\cite{SpielmanT11:journal}.
It can be obtained from the more standard form of Cheeger's inequality
by adding self loops at the vertices so their degrees match.
\begin{lemma}
\label{lem:CheegerProjected}
If a unit weighted, undirected graph $G$ has a subset $S$
of vertices with conductance at least $\phi$,
then for any vector $\xx$ we
have
\[
\sum_{uv \in E\left( G\left[ S \right] \right)}
\left(\xx_{u} - \xx_{v} \right)^2
\geq
\frac{1}{2} \phi^2
\sum_{u \in S} \dd_{u} \left( \xx_{u} - \xhat \right)^2,
\]
where $\xhat = \sum_{v}
\left(\dd_{v}/\left\| \dd \right\|_1\right) \xx_{v}$.
\end{lemma}

\begin{proof}
Since Lemma~\ref{lem:Cheeger} takes any degrees,
we drop the subscript in $S$ and work with any $\LL$ and $\dd$ pair.

Denote $\dd^{1/2}$ as the vector $(\sqrt{d_1}, \sqrt{d_2}, \cdots)$ 
w.r.t. the degree sequence $\dd$.
Then, define $\PPi_{\perp \dd^{1/2}}$ as the orthogonal projection matrix against the unit vector 
$
\dd^{1/2}/\left\| \dd \right\|_1^{1/2},
$
i.e.,
\[
\PPi_{\perp \dd^{1/2}} =
\II - \frac{\dd^{1/2} {\dd^{1/2}}^{\top}}{\left\| \dd \right\|_1}.
\]
$\DD^{-1/2} \LL \DD^{-1/2}$ and $\PPi_{\perp \dd^{1/2}}$ share the
same null space represented by $\dd^{1/2}$ as
$\DD^{-1/2} \dd^{1/2} = \one$.  Hence, the eigenvalue condition from
Cheeger's inequality~\ref{lem:Cheeger} gives
\[
\DD^{-\nfrac{1}{2}} \LL \DD^{-\nfrac{1}{2}}
\succeq 
\frac{1}{2} \phi^2 \PPi_{\perp \dd^{\nfrac{1}{2}}}.
\]
Thus, for any $\xx,$ composing by $\DD^{\nfrac{1}{2}}\xx$ on both sides gives us
\[
  \xx^{\top} \LL
  \xx
  \geq 
\frac{1}{2} \phi^2
\xx^{\top} \DD^{\nfrac{1}{2}} \PPi_{\perp \dd^{1/2}} \DD^{\nfrac{1}{2}}
\xx
= \frac{1}{2} \phi^2
\xx^{\top} 
\left(
\DD - \frac{\dd \dd^{\top}}{\left\| \dd \right\|_1} 
\right)
\xx.
\]
Or equivalently,
\[
\sum_{uv \in E} \left(\xx_{u} - \xx_{v} \right)^2
\geq
\frac{1}{2} \phi^2
\left(
\xx^{\top} \DD \xx - 
\frac{\xx^{\top}\dd \dd^{\top} \xx}{\left\| \dd \right\|_1}
\right)
=
\frac{1}{2} \phi^2
\sum_{u} \dd_{u} \left( \xx_{u} - \xhat \right)^2.
\]
\end{proof}

The lack of restrictions on the vector $\xx$
makes it significantly simpler to sum this guarantee
across a number of expanders.
Specifically, we invoke the following routine for partitioning
graphs into pieces contained in expanders.
\begin{lemma}[\cite{SpielmanT11:journal}, Lemma 32 in~\cite{KelnerLOS14}]
  \label{lem:ExpanderDecomposition}
There is an algorithm $\ExpanderDecompose(G, \phi)$
that for any unit weighted, undirected graph $G$ with $n$ vertices
and $m$ edges and any parameter $\phi > 0$, returns
with an overhead $\gst(n)$ that's upper bounded by $\log^{O(1)}n$
a partition of the vertices of $G$ into
\[
V
=
\hat{S}_1
\cupdot
\hat{S}_2
\cupdot
\ldots
\cupdot
\hat{S}_k
\]
in $\Otil{m \phi^{-2}}$ time such that
\begin{tight_enumerate}
\item the number of edges on the boundary of all $S_i$'s
is at most
$\gst(n) \phi m$,
\item each $\hat{S}_{i}$ is contained in some subset 
$S_{i} \supseteq \hat{S}_i$ such that
\begin{tight_enumerate}
\item
\label{part:PartitionConductance}
the conductance of $S_{i}$ (w.r.t. G) is at least $\phi^2$,
\item
\label{part:PartitionOverlap}
each vertex belongs to at most $O(\log{n})$ of the $S_i$s,
which in turn implies
\[
\sum_{i} \LL_{G\left[ S_i \right]}
\preceq
O\left( \log{n} \right)
\LL_{G}.
\]
\end{tight_enumerate}
\end{tight_enumerate}
\end{lemma}

Repeatedly running this then gives the overall algorithm,
whose pseudocode is in Algorithm~\ref{alg:SpectralSketch}.
Its guarantees come from repeatedly invoking
Lemma~\ref{lem:BetterError} on the pieces given by the expanders.

\begin{algorithm}
\caption{$\SpectralSketch
(G, \eps, \CycleDecomposition)$}

\textbf{Input}:
undirected graph $G$ with positive integer, poly bounded edge weights\\
Error $\eps > 0$\\
Cycle decomposition routine \CycleDecomposition.

\textbf{Output}:
Approximation $H$.

\begin{tight_enumerate}
\item Let $\alpha = 32\eps^{-1} L \gst(n)^{4} \log^{4} n$;
\item Decompose $G$ into $G_1, G_2 \ldots$
where all edge weights in $G_i$ are $2^{i}$.
\item Repeat $O(\log{n})$ times:
\begin{tight_enumerate}
\item Initialize $H$ as empty.
\item For each $G_{i}$,
\begin{tight_enumerate}
\item $\{\widehat{S}_{i1}, \widehat{S}_{i2}, \ldots \} \leftarrow
  \ExpanderDecompose\left(G_i, \frac{1}{2 \gst(n)}\right)$.
\item Add the edges between pieces to $H$, i.e.,
$H \leftarrow H \cup G_{i} \setminus \{\widehat{S_{i1}} \cup \widehat{S_{i2}} \ldots\}$.
\item For each piece $\widehat{S_{ij}}$
\begin{tight_enumerate}
\item $\tilde{S}_{ij} \leftarrow
\DecomposeAndSample(\widehat{S_{ij}}, \alpha,
\CycleDecomposition)$
\item Add $\tilde{S}_{ij}$ to $H$
\end{tight_enumerate}
\end{tight_enumerate}
\item Re-form the $G_i$s from the edges of $H_i$.
\end{tight_enumerate}
\item Return $G$.
\end{tight_enumerate}

\label{alg:SpectralSketch}
\end{algorithm}

\begin{proof}\emph{(of Theorem~\ref{thm:SpectralSketches},
Parts~\ref{part:Sparsify} and ~\ref{part:Forward})}
We first bound the behavior of each iteration of the inner loop in \SpectralSketch.

For the spectral approximation guarantees between $G$ and $H$, 
Cheeger's inequality (Lemma~\ref{lem:Cheeger}) gives
\[
\LL_{G\left[ S_i \right]}
\succeq
\frac{1}{2} \phi^{4}
\PPi_{\perp \one_{S}} \DD_{S} \PPi_{\perp \one_{S}},
\]
which means that the effective resistance between
two vertices $u$ and $v$ in $S_{i}$ is at most
\[
\frac{2 \phi^{-4} }{\min\left\{\dd_{u}, \dd_{v}\right\}}.
\]
As the only edges sampled in \DecomposeAndSample
are the ones with degree at least $\alpha$, the resistances of these
edges are bounded by
\[
2\phi^{-4} \alpha^{-1}
\leq
2 \left( 2 \gst\left( n \right)\right)^4
\frac{\epsilon }{32 L \gst\left( n \right)^4 \log^{4}n}
\leq
\frac{\epsilon}{L \log^{4}{n}}.
\]
As the cycle lengths are at most $L$, 
matrix concentration as stated in Lemma~\ref{lem:GraphSampling} gives
\begin{equation}
  \label{eq:spectral-sketch:sparsifier}
\LL_{G} \approx_{O(\sqrt{\epsilon} / \log{n})} \LL_{H}
\end{equation}
with high probability.

For the error in quadratic form,
on each of the pieces $\Shat_{ij}$,
combining Lemma~\ref{lem:BetterError}
\[
\hat{x}
:=
\frac{
\sum_{u \in S_{ij}} \dd_{u} \xx_{u}
}{
\sum_{v \in S_{ij}} \dd_{v}
},
\]
and Lemma~\ref{lem:CheegerProjected} gives that, with high probability, we have
\[
\abs{\xx^{\top} \LL_{\hat{S}_{ij}} \xx
-
\xx^{\top}\LL_{\tilde{S}_{ij}} \xx}
\leq
2\alpha^{-1} \sqrt{L} \log{n} \cdot \phi^{-4} \cdot
  \xx^{\top} \LL_{S_{ij}} \xx,
\]
where $\phi^2$ is lower bound of the conductance for $S_{ij}$.

Combining these terms, and invoking the guarantees about
the overlaps of $S_{ij}$ in Part~\ref{part:PartitionOverlap}
of Lemma~\ref{lem:ExpanderDecomposition} gives that the total
error in sampling $G_i$ is at most
\[
32\alpha^{-1} \sqrt{L} \log^2{n} \cdot \gst\left(n\right)^{4} 
\cdot
\xx^{\top} \LL_{G_i} \xx.
\]

Here, $G_i$ refers to the $G_i$ variables found in algorithm \SpectralSketch.
Since we have $O(\log{n})$ outer iterations and each of them incurs 
such a multiplicative error, the total error for the whole algorithm 
is upper bounded by
\[
32\alpha^{-1} \sqrt{L} \log^3{n} \cdot \gst\left(n\right)^{4} 
\cdot
\xx^{\top} \LL_{G_i} \xx.
\]
The choice of $\alpha$ then gives an error of at most
$(\epsilon / \log{n}) \xx^{\top} \LL_{G_{i}} \xx  $.
This bound is not tight, however, our choice of $\alpha$ is
constrained by the need for $H$ to be an $\frac{\sqrt{\eps}}{\log n}$
sparsifier for $G$ (Equation~\eqref{eq:spectral-sketch:sparsifier}).
Also, note that the guarantees of {\DecomposeAndSample} means that we
only add edges whose weights are powers of
$2$ back.  So after each step
$H$ can be decomposed back to the
$G_{i}$s without any increases in edge counts.

Now, let $m' \defeq \sum_i |E(G_i)|$.  The choice of $\phi =
\frac{1}{2 \gst(n)}$ for
$\ExpanderDecompose$ means that at least
$m'/2$ edges are contained in the $\Shat_{ij}$s.  As the
$\Shat_{ij}$s are vertex-disjoint, and at least half of the edges went
into them, we have that $m'$ edges get reduced to at most
\[
  \tilde{O}\left(\mhat + n  \eps^{-1} L \right)
+
\frac{7}{8} m'
\]
edges after one iteration.
Treating the first term as a function of $n$, $m(n)$, we get
that as long as $m' > 10 m(n)$, the edge count decreases by a
constant factor after each step.
So the $O(\log{n})$ outer iterations suffices for bringing the
edge count to $O(m(n))$, and we obtain the approximation guarantees
by taking this increase in errors into account.

Now, we'd like to bound the runtime of {\SpectralSketch}. To do this, we bound
the runtime of a single iteration of the outer loop in that algorithm.
Each iteration consists of a call to {\ExpanderDecompose} on each $G_i$, and a
call to {\CycleDecomposition} on each $\widehat{S_{ij}}$. Therefore, the runtime of
each iteration of the loop is upper bounded by:

\[
  \sum_i \tilde{O}\left(|E(G_i)| (2\gst(n))^2 \right) + \sum_{ij}
  \TCycle
  \left(
  \left|E\left(\widehat{S}_{ij}\right)\right|
  , 
  \left|V \left(\widehat{S}_{ij}\right)\right|
  \right). \]
Note that 
\[\sum_{ij} \left|E\left(\widehat{S}_{ij}\right)\right| \leq \sum_i |E(G_i)| = m' \]
and 
\[|V(\widehat{S}_{ij})| \leq n. \]
Recall that $\gst(n)$ is upper bounded by $\log^{O(1)}(n)$, by
Lemma~\ref{lem:ExpanderDecomposition}. Since $\TCycle(\cdot, n)$ is super-additive 
by the assumption in Equation~\ref{assumption:super-additive-T}, and $\TCycle$ is monotonic
in $n$,
we can bound the runtime of each iteration in the loop by:

\[
  \tilde{O}\left( \sum _i |E(G_i)| \right) + 
  \TCycle
  \left(
  \sum_{ij} \left|E\left(\widehat{S}_{ij}\right)\right|
  , n \right)
   \]
\[
  \leq \tilde{O}(m') + \TCycle\left(m',n \right)
\]

Now it remains to sum this quantity over all iterations of the outer loop in {\SpectralSketch}.
Since $m'$ decreases geometrically, and the initial value of $m'$ is less than
$O(m \log n)$, we can once again use the superadditivity property of $T(\cdot,
n)$ to bound the runtime of the entire {\SpectralSketch} algorithm by:
\[
\Otil{m} + \TCycle\left( O(m\log n), n \right),
\]
as desired.
\end{proof}

The dependence of $\alpha$ on $L$ can be improved by a factor of $\sqrt{L}$
by using the operator version of matrix Chernoff that we will state
in Lemma~\ref{lem:SampleGraphGlobal}.
However, we omit this improvement in the current version.

\subsection{Converting Guarantees to on Inverses}
\label{subsec:SketchBackward}

We now turn our attention to the quadratic inverse form.
Here our proof is by a black-box combination of the $\epsilon$-error
guarantees on the quadratic form with a $\sqrt{\epsilon}$ guarantee
on matrix approximations.
\begin{lemma}
\label{lem:Invert}
Suppose $\PP$ and $\QQ$ are matrices, and $\xx$ is a vector such that
for some $\eps \in (0,0.1]$ we have:
\begin{tight_enumerate}
\item
\label{given:operator}
$\PP$ and $\QQ$ $\sqrt{\eps}$-approximate each
other spectrally $\PP \approx_{\sqrt{\eps}} \QQ,$ and
\item
\label{given:quadratic-form}
The quadratic forms of $\PP^{+} \xx$ under $\PP$ and $\QQ$
$\eps$-approximate each other:
(note that $\xx^{\top} \PP^{+} \PP \PP^{+} \xx$ simplifies
to $\xx^{\top} \PP^{+} \xx$.)
\[
\xx^{\top} \PP^{+} \xx
\approx_{\eps}
\left( \PP^{+}\xx \right) ^{\top} \QQ \PP^{+} \xx,
\]
\end{tight_enumerate}
Then, we have,
$\xx^{\top} \QQ^{+} \xx
\approx_{7 \epsilon}
\xx^{\top} \PP^{+} \xx.$
\end{lemma}

\begin{proof}
  We will show that $\xx^{\top} \QQ^{+} \xx$ is
close to
\[
\xx^{\top} \left( 2 \PP^{+} - \PP^{+} \QQ \PP^{+} \right) \xx,
\]
which will in turn enable us to incorporate the condition on
$\xx^{\top} \PP^{+} \QQ \PP^{+} \xx$ being close to $\xx^{\top} \PP^{+} \xx$.
This  holds even in the matrix setting.
The condition of
\[
\PP
\approx_{\sqrt{\eps}}
\QQ
\]
by the preservation of approximations under pseudoinverses
stated in Fact~\ref{fact:spectral-error-invert} implies
\[
\QQ^{+}
\approx_{\sqrt{\eps}}
\PP^{+}.
  \]
Composing by $\QQ^{1/2}$ on both sides gives that
all the eigenvalues of
\[
\PPi - \QQ^{1/2} \PP^{+} \QQ^{ 1/2}
\]
are in the range $[e^{-\sqrt{\eps}}-1, e^{\sqrt{\eps}}-1]$,
where $\PPi$ is the projection matrix
  onto the column space of $\PP$ and $\QQ$.

Squaring this then gives that all the eigenvalues of
\[
\left( \PPi - \QQ^{1/2} \PP^{+} \QQ^{ 1/2} \right)^{2}
\]
are non-negative and no more than $\max\left((e^{-\sqrt{\eps}}-1)^2, (e^{\sqrt{\eps}}-1)^2\right)\le (e^{\sqrt{\eps}}-1)^2 \le 4\eps$,
where the last inequality follows from the fact that $e^{x} \leq 1 + 2x$ for $x \leq 1$.
This, when expanded becomes
\[
0
\preceq
\PPi - 2 \QQ^{1/2} \PP^{+} \QQ^{1/2}
+ \QQ^{1/2} \PP^{+} \QQ \PP^{+} \QQ^{1/2}
\preceq
4\eps \PPi.
\]
Moving the $\PPi$ terms outside, composing both sides by $\QQ^{+1/2},$
and flipping signs gives,
\[
\QQ^{+}
\succeq
2  \PP^{+}
- \PP^{+} \QQ \PP^{+}
\succeq
\left(1 - 4\eps\right) \QQ^{+}
\succeq
e^{-5\eps} \QQ^{+}.
\]
Here the last inequality utilizes the assumption of $\epsilon \leq 0.1$.
Substituting in the vector $\xx$ then gives:
\begin{equation}
  \label{eq:spectral-sketch:inverse-form}
e^{-5{\eps}}  \xx^{\top} \QQ^{+} \xx
\leq
\xx^{\top} \left( 2 \PP^{+} - \PP^{+} \QQ \PP^{+} \right) \xx
\leq
\xx^{\top} \QQ^{+} \xx.
\end{equation}
On the other hand, Assumption~\ref{given:quadratic-form} on
the quadratic forms involving $\xx$, specifically
$\xx^{\top} \PP^{+} \xx \approx_{\epsilon} \xx^{\top} \PP^{+} \QQ \PP^{+} \xx$ gives
\[
-e^{\epsilon} \xx^{\top} \PP^{+} \xx
\leq
-\xx^{\top} \PP^{+} \QQ \PP^{+} \xx
\leq
-e^{-\epsilon} \xx^{\top} \PP^{+} \xx
\]
to which we add $2 \xx^{\top} \PP^{+} \xx$ to both sides to obtain
\[
\left( 2 -e^{\epsilon} \right) \xx^{\top} \PP^{+} \xx
\leq
2 \xx^{\top} \PP^{+} \xx
-\xx^{\top} \PP^{+} \QQ \PP^{+} \xx
\leq
\left( 2 -e^{-\epsilon} \right) \xx^{\top} \PP^{+} \xx.
\]
Simplifying this again using the assumption of $\epsilon \leq 0.1$ then gives
\[
e^{-2\eps} \xx^{\top} \PP^{+} \xx
\leq
2 \xx^{\top} \PP^{+} \xx
- \xx^{\top} \PP^{+} \QQ \PP^{+} \xx
\leq
e^{2\eps} \xx^{\top} \PP^{+} \xx.
\]
Combining this with Equation~\eqref{eq:spectral-sketch:inverse-form},
gives
$\xx^{\top} \PP^{+} \xx \approx_{7 \epsilon} \xx^{\top} \QQ^{+} \xx$.
\end{proof}

The inverse quadratic form bound from
Theorem~\ref{thm:SpectralSketches} Part~\ref{part:Backward}
then follows from the spectral approximation guarantees
and quadratic form guarantees from Parts~\ref{part:Sparsify}
and~\ref{part:Forward}, with a suitable constant factor change in $\epsilon$.

A direct corollary to this theorem is that our graph sketches preserve
effective resistances. Taking a union bound over all $n^{2}$ vectors
$\chi_{uv},$ we obtain Corollary~\ref{cor:resistance-sparsifiers}.

%%% Local Variables:
%%% mode: latex
%%% TeX-master: "main"
%%% End:

%% file: resistance-computation.tex
%!TEX root = main.tex

\section{Computing Effective Resistances with Better
  $\eps$-Dependency}
\label{sec:Computing-Resistances}
In this section, we give an algorithm for computing the effective
resistances of all edges in $m^{1 + o(1)}\eps^{-1.5}$ time, proving
Theorem~\ref{thm:MainER}.

In Section~\ref{subsec:ERReductions}, we will show that computing
effective resistances reduces to sparsifying weighted cliques and bicliques,
and further that sparsifying these implicitly reduces
(with a polylog overhead) to
sparsifying a union of unit weighted bipartite cliques whose sizes are
powers of $2$, $\mathcal{K}_{B, 2^{i}}$.  We also show in
Section~\ref{subsec:kekeke} that such collections of bicliques
interact well with expander decompositions, up to polylog factors.

On the other hand, it's now far harder to require
such collections of bicliques to be simple: checking whether
a collection of $O(\log{n})$ cliques on $\Theta(n)$ vertices
is simple is in fact equivalent to the orthogonal
vectors problem~\cite{FahrbachMPSWX17:arxiv}.
This is problematic for the variance analysis in,
Section~\ref{subsec:LessError}. The analysis, like a similar one
from Jambulapati and Sidford~\cite{JambulapatiS18},
relies on the graph being simple.

However, note that if we are handling $\Kcal_{B, r}$, a collection
of balanced bicliques with size $r$, we can use the value of $r$
to bound the multiplicities of edges in $G(\Kcal_{B, r})$.  A vertex
with degree $d$ in this graph is involved in $d / r$ bicliques, and as
a result, each edge has multiplicity at most $d / r$.  This drop is
sufficient for a reduction of running time by $\eps^{-0.5}$;
we handle smaller bicliques by constructing them explicitly instead.

To obtain this running time, it is critical to access
implicit representations of dense graphs in time proportional
to the \emph{number of vertices} involved.
In particular, bipartite cliques are the most convenient
intermediate states for our algorithms because it is easy
to sample them in a degree-preserving manner: a matching suffices.
\begin{definition}
\label{def:bicliques}
We will use $\mathcal{K}$ to denote a collection of cliques
and bicliques.
In particular:
\begin{tight_itemize}
\item $\mathcal{K}_{B}$ denotes a collection of unit weighted
bipartite cliques.
\item $\mathcal{K}_{B, =}$ denotes a collection of
unit weighted balanced bicliques, that is, each biclique
has the same number of vertices on each side.
\item $\mathcal{K}_{B, r}$ denotes a collection of
unit weighted balanced bicliques with $r$ vertices on each side.
\end{tight_itemize}
\end{definition}

We use $G(\mathcal{K})$ to denote the explicit graph formed
by the union of the cliques in $\mathcal{K}$.  We use $n(\Kcal)$ to
denote the total number of vertices of all the cliques in the
collection, and $m(\Kcal)$ to denote the total edge count, taking
multiplicities into account.  Note that these values are not the same
as the total vertices $n$, and the total edges $m$, since the
bicliques are not disjoint.  Given such a collection of cliques, we
can compute the degrees of all the vertices, and in turn $m(\Kcal)$ by
summing together the sizes of the cliques that a vertex $u$ is
involved in.

A major issue that we need to address is that the sum of leverage
scores of a sampled matching (as required by Lemma~\ref{lem:GraphSampling})
can now be very large.
Specifically, suppose the graph is one copy of $K_{B, n}$,
and we sample it to $s$ matchings, each rescaled to $n / s$.
Then the leverage score of a single edge by symmetry is
about $\Theta(1/n)$, which means the total leverage score of
a matching is $\Theta(1)$.
With the rescaling factor of $n / s$, each sample then has a
seemingly prohibitive total leverage score of about $n / s$,
necessitating all $n$ samples.

Instead, we utilize a multi-edge version of graph sampling based
on the maximum magnitude of a sampled graph against another matrix.

\begin{lemma}
\label{lem:SampleGraphGlobal}
Let $G_1 \ldots G_{k}$ be distributions over random graphs
with expectation $G = \sum_{i} \expec{}{G_{i}}$ so that
for all $i$ we have
\[
\LL_{G_{i}}
\preceq
O\left( \frac{\eps^{2}}{ \log{n}} \right) \LL_{G},
\]
where $\eps \le 1$. Then with high probability we have $\sum_{i} G_{i} \approx_{\eps} G$.
\end{lemma}
The preceding lemma is a corollary of the Matrix Chernoff bound
from~\cite{Tropp12}.

\subsection{Sketching Unweighted bicliques}
\label{subsec:ERUnweightedBiClique}
We start with the simplest case: bicliques that all have the same
size, and are balanced, that is, have $r$ vertices on each side.
The pseudocode of our sampling scheme is in
Algorithm~\ref{alg:SampleMatchings}.
We will bound the convergence of this routine by bounding the
variance of each matching sampled.\begin{lemma}
\label{lem:BiCliqueVariance}
Let $K$ be a bipartite clique between $V_A$ and $V_B$
each of size $r$.
Let $H$ be a random matching on these vertices
with weight set to $r$.
Then for any vector $\xx$ and any value $\xhat$, we have
\[
\var_{H} \left[ \xx^{\top} \LL_{H} \xx \right]
\leq
\sum_{
a \in V_{A},
b \in V_{B}}
4r
\cdot
\left( \xx_{a} - \xhat \right)^2
\cdot
\left( \xx_{b} - \xhat \right)^2.
\]
\end{lemma}
Note that this would imply the variance incurred by
{\SampleMatchings} is $
\frac{4r}{s}
\sum_{a \in V_{a}, b \in V_{b}} ( \xx_{a} - \xhat)^2
( \xx_{b} - \xhat)^2$, as it averages over $s$ matchings.

\begin{algorithm}

\caption{$\SampleMatchings
(\mathcal{K}_{B, r}, s)$}

\textbf{Input}:
A collection of unit weighted bicliques, each of size $r$.
Sampling overhead $s$.

\textbf{Output}:
Approximation $H$.
\begin{tight_enumerate}
\item Initialize $H \leftarrow \emptyset$.
\item For each $K \in \mathcal{K}_{B, r}$
\begin{tight_enumerate}
\item Repeat $s$ times
\begin{tight_enumerate}
\item Add a random matching on the vertices of $K$
to $H$, with weight
$
\frac{r}{s}.
$
\end{tight_enumerate}
\end{tight_enumerate}
\item Return $H$
\end{tight_enumerate}

\label{alg:SampleMatchings}
\end{algorithm}

\begin{proof}
For simplicity, we consider the $\xhat = 0$ case first.
We will use the variables $a$, $a_1$, and $a_2$ to index
over vertices in $V_A$, and $b$, $b_1$, $b_2$ to index
over vertices in $V_B$.

By construction, $\expec{}{\LL_H} = \LL_{K}$.
Note that by the degree preserving property,
\[
\expec{H}{\xx^\top \LL_H \xx}
-
\xx^\top \LL_{K} \xx
=
\expec{H}{\xx^\top \AA_H \xx}
-
\xx^\top \AA_{K} \xx,
\]
where $\AA_H$ represents the adjacency matrix of $H$.
Therefore,
\[
\var_{H}\left[\xx^\top \LL_H \xx\right]
=
\var_{H}\left[\xx^\top \AA_H \xx\right].
\]

Let $X$ denote the random variable $\xx^\top \AA_H \xx$.
We will bound
$\var_{X}[X]
= \expec{X}{X^2} - \expec{X}{X}^2$
by expanding $\expec{X}{X^2}$ and $\expec{X}{X}^2$.
Let $\Mcal$ be the set of all matchings mapping $A$
to $B$, and $M$ be one such matching.
That is, $M(a)$ over all $a \in A$ gives a permutation in $B$.
We have
\begin{align*}
\expec{H}{X^2} 
& = r^2\left(\frac{1}{r!} \sum_{M \in \Mcal}
\left(\sum_{a} \xx_{a} 
    \xx_{M\left(a\right)}\right)^2\right)\\
& = \frac{r^2}{r!}\sum_{M \in \Mcal} \sum_{a} \xx_a^2
    \xx_{M\left(a\right)}^2
    +
    \frac{r^2}{r!}\sum_{M \in \Mcal} \sum_{a_1 \neq a_2}
    \xx_{a_1} \xx_{M\left(a_1\right)}
    \xx_{a_2} \xx_{M\left(a_2\right)}.
\end{align*}
For any $a$ and $b$,
there are $(r-1)!$ matchings that have $M(a) = b$,
so the first term simplifies to
\[
\frac{r^2}{r!} \sum_{M \in \Mcal} \sum_{a} \xx_{a}^2
    \xx_{M\left(a\right)}^2
=
r \cdot \sum_{a,b} \xx_{a}^2 \xx_{b}^2
\]
Also, for any $a_1 \neq a_2$, and $b_1 \neq b_2$,
there are $(r-2)!$ matchings that have $M(a_1) = b_1$
and  $M(a_2) = b_2$.
Therefore the second term simplifies to
\[
\frac{r^2}{r!}\sum_{M \in \Mcal} \sum_{a_1 \neq a_2}
    \xx_{a_1} \xx_{M\left(a_1\right)}
    \xx_{a_2} \xx_{M\left(a_2\right)}
=
\frac{r}{r - 1}
\sum_{a_1 \neq a_2, b_1 \neq b_2}
\xx_{a_1} \xx_{a_2} \xx_{b_1} \xx_{b_2}.
\]

On the other hand, the expectation of $X$ evaluates to
\[
\expec{H}{X}
= \sum_{a, b} \xx_a \xx_{b}
= \left( \sum_{a} \xx_a \right) \left( \sum_{b} \xx_{b} \right).
\]
Squaring this, and grouping the terms by duplicity
in $a$ gives:
\[
\expec{H}{X}^2
=
\sum_{a, b} \xx_{a}^2 \xx_{b}^2
+
\sum_{a, b_1 \neq b_2} \xx_{a}^2 \xx_{b_1} \xx_{b_2}
+
\sum_{a_1 \neq a_2,b} \xx_{a_1} \xx_{a_2} \xx_{b}^2
+
\sum_{a_1 \neq a_2, b_1 \neq b_2} \xx_{a_1} \xx_{a_2} \xx_{b_1} \xx_{b_2},
\]
and in turn:
\begin{align*}
\var_{H}\left[ X \right]
& = \expec{X}{X^2} - \expec{X}{X}^2\\
& = \left( r - 1 \right) \cdot \sum_{ab} \xx_{a}^2 \xx_{b}^2
-
\sum_{a, b_1 \neq b_2} \xx_{a}^2 \xx_{b_1} \xx_{b_2}
-
\sum_{a_1 \neq a_2,b} \xx_{a_1} \xx_{a_2} \xx_{b}^2
+ \frac{1}{r - 1}
\sum_{a_1 \neq a_2, b_1 \neq b_2}
\xx_{a_1} \xx_{a_2} \xx_{b_1} \xx_{b_2}.
\end{align*}
We bound each of these terms separately.
For the second term,
we have $\xx_{b_1}^{2} + \xx_{b_{2}}^{2} \ge
2|\xx_{b_{1}}\xx_{b_{2}}|$, so we have
\[
-
\sum_{a, b_1 \neq b_2} \xx_{a}^2 \xx_{b_1} \xx_{b_2}
\leq
\frac{1}{2} \sum_{a, b_1 \neq b_2} \xx_{a}^2 \left( \xx_{b_1}^2 + \xx_{b_2}^2 \right)
\leq
\left( r - 1\right) \sum_{a, b} \xx_{a}^2 \xx_{b}^2,
\]
and the third term follows similarly by applying the two-term
arithmetic-geometric mean inequality to $\xx_{a_1} \xx_{a_2}$.
We can also bound the last term by applying this inequality
simultaneously on both $a_1a_2$ and $b_1 b_2$:
\begin{align*}
\frac{1}{r - 1}
\sum_{a_1 \neq a_2, b_1 \neq b_2}
\xx_{a_1} \xx_{a_2} \xx_{b_1} \xx_{b_2}
& \leq
\frac{1}{4\left( r - 1\right)}
\sum_{a_1 \neq a_2, b_1 \neq b_2}
\left( \xx_{a_1}^2 + \xx_{a_2}^2 \right)
\left( \xx_{b_1}^2 + \xx_{b_2}^2 \right)\\
& \leq
\left( r  -1 \right)
\left( \sum_{a}\xx_{a}^2 \right)
\left( \sum_{b}\xx_{b}^2 \right).
\end{align*}
Summing across these four terms
gives the lemma statement for $\xhat = 0$.
The general lemma statement follows from observing that
\[
\xx^{\top} \LL_{H} \xx =
  \left(\xx - \xhat \one\right)^{\top}
  \LL_{H}
  \left(\xx - \xhat \one\right),
\]
where $\one$ is the all ones vector.
So $\xx$ can be replaced by $\xx - \xhat$ throughout this calculation.
\end{proof}

We now bound the overall variance when we sample
matchings for a collection of such bicliques.
Here, we also need a sampling overhead $s$.
We give the following relatively technical lemma involving a degree
threshold on the vertices that's a more limited
analog of Lemma~\ref{lem:BetterError}
in that it gives a w.h.p. bound on the errors of sampling
the quadratic form in terms of the degrees involved.

Because we now need to bound $\LL_{K}$ against $\LL_{G}$,
it is useful to work with the case where all vertices in $K$ have fairly
large degree.
To narrow down to this setting, we will partition the degrees
so that the degrees involved in each biclique boil down to a single
parameter.
We first show our sampling routine under this restriction,
and then discuss how to narrow down to this case by bucketing
the vertices in the more general case.
\begin{lemma}
\label{lem:SampleMatching}
Let $\mathcal{K}_{B, r}$ be a collection of $k$ unit weighted balanced
bicliques of size $r$ each,
and $V_{S}$ a set of vertices in $G(\mathcal{K}_{B, r})$ such that for
some $d$ and $\eps$ with $d \ge \eps^{-3/2}$ we have:
\begin{tight_enumerate}
\item Each vertex in $V_{S}$ has degree at most $d$
  in $G(\mathcal{K}_{B, r})$.
\item Each biclique $K \in \mathcal{K}_{B, r}$
with $V_{K} = \{A, B\}$
    has $A \subseteq V_{S}$
(the mirror case of $B \subseteq V_{S}$ is equivalent to this by
swapping the two sides).
\end{tight_enumerate}
The distribution over graphs
\[
H = \SampleMatchings\left(\mathcal{K}_{B, r},
\max\left\{\eps^{-\nfrac{1}{2}}, 4 r \eps^{-1} / d\right\}\right)
\]
gets sampled in $\Otil{n \eps^{-1} +  n(\Kcal) \eps^{-\nfrac{1}{2}}}$ time, and
satisfies the following:
\begin{tight_enumerate}
\item $H$ has
$\Otil{n\eps^{-1} + n(\Kcal) \eps^{-\nfrac{1}{2}}}$ edges.
\item $ \expec{}{H} = G( \mathcal{K}_{B, r} ) $.
\item Any graph $M$
(which is a rescaled matching)
in the support of $H$ satisfies
\[
 \LL_{M}
 \preceq
 \eps
 d \II_{\perp \one},
 \]
 where $\II_{S \perp \one}$ is the identity matrix on $S$ projected
 against the all $1$s vector.
\item For any fixed vector $\xx$ and
any $\xhat$, we have with high probability:
\[
\abs{\xx^{\top} \LL_{H} \xx
  - \xx^{\top} \LL_{G\left(\mathcal{K}_{B, r}\right)} \xx}
\leq
O\left( \log{n} \right)
\cdot \eps
\cdot \sum_{u} d \left( \xx_{u} - \xhat \right)^2.
\]
\end{tight_enumerate}
\end{lemma}

The condition involving $V_S$ is useful because we can only reduce
to collections of bicliques incident to small degree vertices,
and have less control on the degrees of the other side.
Also, the lower bound on edge sampling probabilities is necessary
for a proof of spectral approximations using Lemma~\ref{lem:GraphSampling}.
Such operator (instead of vector) approximations are necessary for
converting the guarantees from quadratic forms to quadratic inverse
forms (and hence resistance approximations) using Lemma~\ref{lem:Invert}.
\begin{proof}
We use $s$ to denote the sample count per clique
\[
s
:=
\max\left\{\eps^{-\nfrac{1}{2}}, 4 r \eps^{-1} / d\right\}.
\]

For the number of edges sampled, we can sum the two bounds,
one per choice of term.
\begin{tight_itemize}
\item Picking $\eps^{-\nfrac{1}{2}}$ matchings per biclique
gives an edge count that's $\eps^{-\nfrac{1}{2}}$ times the
size of each biclique, summing to a
total of $k r \eps^{-\nfrac{1}{2}} = (1/2) n(\calK)\eps^{-\nfrac{1}{2}}$.
\item The ${4r \eps^{-1}}/{d}$ matchings contributes to a total of
\[
\frac{4r^2 \eps^{-1}}{d}
\]
sampled edges per biclique.
On the other hand, the total number of edges $m(\Kcal)$ in $G(\Kcal)$
is at most $nd$. Since each biclique has $r^2$ edges,
the total number of bicliques is at most
\[
\frac{nd}{r^2},
\]
which multiplied by the number of matching edges sampled per biclique gives a total of
\[
\frac{nd}{r^2}
\cdot
\frac{4r^2 \eps^{-1}}{d}
\le
O\left(
n \eps^{-1}
\right).
\]
\end{tight_itemize}

To bound the expectation, consider a single biclique $K$.
Let $M_1,\ldots, M_s$ be the sampled matchings.
By symmetry, we have $\expec{}{M_{i}} = K$, so by linearity of
expectation we get $\expec{}{H} = G(\Kcal)$.

We also need to bound the maximum magnitude of a sample.
The rescaling factor of $r / s$ means that the value
of each sample picked in some clique $K$ (with a matching $M$) is at most
\[
\frac{r}{s} \xx^{\top} \LL_{M} \xx
=
\frac{r}{s}
\sum_{a \in V_{K,A}} \left( \xx_{a} - \xx_{M\left( a \right)} \right)^2
\leq
\frac{2r}{s}
\sum_{a \in V_{K,A}} \left( \xx_{a} - \xhat \right)^2
  + \left( \xx_{M\left( a\right) } - \xhat \right)^2
\leq
\frac{2r}{s} \sum_{u \in V_K} \left( \xx_{u} - \xhat \right)^2.
\]
As the definition of the projection operator $\II_{\perp \one}$ gives
\[
\xx^{\top} \II_{\perp \one} \xx
= \min_{\xhat} \sum_{u} \left( \xx_{u} - \xhat \right)^2,
\]
we get
\[
\frac{r}{s} \xx^{\top} \LL_{M} \xx
\leq
\frac{2r}{s} \xx^{\top} \II_{\perp \one} \xx,
\]
and incorporating the condition of
$s \geq {4 r \eps^{-1}}/{d}$
gives a bound of
\[
\frac{r}{s} \LL_{M}
\preceq
\frac{\eps d}{2} \II_{\perp \one}.
\]
As the size $r$ biclique itself is a sum of $r$ matchings, we can apply a similar bound
\[
\xx^{\top} \LL_{K} \xx
\leq
2r \sum_{u} \left( \xx_{u} - \xhat \right)^2
\]
for any scalar $\xhat$, or equivalently,
\[
\frac{1}{s} \LL_{K}
\preceq
\frac{\eps d}{2} \II_{\perp \one}.
\]
Thus by triangle inequality, we can bound the mean
zero random variable corresponding to the deviation incurred
by matching $M_i$ by
\[
\abs{\frac{r}{s}\xx^\top\LL_{M_i}\xx - \frac{1}{s} \xx^\top\LL_K\xx}
\leq
\eps d \sum_{u} \left( \xx_{u} - \xhat \right)^2
\]
for any scalar $\xhat$.

This gives the first condition on the magnitude of each random
variable required by Bernstein's inequality as
stated in Lemma \ref{lem:Bernstein}.
Specifically, summing over $M_i,$ the random variable corresponding to $K$
is $\frac{r}{s}\left(\sum_{i=1}^{s}\xx^\top\LL_{M_i}\xx\right)-\xx^\top\LL_K\xx.$
Since the expectation for any clique is $0$,
Lemma~\ref{lem:BiCliqueVariance} gives that the variance per clique is
at most
\[
\sum_{a \in V_{A}, b \in V_{B}} 
\frac{r}{s} \left( \xx_{a} - \xhat \right)^2
         \left( \xx_{b} - \xhat \right)^2.
\]
Furthermore, each vertex in $V_S$ can participate in at most $d / r$
such bicliques, as each of them incurs a degree of $r$.  Since each
edge has at least one endpoint in $V_S$, each edge occurs in at most
$d / r$ bicliques, and the total variance is at most
\[
\sum_{u, v}
\frac{d}{r} \cdot \frac{r}{s}
\cdot \left( \xx_{u} - \xhat \right)^2
         \left( \xx_{v} - \xhat \right)^2
=
\left( \sqrt{d / s} \sum_{u} \left(\xx_{u} - \xhat \right) ^2 \right)^2
\leq
  \left( \eps d \cdot \sum_{u} \left(\xx_{u} - \xhat \right) ^2 \right)^2,
\]
where we use $s \cdot d \ge \eps^{-\nfrac{1}{2}} \cdot
\eps^{-\nfrac{3}{2}} = \eps^{-2}.$
The bound on deviation then follows from Bernstein's inequality
as stated in Lemma~\ref{lem:Bernstein}.
\end{proof}

Note that $s \geq 4r\eps^{-1}/d$ was required only to lower bound the
probability that any edge was included in $H$, which in turn is needed to show
that $H$ is a $\sqrt{\eps}$ spectral sparsifier with high probability, something
we will do in Section~\ref{subsec:kekeke}. 
This sparsifier property is required so that the guarantees in
Section~\ref{subsec:SketchBackward} apply, and
we will use these guarantees later on, also in Section~\ref{subsec:kekeke}.
Note that we could have enforced a tighter constraint of
$s = \max\left(\eps^{-\nfrac{1}{2}}, \frac{r \eps^{-3/2}}{d}\right)$
and obtained a bound of $O(n \eps^{-3/2} + kr\eps^{-\nfrac{1}{2}})$ edges in $H$, which would allow us to carry the rest of the proof forward.
If we do this, $H$ will be a $\eps^{\nfrac{3}{4}}$ spectral sparsifier
(proof omitted).
However, this ultimately does not improve our final bound on the edge count in a resistance sparsifier, so we have not set $s$
accordingly.

  Lemma~\ref{lem:SampleMatching} guarantees that for a collection of
  balanced bicliques of equal size,  where one side of the biclique has
  uniformly low degree in the original graph, that the error
  incurred by sampling matchings from that biclique is bounded in
  terms of that
  uniformly low degree. 
 
We can then extend this procedure to general unit weighted bicliques
and degree sequences.  
The idea is to break the bicliques of $\mathcal{K}_B$
up into a collection of balanced bicliques $\mathcal{K}_{B, =}$ ,
where each edge $uv$ in the biclique has roughly the
same value for $\min(\dd_u, \dd_v)$ where $\dd_u$ represents the
degree of $u$ in $\mathcal{K}_B$.
This is done so that Lemma~\ref{lem:SampleMatching} can directly be
applied to bound the variance of $\mathcal{K}_{B,=}$, as
Lemma~\ref{lem:SampleMatching} provides good bounds on the variance in
terms of the min-degree of an endpoint for any given edge in
$\mathcal{K}_B$.
The rough overview of the process is:
\begin{tight_enumerate} 
  \item Incorporate vertices of low degree into $H$. This is
    done since the assumption on Lemma~\ref{lem:SampleMatching} is that
    $d \geq \eps^{-1.5}$.
  \item For each power of $2$, split up the bicliques as follows:
    \begin{tight_enumerate}
      \item For each biclique, find the set of edges with min-degree
        $2^{j-1} \leq d < 2^j$. This will be the disjoint union of two
        bicliques.
      \item Balance the bicliques using \MakeBalanced.
      \item On the resulting bicliques: sparsify using
        {\SampleMatchings}, with the appropriate value of $s$ as
        specified in Lemma~\ref{lem:SampleMatching}.
    \end{tight_enumerate}
\end{tight_enumerate}
The pseudocode of this routine is given in
Algorithm~\ref{alg:SampleBiCliques}.
It  transforms a collection of bicliques into a slightly larger
(in total vertex count) collection of bicliques that are balanced,
and which have the property that one side of the biclique has
small degrees (in the original graph).
This then allows us to invoke {\SampleMatchings} on these
bicliques with $d$ set to these degrees.

\begin{algorithm}
\caption{$\SampleBiCliques(\mathcal{K}_{B}, \eps)$}
\label{alg:SampleBiCliques}

\textbf{Input}:
A collection of $k$ unit weighted bicliques $\mathcal{K}_{B}$.\\
Error tolerance $\eps$.

\textbf{Output}:
Sketch $H$.

\begin{tight_enumerate}
\item Compute the degrees of
$G(\mathcal{K}_{B})$, let $\dd$ be the degree sequence.

\item Initialize $H$ to be empty.

\item For each vertex $u$ with $\dd_{u} \leq \eps^{-1.5}$,
\begin{tight_enumerate}
\item Add all edges incident to $u$ explicitly to $H$.
\item Remove $u$ from all the entries in $\Kcal_{B}$.
\end{tight_enumerate}

\item  Let $V_{j}$ be the vertices with degrees $\dd_{u}$
in the range $[2^{j - 1}, 2^{j} - 1]$.

\item For each $V_{j}$
\begin{tight_enumerate}
  \item Initialize $\mathcal{K}_{B}(j)$ to be empty.
  \item  (Break up each biclique into two bicliques such that one side of the
	biclique has degrees between $2^{j-1}$ and $2^j$, and the
	other side has degrees higher than $2^{j-1}$).
For each $K \in \mathcal{K}_{B}$
\begin{tight_enumerate}
  \item Let $K_A$ and $K_B$ denote the two sides of biclique
  $K$. Let
    \begin{align*}
      & S_A := K_A \cap V_j,
      \\
      & S_B := K_B \cap V_j,
      \\
      & T_A := K_A \cap (V_{j} \cup V_{j+1} \cup \ldots),
      \\
      & T_B := K_B \cap (V_{j} \cup V_{j+1} \cup \ldots).
    \end{align*}
    Pass $S_A \times T_B$, $S_B \times (T_A \setminus S_A)$
    into to $\mathcal{K}_{B}(j)$ implicitly (as subsets of vertices).
    These bicliques we pass are exactly the edges 
    $(u,v)$ of $G(\Kcal)$ with $2^{j-1} \leq \min(\dd_u, \dd_v) < 2^j$.
\label{step:SplitBasedOnDegrees}
\end{tight_enumerate}

\item For each $K \in \mathcal{K}_B(j)$,
run $\MakeBalanced(B)$, and
add the resulting bicliques to $\{\mathcal{K}_{B, 1}(j),
    \mathcal{K}_{B, 2}(j), \mathcal{K}_{B, 4}(j), \mathcal{K}_{B,
    8}(j) \ldots \}$. In this step, a balanced biclique with $r$ vertices
    on each side (for any $r$) is added only to $\mathcal{K}_{B,r}(j)$.
    (Note that each biclique resulting from this has $r = 2^{l}$ vertices
on each side, and one side has degrees between $2^{j-1}$ and $2^j$
in the original graph while the other has degrees at least $2^{j - 1}$.)

\item (Bucket $\mathcal{K}_{B,r}(j)$ by $r$ and run
  {\SampleMatchings} on the buckets.) For each $r = 2^{l}$
\begin{tight_enumerate}
\item If $r \leq \eps^{-\nfrac{1}{2}}$
\begin{tight_enumerate}
\item Add all edges from $\{\mathcal{K}_{B, r}(j)\}$
 to $H$ explicitly.
\end{tight_enumerate}
\item Else
\begin{tight_enumerate}
\item Let $
H_{j, l}
\leftarrow 
\SampleMatchings\left(\mathcal{K}_{B, r}(j), 
   \max\left\{\eps^{-\nfrac{1}{2}}, 4 r \eps^{-1} 2^{-j} \right\}
\right)$.
\item $H \leftarrow H \cup H_{j, l}$.
\end{tight_enumerate}

\end{tight_enumerate}

\end{tight_enumerate}

\item Return $H$.
\end{tight_enumerate}

\end{algorithm}

We first break the vertices based on their degrees,
decompose a biclique into balanced bicliques.
The formal guarantees of this step is:
\begin{lemma}
\label{lem:Balancing}
There is a routine $\MakeBalanced(G)$ that breaks
any biclique on $n$ vertices into a sum of balanced
bicliques with $2^{i}$ vertices whose total vertex count
is $O(n \log{n})$.
\end{lemma}

\begin{proof}
This decomposition has two steps: the first is to decompose $G$
into bicliques whose vertex sizes are powers of $2$, and the
second is to make them balanced.

For the first step, suppose the two sides have $n_1$ and $n_2$
vertices respectively.
We can write $n_1$ and $n_2$ both as sums of powers of $2$, and split the vertices into groups correspondingly. We create one biclique
between each pair of powers of $2$ present in $n_1$ and $n_2$.
Since there are only $O(\log{n})$ such groups on one side, one vertex can appear in $O(\log n)$ new bicliques. This increases
the total number of vertices by a factor of $O(\log{n})$.

Then in order to balance the bicliques, note that a
$2^{i} \times 2^{j}$ biclique with $i < j$ is the sum
of $2^{j - i}$ balanced bicliques with $2^{i}$ vertices
on each side.
This at most doubles the number of vertices, so overall we still
have $O(n \log{n})$ vertices.
\end{proof}

With this guarantee, we can now prove the overall
guarantees for an arbitrary degree sequence.
Pseudocode of the algorithm is in Algorithm~\ref{alg:SampleBiCliques}.
Note that we do not need to pick a good bipartition
as we're already working with bicliques.
\begin{lemma}
\label{lem:SampleBiCliques}
Let $\Kcal_{B}$ be a collection of unit weighted bicliques on $n$ vertices
whose sum $G(\mathcal{K}_{B})$ has degree sequence $\dd$.
Then $\SampleBiCliques(\mathcal{K}_{B}, \eps)$
returns  in $\Otil{n\eps^{-1.5} + n(\Kcal_{B}) \eps^{-0.5}}$
time, a graph $H$ sampled from a sum of independent random graphs such that:
\begin{tight_enumerate}
\item $\expec{}{H} = G(\mathcal{K}_{B})$;
\item
\label{claim:bounded-by-diagonal}
  any graph $\Hhat$ in the support of $H$ satisfies
\[
\LL_{\Hhat}
\preceq
\eps \DD_{G\left(\Kcal_{B}\right) \perp \one}
\]
where $\DD_{G(\Kcal_{B}) \perp \one}$ is the diagonal matrix projected
against the all $1$s vector;
\item the number of edges in $H$ is bounded by
\[\textstyle
\Otil{n \eps^{-1.5}
  + n\left( \mathcal{K}_B \right) \eps^{-0.5}};
\]
where the trailing term is the total number of vertices in this
collection of cliques, times an overhead of $\eps^{-0.5} \log^{O(1)}n$;
\item for any vector $\xx$ and any value $\xhat$, we have w.h.p.
\[
\abs{\xx^{\top} \LL_{ G\left(\mathcal{K}_{B}\right) } \xx
   - \xx^{\top} \LL_{H} \xx}
\leq
O\left(\eps \log^{2}{n} \right)
\cdot \sum_{u} \dd_{u} \left( \xx_{u} - \xhat \right)^2.
\]
\end{tight_enumerate}
\end{lemma}

\begin{proof}
We start from the claim about the concentration of $\xx^\top L_H \xx$.
  Let $uv$ be an edge in $\mathcal{K}_B$, and without loss of generality,
  let $\dd_u \leq \dd_v$.
  In $\SampleBiCliques$, edge $uv$ is placed into
  $\mathcal{K}_{B} \left(\lfloor \log \dd_u \rfloor\right)$, the bucket
  of cliques corresponding to $\dd_u$ in Step~\ref{step:SplitBasedOnDegrees}.

Therefore, by Lemma~\ref{lem:SampleMatching}, 
the output of the calls to $\SampleBiCliques$,
$H_{j, l}$ satisfies
\[
\expec{}{H_{i, l}} = G\left(\Kcal_{B, 2^{l}}\right),
\]
any graph $\Hhat$ in its support satisfies
\[
\LL_{\Hhat}
\preceq
\eps \cdot 2^{j} \cdot \II_{V\left( \mathcal{K}_{B, 2^{l}}  \left(j\right)\right) \perp \one},
\]
where $\II_{V( \mathcal{K}_{B, 2^{l}}  (j)) \perp \one}$ is the identity
matrix on $V( \mathcal{K}_{B, 2^{l}}  (j))$ projected against the all ones vector;
and for any $\xx$ and $\xhat$, we have with high probability:
\[
\abs{ \xx^\top \LL_{H_{j, l}} \xx - \xx^\top \LL_{\mathcal{K}_{B, 2^{l}} \left( j \right)} \xx }
\leq O( \log n) \cdot \eps \cdot
\sum_{u \in V\left( \mathcal{K}_{B, 2^{l}}  \left(j\right)\right)} 2^j
    \left(\xx_u - \widehat{x}\right)^2.
\]
  
The first condition with linearity of expectation implies
$\expec{}{H} = G(\Kcal_{B})$.

To convert the second and third condition to global ones
encompassing all $H_{j, l}$s, the
The key point here is that our construction was
designed to guarantee that $u$ is not contained in any
  $\mathcal{K}_B(j)$ with $2^{j-1} > \dd_u$.
 For the second condition, this gives that
 all vertices in $\mathcal{K}_{B} (j)$ have degrees at
least $2^{j - 1}$, so
\[
2^{j} \cdot \II_{V\left( \mathcal{K}_{B, 2^{l}}  \left(j\right)\right) \perp \one}
\preceq
\DD_{V\left( \mathcal{K}_{B, 2^{l}}  \left(j\right)\right) \perp \one}
\preceq
\DD_{G\left( \Kcal_{B} \right) \perp \one},
\]
which bounds the sizes of the samples against the overall degree
sequences.

For the third condition, the containment condition formalizes to
\[
\sum_{j: u \in \mathcal{K}_B\left(j\right) } 2^{j}
\leq \sum_{2^{j-1} \leq \dd_u} 2^j
\leq 4 \dd_u,
\]
which coupled with the fact that there are only $O(\log{n})$ distinct
values of $r = 2^{l}$ gives:
\[
\abs{ \xx^\top \LL_{H_{j, l}} \xx - \xx^\top \LL_{\mathcal{K}_{B, 2^{l}} \left( j \right)} \xx }
\leq
O\left(\log^{2} n  \cdot \eps \sum_u \dd_u \left(\xx_u-\xhat\right)^2\right),
\]
which proves concentration of $\xx^\top \LL_H \xx$.

Now we move to bound the number of edges in $H$.
Directly applying Lemma~\ref{lem:SampleMatching} gives that the
number of edges in $H$ is always bounded above by:
\[
\sum_j \sum_{l}
O\left(n \eps^{-1.5} +n\left( \Kcal_{B, 2^{l}}\left(j\right) \right) 
\eps^{-\nfrac{1}{2}}\right).
\]
For the total number of vertices,
total number of vertices in $K_B(j)$ summed across all $j$ is $O(n\log n)$,
and the number of vertices in $K_{B, 2^i}(j)$ summed across all $i$
and $j$ is $O(n \log n)$ by Lemma~\ref{lem:Balancing}.
So the first term is bounded by $O(n \log^2{n})$.

For the second term,
$\sum_{j, l} n(\Kcal_{j, l})$ tracks the sum of vertex counts
across all cliques in $K_{B, 2^l}(j)$ for any $j$ and $l$.
Note that each clique in $K_{B, 2^i}(j)$ has a corresponding clique $K
  \in \mathcal{K}_B$ that it originated from and is a subgraph of.
Edges from each clique in $K$ can go to at most
$O(\log^2{n})$ different values of $j$ and $l$, so therefore
\[
\sum_{j, l} n \left( \Kcal_{B, 2^{l} } \left( j \right) \right)
\leq \sum_{K \in \Kcal_B} O\left( \log^2{n} \right) n\left( K \right)
\leq O\left( \log^2 n \right) n\left( \Kcal \right)
\]
and we get our bound on edge count.
\end{proof}

\subsection{Incorporating Recursive Expander Decompositions}
\label{subsec:kekeke}
We now incorporate expander partitioning in a way
analogous to Section~\ref{subsec:SpectralSketches}.
The main result that we shall prove in this section is:
\begin{lemma}
\label{lem:SketchUnweightedBiCliques}
Given a collection of unit weighted bicliques $\Kcal$, 
and any error $\eps$,
invoking the routine
$\ImplicitSketchUnweightedBiCliques(\Kcal, \epshat, \phi, q)$
as shown in Algorithm~\ref{alg:SketchUnweightedBiCliques}
with
\begin{align*}
q & = O\left(\sqrt{ \frac{\log{n}}{\log\log{n}}} \right), \\
\phi & = n^{-2/q}
= \exp\left( -O \left( \sqrt{\log{n} \log\log{n} } \right) \right),\\
\epshat & = \frac{\phi^{4}}{O\left(1\right) \cdot \log^{4}n \cdot q} \eps,
\end{align*}
produces in time
\[
n^{1 + o\left(1\right)} \eps^{-\nfrac{3}{2}} 
     + n\left( \Kcal \right) n^{o\left(1\right)} \eps^{-\nfrac{1}{2}}
\]
a graph $H$ such that:
\begin{tight_enumerate}
\item $H$ has
$n^{1 + o(1)} \eps^{-1.5} + n(\Kcal) n^{o(1)} \eps^{-0.5}$
edges.
\item With high probability,
$H$ is a $\sqrt{\eps}$-approximation to $G(\Kcal)$,
\emph{i.e.}, $\LL_H \approx_{\sqrt{\eps}} \LL_{G(\Kcal)}$.
\item For any fixed vector $\xx \in \rea^{V}$, with high probability
  we have
$\xx^{\top} \LL_{G(\Kcal)} \xx
  \approx_{\eps}
\xx^{\top} \LL_{H} \xx.$
\end{tight_enumerate}
\end{lemma}

Note that the constants in the base of $\phi$, and in turn $q$,
depends on the exponent of $\log{n}$ in $\gst(n)$, the overhead
from the expander decomposition routine given in Lemma~\ref{lem:ExpanderDecomposition}.

Expander decomposition as stated in
Lemma~\ref{lem:ExpanderDecomposition} works only on an
explicitly specified graph.
Our goal is to run in time almost-linear
in the number of vertices of these bicliques, instead
of the number of edges.
Instead, we will compute such a decomposition on a spectral
sparsifier of the sum of bicliques, and transfer the errors
using the spectral guarantees.
Specifically, we invoke the following fact which plays a central
role in previous results on implicitly sparsifying
graphs~\cite{PengS14,KyngLPSS16,CohenKPPRSV17,ChengCLPT15}:
\begin{lemma}
\label{lem:MatchingSparsifySimple}
There exists an absolute constant $C_{S}$ depending on the (w.h.p.)
success probability such that for any balanced biclique $K_{B, r}$
with $r$ vertices on each side,
the union of $s$ random matchings for $s = C_{S} \log{n}$,
\[
\Gtil
=
\SampleMatchings\left(
  \left\{
    K_{B, r}
  \right\},
  s
\right)
\]
has the same (weighted) degrees as $K_{B, r}$,
and satisfies with high probability:
\[
\LL_{G\left( K_{B, r} \right)}
\approx_{2}
\LL_{\Gtil}.
\]
\end{lemma}
Running this on each biclique
reduces the number of edges to $O(n(\Kcal_{B, r}) \log{n})$,
instead of edges, involved in the
expander representation given by $\Kcal_{B, r}$.
This is small enough to allow us to explicitly run expander decompositions
on this graph, which also has the same weights on each edge.
Note that here we only invoke {\SampleMatchings} in the constant
error regime, so do not need to perform a more careful analysis
of its $\eps$ dependencies.
This approximation between $\Gtil$ and  $G(\Kcal)$ is useful in two ways:
\begin{tight_enumerate}
\item 
Any partition of vertices cuts a similar weight of edges in the two graphs.
This means the process of sparsifying the edges within clusters,
and then repeating on edges between clusters makes sufficient
progress to terminate in a few rounds.
\item The sum of $\dd_{u} (\xx_{u} - \xhat)^2$ is the same in the two graphs.
Combining this with the error guarantees of quadratic forms
means that we can still transform the degree-dependent guarantees from
Lemma~\ref{lem:SampleBiCliques} to guarantees involving the Laplacian
quadratic form in a way analogous to the proof of
Theorem~\ref{thm:SpectralSketches} Part~\ref{part:Forward}).
\end{tight_enumerate}

Note in particular for the second condition, it is not the case that an expander 
subgraph in a sparsifier ($\Gtil$) is an expander in the original graph ($G(\Kcal)$).
Consider for example adding/removing an expander on a small subset
of vertices that would be otherwise independent in a larger expander.
Instead, we need to bound the overall costs in degrees against
the quadratic form of the sparsifier, and transform the costs
over via the overall quadratic form.

Note that in order to work with a unit weighted graph for
expander partitioning, we need to restrict to balanced
bicliques with size $r$ once more.

The pseudocode of the algorithm is divided into two parts:
Algorithm~\ref{alg:ImplicitPartitionAndSample} performs
the partition, and sparsifies the edges within the clusters.
The outer loop in Algorithm~\ref{alg:SketchUnweightedBiCliques} then creates
another representation by bicliques of the edges between the pieces,
and recurses upon them.

\begin{algorithm}

\caption{$\ImplicitPartitionAndSample(\Kcal_{B, r}, \eps, \phi)$}
\label{alg:ImplicitPartitionAndSample}

\underline{Input}:
Collection of balanced bicliques $\Kcal_{B, r}$.
Error threshold $\eps$.
Conductance $\phi$.

\underline{Output}:
Partition $\{ \Shat_1, \Shat_2, \ldots\}$ of the vertices,
sketch $H$ of the edges of $\Kcal_{B, r}$ contained in
the partitions $\Shat_1, \Shat_2, \ldots$.

\begin{tight_enumerate}

\item Set $s  = C_{S} \log{n}$,
and build $\Gtil = \sum_{K_{B, r} \in \Kcal_{B, r}}
{\SampleMatchings}(\{K_{B, r}\},  s)$ explicitly.
\label{line:CrudeSparsify}

\item $\{ \Shat_1, \Shat_2, \ldots\}
\leftarrow \ExpanderDecompose(\Gtil, \phi)$.

\item For each $i$
\begin{tight_enumerate}
\item Initialize $\Kcal_{i}$ to empty
\end{tight_enumerate}

\item For each $K_{B, r} \in \Kcal_{B, r}$
and each $\Shat_{i}$ with non-zero intersection with it
\begin{tight_enumerate}
\item $\Kcal_{i} \leftarrow \Kcal_{i} \cup K_{B, r}[\Shat_{i}]$,
where $K_{B, r}[\Shat_{i}]$ is the portion of $K_{B, r}$ contained
in $\Shat_{i}$.
\end{tight_enumerate}

\item Initialize $H$ to empty.

\item For each $\Kcal_{i}$
\begin{tight_enumerate}
\item $H_{i} \leftarrow \SampleBiCliques (\Kcal_{i}, \eps)$
\item $H \leftarrow H \cup H_{i}$
\end{tight_enumerate}
\item Return $\{ \Shat_1, \Shat_2, \ldots\}, H$.
\end{tight_enumerate}
\end{algorithm}

\begin{lemma}
\label{lem:ImplicitPartitionAndSample}
Given any collection of $k$ balanced bicliques $\Kcal$,
and any error $\eps$ and $\phi$,

\noindent $\ImplicitPartitionAndSample(\Kcal, \eps, \phi)$
  returns in $\Otil{n \eps^{-1.5} + n(\Kcal) \eps^{-0.5} + n(\Kcal)
  \phi^{-2}}$ time a partition
$\{\Shat_{1}, \Shat_{2}, \ldots\}$ and a graph $H$
such that with high probability:
\begin{tight_enumerate}
\item $H$ has at most $\Otil{n \eps^{-1.5}
  + n(\Kcal) \eps^{-0.5}}$ edges,
\item the number of edges in $\Kcal$ on the boundary of some cluster,
is at most $4\gst(n) \phi m(\Kcal),$
\end{tight_enumerate}
and if we let $G(\Kcal)[\Shat_{i}]$ be the edges of $G(\Kcal)$
contained in some piece $\Shat_{i}$, we have
\begin{tight_enumerate}
\item With high probability
\[
\LL_{H}
\approx_{O\left( \sqrt{\phi^{-4} \eps \log{n} } \right)}
\sum_{i} \LL_{G\left(\Kcal\right)\left[\Shat_{i}\right]}.
\]
\item For any fixed vector $\xx$, with high probability
\[
\abs{
\xx^{\top} \LL_{H} \xx
-
\sum_{i} \xx^{\top} \LL_{G(\Kcal)\left[\Shat_{i}\right]} \xx
}
\leq
    O\left(\phi^{-4} \log^{4}{n} \right)
\cdot \eps
\cdot
\xx^{\top} \LL_{\Gtil\left( \Kcal \right)} \xx.
\]
\end{tight_enumerate}
\end{lemma}

\begin{proof}
Applying the guarantees of Lemma~\ref{lem:MatchingSparsifySimple}
to each $K_{B, r} \in \Kcal_{B, r}$ and summing over the results gives
that the crude sparsifier generated on Line~\ref{line:CrudeSparsify} satisfies:
\[
\Gtil
\approx_{2}
G\left( \Kcal_{B, r} \right),
\]
and the degrees of $\Gtil$ and $G(\Kcal_{B, r})$ are the same.

We first bound the edge count between clusters.
For the edges in $H_{i}$,
the result follows from the clusters $\Shat_{i}$ being
disjoint and Lemma~\ref{lem:SampleBiCliques}.

For the edges between clusters,
the guarantees of expander decomposition gives that the number
of such edges in $\Gtil$, when treated as a unit
graph, is $ \gst(n) \phi |E( \Gtil )|$.
Since each of these edges has weight $r / s$, and both the total
number of edges, and edges across cuts, are within factors of
$2$ in $G(\Kcal)$ and $\Gtil$, the total number of edges
on the boundaries of some $\Shat_{i}$ in $G$ is at most
\[
4  \gst\left( n \right) \phi \abs{E( G( \Kcal ) )}
= 4 \gst\left( n \right) \phi m\left( \Kcal\right).
\]

We now give the approximation guarantees.
The guarantees of expander decompositions from
Lemma~\ref{lem:ExpanderDecomposition} gives that
each $\Shat_{i}$ is contained in some $S_i$
(which are unknown to the algorithm) such that
$G[S_i]$ has conductance at least $\phi^2$.
Cheeger's inequality from Lemma~\ref{lem:Cheeger} then gives
\[
\LL_{G\left[S_{i}\right]}
\succeq
\frac{\phi^{4}}{2} \DD_{S_{i}, \perp 1}
\succeq
\frac{\phi^{4}}{2} \DD_{\Shat_{i}, \perp 1},
\]
This, combined with Claim~\ref{claim:bounded-by-diagonal} of
  Lemma~\ref{lem:SampleBiCliques} means that any graph $\Hhat$ sampled in the creation of $H_i$
is small w.r.t. $G(\Kcal_{B, r})$:
\[
\LL_{\Hhat}
\preceq
\eps \DD_{\Shat_{i}, \perp 1}
\preceq
2\eps \phi^{-4} \LL_{G\left[S_{i}\right]}
\preceq
2\eps \phi^{-4} \LL_{G}.
\]
So by Lemma~\ref{lem:SampleGraphGlobal}, gives
\[
\LL_{H_{i}}
\approx_{O\left( \sqrt{ \eps \phi^{-4} \log{n}} \right)}
\LL_{G\left( \Kcal \right)\left[\Shat_{i} \right]}
\]
which, as the sampled graphs are disjoint, gives
the approximation between $G$ and $H$.

We now turn to the per-vector guarantee.
Lemma~\ref{lem:SampleBiCliques} gives that on each
partition $\Shat_{i}$, we have
for any $\xhat$
\[
\abs{
  \xx^{\top} \LL_{G\left( \Kcal \right)\left[ \Shat_{i} \right]} \xx
  -
  \xx^{\top} \LL_{H_{i}} \xx
}
\leq
O\left( \log^2{n} \right)
\cdot
\eps
\cdot
\sum_{u \in \Shat_{i} } \dd_{u}
\left( \xx_{u} - \xhat \right)^2
\leq
O\left( \log^2{n} \right)
\cdot
\eps
\cdot
\sum_{u \in S_{i} } \dd_{u}
\left( \xx_{u} - \xhat \right)^2.
\]
As $\Gtil$ and $G(\Kcal_{B, r})$ have the same degree sequences,
we can make the critical step of interpreting the last term on $\Gtil$.
Lemma~\ref{lem:CheegerProjected} and the guarantees of
expander decompositions from Lemma~\ref{lem:ExpanderDecomposition}
then implies that for appropriate choices of $\xhat$ per expander cluster $S_{i}$,
specifically
\[
\xhat_{i}
:=
\frac{\sum_{u \in S_{i}} \dd_{u} \xx_{u}}
  {\sum_{u \in S_{i}} \dd_{u}  },
\]
gives
\[
\sum_{u \in S_{i} } \dd_{u}
\left( \xx_{u} - \xhat \right)^2
\leq
O\left(\phi^{-4} \log{n}  \right)
\xx^{\top} \LL_{\Gtil\left[S_{i}\right]} \xx,
\]
or if we sum over all partitions $\Shat_{i}$ gives
\[
\abs{
  \sum_{i} \xx^{\top} \LL_{G\left(\Kcal\right)\left[\Shat_{i}\right]} \xx
  -
  \xx^{\top} \LL_{H} \xx
}
=
\sum_{i}
\abs{
  \xx^{\top} \LL_{G\left( \Kcal \right)\left[ \Shat_{i} \right]} \xx
  -
  \xx^{\top} \LL_{H_{i}} \xx
}
\leq
O\left(\phi^{-4}  \log^{3}{n} \right)
\cdot \eps
\cdot
\sum_{i} \xx^{\top} \LL_{\Gtil\left[S_{i}\right]} \xx.
\]
We can then invoke the condition on the at most $O(\log{n})$ total
overlaps between the $S_{i}$s from Part~\ref{part:PartitionOverlap}
of Lemma~\ref{lem:ExpanderDecomposition} to get
\[
\sum_{i} \LL_{\Gtil\left[ S_i \right] }
\preceq
O\left( \log{n} \right)
\LL_{\Gtil}
\]
and in turn
\[
\abs{
  \sum_{i} \xx^{\top} \LL_{G\left(\Kcal\right)\left[\Shat_{i}\right]} \xx
  -
  \xx^{\top} \LL_{H} \xx
}
\leq
O\left(\eps \phi^{-4} \log^{4}{n} \right)
\xx^{\top} \LL_{\Gtil} \xx
\leq
O\left(\eps \phi^{-4} \log^{4}{n}  \right)
\xx^{\top} \LL_{G\left( \Kcal_{B, r} \right)} \xx.
\]
where the last inequality again follows from the spectral approximations
between $G(\Kcal_{B, r})$ and $\Gtil$.

The time bound holds because we run {\ExpanderDecompose} on a graph
with $O(n(\Kcal) \log n)$ edges with parameter $\phi$, and we run
{\SampleBiCliques} on $\Kcal_i$ for all $i$, where the sum of the vertices
in $\Kcal_i$ across all $i$ is bounded above by $n(\Kcal)$.
Note that the number of vertices in each
$\widehat{S_i}$ across all $i$ is equal to $n$. Therefore, the total
runtime (by invoking the runtimes in Lemma~\ref{lem:ExpanderDecomposition} and
Lemma~\ref{lem:SampleBiCliques} respectively) gives a total runtime of:
\[\Otil{n(\Kcal) \phi^{-2} + n \eps^{-1.5} + n(\Kcal) \eps^{-0.5}}
\]
as desired.
\end{proof}

We then need to repeat our process
on the edges between the pieces $\Shat_{i}$.
Note again that explicitly writing down these edges would incur
a significantly higher cost.
So our recursion relies on the interaction between a biclique
and a partition of vertices, which we will show is also a sum
of bicliques whose total representation size (in terms of vertex
counts) is only higher by a factor of $O(\log{n})$.
\begin{lemma}
\label{lem:BiCliqueSplit}
Let $K$ be a biclique on $n$ vertices,
and $\Shat_1, \Shat_2, \ldots$ be a partition of its vertices.
Then
\begin{tight_enumerate}
\item for any $i$,
the graph $K[\Shat_i]$, that's the restriction of $K$ to $\Shat_{i}$,
is also a biclique, and
\item the edges between the pieces,
\[
\bigcup_{1 \leq i \leq t}
E_{K}\left[\Shat_{i}, V \setminus \Shat_{i}\right]
\]
is also the union of bicliques whose total vertex count
is $O(n \log{n})$.
\end{tight_enumerate}
\end{lemma}

\begin{proof}
The restriction of $K$ to $\Shat_i$ is a biclique between
the subsets of both partitions contained in $\Shat_{i}$.
The boundary edges can be dealt with via a divide-and-conquer
argument:
\begin{tight_enumerate}
\item If any $\Shat_{i}$ has more than $1/3$ of the vertices, its
boundary is the sum of two bicliques to $V \setminus \Shat_{i}$,
and we can recurse 
on at most $(2/3)n$ vertices in $V \setminus \Shat_{i}$.
\item Otherwise we can partition the $\Shat_{i}$s into two subsets,
each with at most $(2/3)n$ vertices.
The edges between these two subsets is once again the sum of two
bicliques, and the
overall decomposition follows from recursing on the two halves.
\end{tight_enumerate}
In each case, the size of the subproblems that we recurse
on decreases by a constant factor, and they have at most $n$ vertices.
This means there are at most $O(\log{n})$ layers of recursion,
for a total of $O(n \log{n})$ vertices in these bicliques.
\end{proof}

This means we can then recurse on the boundary edges,
leading to an overall recursive scheme whose pseudocode
is in Algorithm~\ref{alg:SketchUnweightedBiCliques}.

\begin{algorithm}

\caption{$\ImplicitSketchUnweightedBiCliques(\Kcal_{B}, \eps, \phi, q)$}
\label{alg:SketchUnweightedBiCliques}

\underline{Input}:
Collection of bicliques $\Kcal_{B}$,
Error threshold $\eps$.
Conductance $\phi$ and recursion layer $q$.

\underline{Output}:
sketch $H$.

\begin{tight_enumerate}
\item If $q = 0$
\begin{tight_enumerate}
\item Let $H$ be the sum of all edges in each $G_i$ explicitly.
\end{tight_enumerate}

\item Else
\begin{tight_enumerate}
\item Initialize $H$ as empty.

\item $\left\{\Kcal_{B, 1}, \Kcal_{B, 2}, \Kcal_{B, 4} \ldots \right\}
\leftarrow
\MakeBalanced(\Kcal_{B})$.

\item For each $r = 2^{l}$

\begin{tight_enumerate}
\item $\{\Shat_1, \Shat_2, \ldots \}, H_{l} \leftarrow 
\ImplicitPartitionAndSample(\Kcal_{B, r}, \eps, \phi)$.

\item $H \leftarrow H \cup H_{l}$.

\item Let $\Kcal^{next}_{B, r}$ be the intersection of
$\Kcal_{B, r}$ with the boundary of the partitions
$\{\Shat_{1}, \Shat_{2}, \ldots \}$,
implicitly represented using Lemma~\ref{lem:BiCliqueSplit}.

\item $H \leftarrow H \cup \ImplicitSketchUnweightedBiCliques
  (\Kcal^{next}_{B, r}, \eps, \phi, q - 1)$.
\end{tight_enumerate}

\end{tight_enumerate}

\item Return $H$.
\end{tight_enumerate}

\end{algorithm}

\begin{proof}(of Lemma~\ref{lem:SketchUnweightedBiCliques})
We start with the approximation guarantees.
This follows from noting that each $\Kcal_{j}^{next}$
are edge-disjoint, so the total error across each layer
of the recursion sums to at most
\[
O\left( \phi^{-4} \log^{4} n \right) \cdot \epshat \cdot
\xx^{\top} \LL_{G\left( \Kcal \right)} \xx
\leq
\frac{\epsilon}{O\left( q \right)} \xx^{\top} \LL_{G\left( \Kcal \right)} \xx,
\]
which summed across the $q$ levels of recursion gives
$|\xx^{\top} \LL_{H} \xx - \xx^{\top} \LL_{G(\Kcal)} \xx|
\leq 0.1  \epsilon \cdot \xx^{\top} \LL_{G(\Kcal)} \xx$,
or $\xx^{\top} \LL_{G(\Kcal)} \xx
  \approx_{0.1 \eps}
\xx^{\top} \LL_{H} \xx$.

We now bound the size of the output and overall running time.
Lemma~\ref{lem:ImplicitPartitionAndSample} gives that
the edges between the pieces is bounded by
\[
4 \gst\left( n \right) \phi \cdot m\left( \Kcal \right)
\]
where $\gst(n)$ is the polylog overhead from expander decomposition
from Lemma~\ref{lem:ExpanderDecomposition}.
This means that after $q$ levels of recursion, the total number of edges
in the bicliques between the pieces is at most
\[
\left( 4 \gst\left( n \right) \phi \right)^{q} \cdot m\left( \Kcal \right)
\leq
\left( 4 \gst\left( n \right) \phi \right)^{q} \cdot n^2.
\]
On the other hand, Lemma~\ref{lem:BiCliqueSplit}
gives that the number of vertices in the bicliques increases by
a factor of $c_1 \log{n}$ per iteration, leading to a total edge count of
\[
q n \eps^{-1.5}
+ n\left( \Kcal \right)\left(c_1 \log{n} \right)^{q} \eps^{-0.5}
+ \left( 4 \gst\left( n \right) \phi \right)^{q} \cdot m,
\]
where the last term comes from explicitly forming the graph once
$q$ reaches $0$ in the recursion.
The choices of
\begin{align*}
q & = \exp\left(O\left( \sqrt{\frac{\log{n}}{ \log\log{n}}} \right) \right)\\
\phi & = \exp\left(-O\left( \sqrt{\log{n} \log\log{n}} \right) \right),
\end{align*}
bounds the overhead on $n(\Kcal)\epsilon^{-0.5}$:
\[
\left(c_1 \log{n} \right)^{q}
= \exp \left( O\left( \log\log{n} \right)
 \cdot O\left( \sqrt{\frac{\log{n}}{ \log\log{n}}} \right) \right)
\leq
\exp \left( O\left( \sqrt{\log{n}\log\log{n}} \right) \right)
\leq
n^{o\left( 1 \right)},
\]
and the fraction of edges remaining:
\[
\left( 4 \gst\left( n \right) \phi \right)^{q}
\leq
\exp\left( - O\left( \sqrt{\log{n} \log\log{n} } \right) \cdot q \right)
\leq
\exp\left( -O\left( \log{n} \right) \right)
\leq
m^{-1}.
\]
The choices of the two parameters also result in
the running time of the steps being almost-linear in
the representation size of $\Kcal$, $n(\Kcal)$.
Note that the allowed error, $\epshat$ also needs to be adjusted
to account for the $q$ steps, as well as the overhead
from Lemma~\ref{lem:ImplicitPartitionAndSample}. 
\end{proof}

\subsection{From Square-Sparsifiers to Resistances}
\label{subsec:ERReductions}

Now that we can efficiently sketch bicliques,
it only remains to show how this interacts with
being able to compute resistance estimates in a graph.
This top-most level of the algorithm relies on several reductions.
The purpose of this section is to elaborate on those reductions, and
reference prior works from which they are largely obtained.

The first of these is Section 6 of~\cite{DurfeeKPRS17}
\footnote{version 1 at~\url{https://arxiv.org/pdf/1611.07451v1.pdf}},
which is a reduction (via recursion) from computing Schur complements
to computing effective resistances.
The key definition is the Schur complement of a graph Laplacian:

\begin{definition}
\label{def:Schur}
The Schur complement of a graph Laplacian $\LL$ onto a subset
of vertices $C$, $\SC\left( \LL , C \right)$, is obtained by splitting the matrix $\LL$
into blocks separating $C$ and $F = V \setminus C$:
\[
\LL = 
\left[
\begin{array}{cc}
\LL_{FF} & \LL_{FC}\\
\LL_{CF} & \LL_{CC}
\end{array}
\right],
\]
and then computing the $C \times C$ matrix
\[
\SC\left( \LL , C \right)
:=
\LL_{CC} - \LL_{CF} \LL_{FF}^{+} \LL_{FC}.
\]
\end{definition}
The Schur complement has a natural interpretation as the intermediate state
of Gaussian elimination after removing $F$.
However, for algorithmic purposes, the above closed form is arguably more
important~\cite{KyngLPSS16,KyngS16,DurfeeKPRS17,DurfeePPR17}.
This is due to it exactly preserving a variety of quantities, especially
effective resistances.
\begin{fact}
\label{fact:SchurER}
For any pair of vertices $u, v \in C$,
the effective resistance between $u$ and $v$ is the same
in $\SC(\LL, C)$ and $\LL$.
\end{fact}
We remark that this preservation of resistances is
due to the more general fact that
the inverse of the Schur complement, $\SC(\LL, C)^{+}$
is the minor on $C$ of the pseudo-inverse of $\LL$,
$(\LL^{+})_{C,C}$, after projection against the null space.
However, for the purpose of estimating effective resistances,
the interaction with algorithms for producing approximate
Schur complements can be summarized as follows.
\begin{lemma}
(Reduction from Computing Effective Resistances to
Resistance-Approximating Sketches of Schur Complements,
Section 6 of~\cite{DurfeeKPRS17})
\label{lem:ERFromSC}
Given a routine that takes a graph on $\nhat$ vertices,
$\mhat$ edges, a subset of vertices $C$, and an error parameter
$\epshat$, and returns a resistance-preserving sketch of $\SC(\LL, C)$
of size $\nhat \Delta(\nhat, \epshat)$
in time $T_{SC}(\nhat, \mhat, \epshat)$,
we can compute $1\pm \epsilon$ approximations of resistances between
$t$ pairs of vertices of a graph on $n$ vertices and $m$ edges in time
\[
\Otil{
T_{SC}\left(n, m, \frac{\epsilon}{2} \right)
+
\left( 1 + \frac{t}{n} \right) \cdot
  T_{SC}\left(n, n \Delta\left(n,
      \frac{\eps}{O\left( \log{n} \right)}\right),
    \frac{\eps}{O\left( \log{n} \right)} \right)
    }.
\]
\end{lemma}
So our goal becomes adapting the resistance sparsifiers from
Section~\ref{sec:Resistance-Sparsifiers} to work (implicitly)
on Schur complements.
However, the lack of matrix concentration/martingale bounds
precludes us from invoking the single-vertex pivoting schemes
from~\cite{KyngS16}.  Instead, we utilize the square and sparsify
method from Appendix G of~\cite{KyngLPSS16}\footnote{we cite version 1
  at~\url{https://arxiv.org/pdf/1512.01892v1.pdf}}, which has a higher
polylog overhead, but yields more control over the intermediate
states.  Specifically, for a block decomposition into $F$ and $C$ with
the top-left block further decomposed into
\[
\LL_{FF} = \DD_{FF} - \AA_{FF},
\]
the Schur complement onto $C$ can be expressed as
\begin{equation}
\SC\left( \LL , C \right)
=
\frac{1}{2}
\SC
\left(
\left[
\begin{array}{cc}
\LL_{FF} - \AA_{FF} \DD_{FF}^{-1} \AA_{FF}
& 
\LL_{FC} + \AA_{FF} \DD_{FF}^{-1} \LL_{FC}\\
\LL_{CF} + \LL_{CF} \DD_{FF}^{-1} \AA_{FF}
&
2 \LL_{CC} - \LL_{CF} \DD_{FF}^{-1} \LL_{FC}
\end{array}
\right]
\right).
\label{eq:SchurStep}
\end{equation}
This is given in Equation 10 of~\cite{KyngLPSS16}.
Note that $\LL_{CF}$ consists of entirely off-diagonal entries,
and is often also represented as $-\AA_{CF}$,
a minor of the adjacency matrix, instead.

Somewhat surprisingly, this matrix is in fact a graph Laplacian.
Furthermore, it decomposes into sums of weighted cliques
and weighted bipartite cliques.
Formally a weighted (bi)clique is given by a set of vertex
weights $\ww_{u}$, with the edge between $u$ and $v$ having
weight $\ww_{u} \cdot \ww_{v}.$
Here, the terms on the top-left and bottom-right blocks,
$\AA_{FF} \DD_{FF}^{-1} \AA_{FF}$
and
$\LL_{CF} \DD_{FF}^{-1} \LL_{FC}$
are sums of weighted cliques on $F$ and $C$,
and $-\AA_{FF} \DD_{FF}^{-1} \LL_{FC}$
is a sum of weighted bicliques between $F$ and $C$.
Specifically,
\begin{tight_itemize}
\item $\AA_{FF} \DD_{FF}^{-1} \AA_{FF}$ is a sum of cliques on $F$,
one per vertex of $f \in C$.
The weight of an edge between $f_1$ and $f_2$ in the clique centered
at $f$ is given by
\[
\frac{\AA_{f_1f} \AA_{f_2f}}{\DD_{ff}},
\]
which corresponds to a weight vector
$\ww_{\fhat} = \DD_{ff}^{-\nfrac{1}{2}} \AA_{\fhat f}$.
\item $\LL_{CF} \DD_{FF}^{-1} \LL_{FC}$ is a sum of cliques on $C$,
one per vertex of $f \in F$.
The weight of an edge between $c_1$ and $c_2$ in the clique centered
at $f$ is given by
\[
\frac{\LL_{c_1f} \LL_{c_2f}}{\DD_{ff}},
\]
which corresponds to a weight of
$\ww_{\chat} = -\DD_{ff}^{-\nfrac{1}{2}} \LL_{\chat f}$.
Note that the negation is necessary because the off-diagonal
entries of $\LL$ are negative.
\item $-\AA_{FF} \DD_{FF}^{-1} \LL_{FC}$
is a sum of bicliques between $F$ and $C$, one per vertex in $F$.
The weight of an edge between $\chat$ and $\fhat$ in the clique centered
at $f$ is given by
\[\frac{-\AA_{\chat f} \LL_{\fhat f}}{\DD_{ff}},
\]
which, as off-diagonal entries of $\LL$ are negative,
corresponds to a weight of
\begin{align*}
\ww_{\fhat} & = \DD_{ff}^{-\nfrac{1}{2}} \AA_{\fhat f},\\
\ww_{\chat} & = -\DD_{ff}^{-\nfrac{1}{2}} \LL_{\chat f}.
\end{align*}
\end{tight_itemize}

This provides the interaction point with the clique sparsifiers
from Lemma~\ref{lem:SketchUnweightedBiCliques}.
However, as Lemma~\ref{lem:SketchUnweightedBiCliques} only works
for a sum of unweighted bicliques, we need to make two additional
transformations.
\begin{lemma}
\label{lem:ReductionToBiCliques}
Any weighted clique on $n$ vertices is equivalent to a sum of
weighted bicliques on a total of $O(n \log{n})$ vertices.
\end{lemma}

\begin{proof}
This is via divide-and-conquer: we can split the vertices into
two halves of (roughly) equal size, and represent all edges between
them by a biclique.
Then recursing on both halves gives a total of $O(n \log{n})$ vertices.
\end{proof}
\begin{lemma}
\label{lem:ReductionToUnweighted}
A weighted biclique on $n$ vertices with weights in
a $poly(n)$ range can be approximated
by a sum of uniform weighted bicliques on a total of $O(n \log^2{n})$
vertices while incurring a multiplicative error of $1 / poly(n)$.
\end{lemma}

\begin{proof}
We first decompose each $\ww_{u}$ into its binary representation.
Here we can truncate after $O(\log{n})$ bits while incurring
an error of $1 / poly(n)$.

Then an edge of weight $\ww_{u} \cdot \ww_{v}$ can be factored
based on the binary representation of $\ww_{u}$ and $\ww_{v}$.
Suppose the 1 bits in the binary representation of $\ww_{u}$
are $i_{u, 1}, i_{u, 2} \ldots i_{u, l(u)}$, and similarly for $v$,
then we have
\begin{align*}
\ww_u
& = 2^{i_{u, 1}} + 2^{i_{u,2}} + \ldots + 2^{i_{u,l(u)}}\\
\ww_v
& = 2^{i_{v, 1}} + 2^{i_{v,2}} + \ldots + 2^{i_{v,l(v)}}
\end{align*}
which then gives
\[
\ww_u \cdot \ww_{v}
=
\prod_{1 \leq j \leq l(u), 1 \leq k \leq l(v)}
2^{i_{u, j} + i_{v, k}}.
\]
As the vertices that have $1$ on the $i\textsuperscript{th}$ bit
is a subset, this creates one biclique per pair of such subsets,
for a total of $O(\log^2{n})$ bicliques, each including every vertex
in the worst case.
\end{proof}

Combining these with Lemma~\ref{lem:SketchUnweightedBiCliques}
then gives our algorithm for carrying one step of the Schur
complement approximation algorithm.
\begin{lemma}
\label{lem:SchurStep}
Given a graph $G$ with $n$ vertices and $m$ edges,
we can produce an $\eps$-graphic spectral sketch of the graph
corresponding to the matrix on the RHS of Equation~\ref{eq:SchurStep} with
$n^{1 + o(1)} \eps^{-1}$ edges in
$n^{1 + o(1)} \eps^{-1.5} + m n^{o(1)} \eps^{-0.5}$ time.
\end{lemma}

\begin{proof}
As observed above, the RHS of Equation~\ref{eq:SchurStep}
is a sum of weighted cliques
whose total number of vertices is $m$.
Lemmas~\ref{lem:ReductionToBiCliques} and~\ref{lem:ReductionToUnweighted}
allow us to decompose these into unit weighted bicliques.

Theorem~\ref{thm:ReductionToUnit} further reduces the number of vertices
involved in the graph representation of these bicliques.
Specifically we obtain $\Kcal_{B, 1}, \Kcal_{B, 2} \ldots$ such that
\[
\sum_{i} n\left( \Kcal_{B, i} \right)
\leq
\Otil{m},
\]
and each $\Kcal_{B, i}$ has edge weights which are powers of $2$.

Invoking Lemma~\ref{lem:SketchUnweightedBiCliques}
on each $\Kcal_{B, i}$ gives an $\eps$-spectral
sketch of this graph with
\[
n^{1 + o\left(1\right)} \epsilon^{-1.5 }
+
m n^{o\left( 1 \right)} \epsilon^{-0.5}
\]
edges.
Note that the guarantees of Lemma~\ref{lem:SketchUnweightedBiCliques}
are additive across a sum of graphs.
Thus Lemma~\ref{lem:Invert} gives the preservation of quadratic
inverse forms in this approximation as well.

Sketching this graph explicitly using
$\SpectralSketch(\cdot, \epsilon / 2,
\ShortCycleAlgo)$ as specified in
Theorems~\ref{thm:SpectralSketches}
and ~\ref{thm:ShortCycleDecomposition} gives the final edge count
and approximation guarantees.
The running time follows from the size of this explicit graph
that we sketch, as well as the guarantees of
Lemma~\ref{lem:SketchUnweightedBiCliques}.
\end{proof}

Invoking this routine repeatedly as in~\cite{KyngLPSS16},
and then recursively within the effective resistance estimation
routine from~\cite{DurfeeKPRS17} gives the overall result
for computing effective resistances.

\begin{proof}(of Theorem~\ref{thm:MainER})
We first consider graphs whose weights are $poly(n)$, and therefore
have condition number $poly(n)$.
An additional property of Schur complements is that they can only
increase the minimum (non-zero) eigenvalue, and decrease
the maximum (weighted) degree.
This means as long as the initial eigenvalues are poly bounded,
we can `fix' the Schur complements after each step
based on eigenvalues to ensure that the edge weights remain
poly bounded while incurring a $1 / poly(n)$ error at each step. 

Lemma G.2 from~\cite{KyngLPSS16} along with this poly(n) bound
on condition number means that after $O(\log{n})$ iterations
of applying the transformation from Equation~\ref{eq:SchurStep},
the top-left (FF) block becomes negligible.
This then means that the bottom right block is a resistance
preserving sparsifier of $\SC(\LL, C)$.
This then fits into the requirement of Lemma~\ref{lem:ERFromSC} with
parameters
\begin{align*}
\Delta\left(\nhat, \epshat \right)
& = \nhat^{o(1)} \cdot \epshat^{-1}\\
T_{SC}\left(\nhat, \mhat, \epshat^{-1} \right)
& = \nhat^{1 +o(1)} \epshat^{-3/2}
+  \mhat^{1 + o(1)} \epshat^{-\nfrac{1}{2}},
\end{align*}
which gives the overall running time.
Note that we accumulate errors naively during this process,
leading to additional polylog factors in the errors.

Finally, Appendix F of~\cite{CohenKPPRSV17} provides a reduction
from solving linear systems in arbitrary graph Laplacians to solving
linear systems in graphs with poly(n) bounded edge weights.
Applying it to resistances allows us to make the poly(n) bounded
weights assumption as above, at the cost of another $O(\log{n})$ overhead.
Broadly speaking, such reductions are similar to those in
Appendix~\ref{sec:ReductionToUnit} in that they remove/contract
edges whose weights are smaller/larger than true estimates
by $poly(n)$ factors, while incurring an error of $1 / poly(n)$.
We omit explicitly adapting this argument to the effective resistance
case in anticipation of a more systematic and unified treatment
of such reductions in the near future.
\end{proof}

%%% Local Variables:
%%% mode: latex
%%% TeX-master: "main"
%%% End:

%% file: construction.tex
%!TEX root = main.tex

\section{Efficient Construction of Short Cycle Decomposition}
\label{sec:construction}
In this section, we give an almost-linear time algorithm to construct
a short cycle decomposition of a graph, proving
Theorem~\ref{thm:ShortCycleDecomposition}.  We start with the
existence proof (Theorem~\ref{thm:NaiveCycleDecomposition}), which can also be
phrased as an $O(nm)$ time algorithm.

\begin{algorithm}
\caption{$\NaiveCycleDecomposition(G)$}\label{fig:NaiveCycleDecompose}

\textbf{Input:} Undirected Unweighted Graph $G(V,E)$.

\textbf{Output:} $\mathcal{C}$, that's a $(2n, 2\log_2{n})$ cycle
decomposition.

\begin{tight_enumerate}

\item Repeat until $G$ is empty.

\begin{tight_enumerate}

\item While $G$ has a vertex $u$ of degree $\le 2$, remove $u$ and the
  edges incident to $u$ from $G.$

\item Run breadth-first search (BFS) from an arbitrary vertex $r$ until you encounter the
  first non-tree edge. 

\item Let $C$ be the cycle formed by $e$ and the tree edges.

\item Add $C$ to the collection of cycles,
$\mathcal{C} \leftarrow \mathcal{C} \cup \{ C \}$.

\item Remove the edges of $C$ from $G$,
$E(G) \leftarrow E(G) \setminus E(C)$.

\end{tight_enumerate}

\item Return $\mathcal{C}.$

\end{tight_enumerate}

\end{algorithm}

\begin{proof} (of Theorem~\ref{thm:NaiveCycleDecomposition}). 
  We first bound the running time: Each iteration of the outer loop
  takes $O(n)$ time since we stop constructing the BFS tree when we
  encounter the first non-tree edge. Since each iteration of the outer
  loop removes at least one edge from $G$, the overall running time is
  $O(mn)$.

  For every vertex deleted, we delete at most 2 edges, and hence at
  most $2n$ edges are not part of the cycles.

  When running BFS, $G$ does not have any vertices of degree $1,$ and
  hence it cannot be a tree, and thus BFS must find a non-tree edge.
  It remains to bound the lengths of cycles, which is at most the
  depth of the first non-tree edge in the BFS tree $T$.  Suppose this
  happens at a depth of $d$.  Then every node at depth $d$ or higher
  has at least two children, for a total of at least $2^{d}$ vertices.
  Since the number of vertices remaining is at most $n$, we get
  $2^{d} \leq n$, or $d \leq \log_2{n}$.
\end{proof}

We now turn to the almost-linear time cycle decomposition algorithm
that produces longer cycles. Its pseudocode is presented as
Algorithm~\ref{alg:CycleDecompose}. At a high level, its outline is as
follows:
\begin{tight_enumerate}
\item Decompose the graph into a set of disjoint expanders
and a small number of edges.
\item In each expander, select a small set of high degree vertices
and ``port'' a large number of the edges into these vertices
via random walks.
\item Recurse on the ported edges on these smaller vertex sets.
\end{tight_enumerate}

\subsection{Porting Edges onto Fewer Vertices}
\label{subsec:Construction:MoveEdges}
In this section, we describe our key routine that ports a large
fraction of edges onto a smaller set of vertices $S$.
The object that allow this movement and subsequent reconstruction
can be formalized as a cycle decomposition of the graph $G / S$,
the graph $G$ with the vertices of $S$ shrunk to a single vertex,
which we'll denote as $s^{*}$.
\begin{definition}
\label{def:PartialCycleDecomp}
Given a graph $G$ and a vertex subset $S \subseteq V$,
a $(\lhat, \khat)$-partial cycle decomposition onto $S$ is a union
of at least $\khat$ edge-disjoint cycles of length at most $\lhat$
in $G / S$, the graph formed by shrinking $S$ to a single vertex
$s^{*}$.
\end{definition}

Such a partial cycle decomposition allows us to create
a graph on $S$ such that any cycle decomposition of this smaller graph
can be transformed into a cycle decomposition of $G$ with the same
number of cycles, and a length overhead of $\lhat$.
\begin{lemma}
\label{lem:Extend}
Given $G$, $S \subset V$, and a $(\lhat, \khat)$-partial cycle decomposition
onto $S$, we can create a graph $G_{S}$ with at most $\khat$
edges such that any $(\mhat, l)$-short cycle decomposition of $G_{S}$
can be transformed in $O(m)$ time into cycles on
$G$ with length at most $l \cdot \lhat$ that contain at least
$\khat - \mhat$ edges.
\end{lemma}

\begin{proof}
First, note that any cycle in $G / S$ that does not pass through
$s^{*}$ is already a cycle in $G$.

We create $G_{S}$ as follows: for each cycle in the partial decomposition
that passes through $s^{*}$ (the new vertex corresponding to $S$),
let the edges incident to $s^{*}$ in this cycle be $e_1$ and $e_2$.
These two edges are also then incident to vertices in $S$,
$s_1$ and $s_2$.
We add an edge between $s_1$ and $s_2$ in $G_S$.

This means any edge in $G_{S}$ has a path in $G$ between its two end
points.
Furthermore, the corresponding paths between all edges of $G_{S}$ are
edge-disjoint.
Then for any cycle in $G_S$, replacing each edge in
$G_S$ with the corresponding path in $G$ completes it into a circuit in $G$,
which we can then break into cycles.  As all edges in this cycle
decomposition of $G_{S}$ remain in these cycles, we get the lower
bound on number of edges involved.
\end{proof}

Our algorithm is then based on repeatedly generating such partial
cycle decompositions containing a significant fraction of the edges,
recursively finding cycle decompositions of the resulting graph $G_{S}$,
and extending them back to cycles in $G$ using Lemma~\ref{lem:Extend} above.
It relies on the following key size reduction routine.
\begin{lemma}
\label{lem:MoveEdges}
Let $\gns(n) = \exp(O(\sqrt{ \log n \log \log n}))$ be the overhead
term from expander decomposition routine
(Lemma~\ref{lem:EdgeExpanderDecomposition}~\cite{NanongkaiS17}).
Given a graph $G$ on $n$ vertices with $m$ edges and degrees between
$[d_{\min}, d_{\max}],$ along with a reduction factor
$k \ge 10 \log n$  and $k \le n$ such that
\[
d_{\min}
\geq
8000
\left( \frac{d_{\max}}{d_{\min}} \right)^2
\gns^3
k \log n,
\]
the routine $\MoveEdges(G,k)$ runs in
$O(m k (d_{\max} / d_{\min})^2 \gns^{2} \log {n})$ expected time,
and returns a subset of at most $2 n / k$ vertices $S$ and a
\[
\left(
400 \left( \frac{d_{\max}}{d_{\min}} \right)^2 \gns^2 \log n,
\left( \frac{d_{\min}}{d_{\max}} \right)^4
\cdot
\frac{1}{10^{6} k \gns^4 \log^2{n}}
\cdot
m
\right)
\]
partial cycle decomposition of $G$ onto $S$.
\end{lemma}

Our notion of expanders is given by conductance.
The conductance of a graph $G$ is:
\[
\phi\left( G \right)
:=
\min_{S \subseteq V\left( G \right), 0 < \vol{S} \leq m}
\frac{\abs{E_{G}\left( S, V \setminus S\right)}}{\vol{S}},
\]
where $\vol{S}$ is the total degrees of the vertices in $S$. A large
conductance means that a random walk starting anywhere in the graph
quickly approaches the uniform distribution.  For our purposes, the
role played by conductance can be summarized as:
\begin{lemma}
\label{lem:ExpanderMixing}
Let $G$ be an unweighted undirected simple graph with conductance
$\phi$, and let $S$ be any subset of vertices.  Then for any vertex
$u$, and any random walk starting at $u$ continuing for
$10 \phi^{-2} \log n$ steps is at a vertex in $S$ with probability at
least $\frac{\vol{S}}{3m}.$
\end{lemma}
\begin{proof}
  Let $A$ denote the adjacency matrix of $G$ and $D$ denote the
  diagonal degree matrix. The random walk matrix for the lazy walk on
  $G$ is given by $W = \frac{1}{2} I + \frac{1}{2} AD^{-1},$ and the
  normalized adjacency matrix
  $N = D^{-\nfrac{1}{2}} (D - A) D^{-\nfrac{1}{2}}.$ Cheeger's
  inequality gives us $\lambda_2(N) \ge \phi^2 / 2.$ Thus,
  \[\max_{\lambda \neq 1} \lambda(W) = 1 - \lambda_{2}(I-W) = 1 -
    \frac{1}{2} \lambda_{2}(N) \le 1-\phi^2 / 4.\]
  Thus, if we take a $k$-step lazy random walk in $G$ with
  $k = 10\phi^{-2} \log n,$ we get,
  $\max_{\lambda \neq 1} \lambda(W^k) \le n^{-\nfrac{5}{2}}.$ Thus,
  the probability that a $k$-step random walk from $u$ ends up at a
  vertex in $S$ is given by $\one^{\top}_S W^{k} \one_u.$ In order to
  lower bound this, we observe that
  $D^{-\nfrac{1}{2}} W D^{\nfrac{1}{2}}$ is a symmetric matrix with
  the same eigenvalues as $W$ and thus with a largest eigenvalue of 1
  (with eigenvector $D^{\nfrac{1}{2}}\one).$ Thus,
  \begin{align*}
    \one^{\top}_S W^{k} \one_u
    & = \one^{\top}_S
      D^{\nfrac{1}{2}} (D^{-\nfrac{1}{2}} W D^{\nfrac{1}{2}})^{k}
      D^{-\nfrac{1}{2}} \one_u \\
    & \ge \one^{\top}_S D^{\nfrac{1}{2}}
      \left(\frac{1}{2m} D^{\nfrac{1}{2}} \one \one^{\top} D^{\nfrac{1}{2}}
      \right) D^{-\nfrac{1}{2}} \one_u - \frac{1}{n^{\nfrac{5}{2}}}
      \norm{D^{\nfrac{1}{2}}\one_S} \norm{D^{-\nfrac{1}{2}}\one_u} \\
    & \ge
      \frac{\vol{S}}{2m} - \frac{1}{n^{\nfrac{5}{2}}}\sqrt{\frac{\vol{S}}{deg(u)}}
      \ge \frac{\vol{S}}{3m},
  \end{align*}
  for $n$ large enough.
\end{proof}
Note that the $\log n$ factor in the length of the random walks is
necessary, since there exist constant degree expander graphs (with
$\phi = \Omega(1)$) where at least half the vertices are at a distance
$\Omega(\log n)$ from any starting vertex $u.$

The above lemma prompts us to start with
expanders. Algorithm~\ref{alg:MoveEdgesExpander} describes the
pseudocode for this procedure.
\begin{lemma}\label{lem:MoveEdgesExpander}
  Given an unweighted undirected graph $G$ with $n$ vertices, $m$
  edges, conductance at least $\phi$, for any parameter $k$ such that
  \[
    10\log n \leq k \leq \frac{\phi^{2} d_{\min}}{100 \log{n}},
  \]
  where $d_{\min}$ is the minimum degree of a vertex in $G$,
  $\MoveEdgesExpander(G, \phi, k)$ returns a graph on $\ceil{n /k}$
  vertices along with a
  \[
    \left( 25 \phi^{-2} \log n , \frac{\phi^4 m}{2 \cdot 10^3 k
        \log^{2}{n}} \right)
  \]
  partial cycle decomposition of $G$ onto $S,$ in expected running
  time $O(mk\phi^{-2} \log n).$
\end{lemma}

\begin{algorithm}

\caption{$\MoveEdgesExpander(G, \phi, k)$}
	
\textbf{Input}: Undirected unweighted graph $G = (V, E)$
with $n \geq k$ vertices and $m \geq n$ edges\\
$\phi,$ a lower bound on conductance of $G$\\
Reduction factor $k$.

\textbf{Output}: A set of $\ceil{n / k}$ vertices $S$\\
A partial cycle decomposition of $G$ onto $S$.

\begin{tight_enumerate}
\item Pair up multiple edges to form cycles.
\item Pick $S$ consisting of the $\ceil{n / k}$ vertices of maximum degree in
  $G$.
\item
\label{step:GenerateWalks}
For every edge $e=(u,v) \in E$,
\begin{tight_enumerate}
\item Generate $4k$ lazy random walks each from $u$ and $v$ of length
  $10 \phi^{-2} \log n$. A lazy random walk stays at the current
  vertex with probability $\nfrac{1}{2}$ at each step.
  \item Discard walks that use $e.$
  \item Pick any one walk each from $u$ and $v$ (if there is one) that
    terminates in $S$ and convert them into simple paths.
    Add the corresponding cycle in $G/S$ to consideration.
\end{tight_enumerate}
\item Greedily pick a set of edges whose corresponding cycles are edge-disjoint.
\item If fewer than
  $\frac{\phi^4 m}{2 \cdot 10^3 k \log^{2}{n}}$ cycles are formed,
  go to Step~\ref{step:GenerateWalks}.

\item Return $S$, and the partial cycle decomposition 
\end{tight_enumerate}

\label{alg:MoveEdgesExpander}

\end{algorithm}

We first prove that sampling random walks does not incur too much
congestion.
\begin{lemma}
\label{lem:RandomWalkCongestion}
In any graph $G$, if we sample $k \geq 10\log{n}$ random walks, each
of length $L,$ from both end points of each edge in $G.$ The
congestion on every edge $e$ is bounded by $4kL$ with probability at
least $1-\frac{1}{n}.$
\end{lemma}

\begin{proof}
  Orient the edges arbitrarily, and fix an edge $\hat{e}.$ Let
  $X_{e, j, i}$ denote the probability that a single random walk
  starting at the $j^\textrm{th}$ end point of $e$ ($j \in \{1, 2\}$)
  goes through $\hat{e}$ at step $i$.  Since the stationary
  distribution is uniform across all the edges, for all $i$, we have,
\[
\sum_{e} \sum_{j =1,2} \expec{}{X_{e,j, i}} = 2.
\]
However, note that $X_{e, j, i}$ are not independent.

Thus, we define the random variable $X_e^{(j)}$ to be the congestion
on $\hat{e}$ incurred by the $j\textsuperscript{th}$ length $L$ walks
out of both end points of $e.$ We then have
$ X_{e}^{\left(j\right)} \leq 2L,$ and the total congestion on
$\hat{e}$ is $\sum_{e, j} X_{e}^{(j)}$ with the expectation
\[
  \sum_{1 \leq j \leq k} \sum_{e \in E} \expec{}{X_{e}^{(j)}} \leq 2Lk.
\]
Since $X_{e}^{(j)}$ are independent for all $e,j,$ and we have
$k > \log{n},$ applying Hoeffding's bound, we obtain that the total
congestion of $\hat{e}$ is bounded by $4kL$ with probability at least
$1-e^{-k/3} \ge 1-n^{-3}$ with high probability.

  Applying union bound over all edges $\hat{e},$ we obtain our claim.
\end{proof}

\begin{proof} (of Lemma~\ref{lem:MoveEdgesExpander}) Since $S$ is the
  set of $n/k$ nodes with the highest degree, its volume is at least
  $\frac{2m}{k}.$ So by Lemma~\ref{lem:ExpanderMixing}, the
  probability that a random walk on $G$ arrives at $S$ after
  $10 \phi^{-2} \log n$ steps is at least
  $ \frac{2m}{k} \cdot \frac{1}{3m} = \frac{2}{3k}.$
  
  Furthermore, since each vertex of $G$ has degree at least
  $100 \phi^{-2} k \log n$, this walk uses the edge $e$ with probability at most
  $ \frac{1}{100 k \phi^{-2} \log n} \cdot {10 \phi^{-2} \log n} =
  \frac{1}{10k}.$ Thus, the walk hits $S$ without using the edge $e$
  with probability at least
  $ \frac{2}{3k} - \frac{1}{10k} \ge \frac{1}{2k}.$

  Since we take $4k$ random walks from each end point of $e$, this
  means that with probability at least $1 - 2/e^2 \ge \frac{3}{5},$ at
  least one of the walks from each end point of $e$ reaches $S$.
  Thus, in expectation, at least $\frac{3m}{5}$ edges reach $S$,
  and the probability of fewer than $\frac{m}{2}$ edges reaching $S$
  is at most $\frac{4}{5}.$ 

  We now need to bound the number of disjoint cycles found.
  Lemma~\ref{lem:RandomWalkCongestion} proves that the congestion of
  each edge is at most $40 k \phi^{-2} \log{n}$ with high
  probability. Assuming this congestion bound holds, since each
  cycle uses at most $20 \phi^{-2} \log n$ edges (in addition to
  $e$), every cycle conflicts with at most
  $(1 + 20 \phi^{-2} \log n)(1 + 40 \phi^{-2} \log n)$ other ones.
  Thus, with probability at least $\frac{1}{10}$,
  we have at least
  \[
  \frac{m}{2} \cdot \frac{1}{(1 + 20 \phi^{-2} \log n) (1 + 40 k \phi^{-2}
      \log n)} \ge \frac{m \phi^{4}}{2\cdot 10^3
      k \log^{2} n}
  \]
  cycles, giving us our bound.
\end{proof}

Lemma~\ref{lem:MoveEdgesExpander} gives a way to partially cycle decompose
an expander, which suggests using it in conjunction with an expander
decomposition scheme.
However, most existing efficient expander decomposition schemes only provide
clusters that are contained in expanders.
Since our algorithm requires subgraphs where the random
walks mix well, we need to restrict the random walks to expanders.

Instead we utilize a more result by Nanongkai and
Saranurak~\cite{NanongkaiS17} that guarantees an expander
decomposition where the resulting subgraphs are expanders.
This routine gives its guarantees in terms of edge expansion,
$h(G)$, which can be defined as
\[
h\left( G \right)
=
\min_{S \subseteq V\left( G \right), 0 < \abs{S} \leq \frac{n}{2}}
\frac{\abs{E_{G}\left( S, V \setminus S\right)}}{\abs{S}}.
\]

\begin{lemma}[Theorem 5.1 from~\cite{NanongkaiS17} \footnote{From
    Version 2~\url{https://arxiv.org/pdf/1611.03745v2.pdf}.}]
  \label{lem:EdgeExpanderDecomposition}
  There is an algorithm {\NSExpanderDecompose} that for any
  undirected graph $G=(V,E)$ and parameter $\alpha>0$, partitions $E$
  into $E^s$ and $E^d$ such that for
  $\gns(n) = \exp(O(\sqrt{ \log n \log \log n})),$
  \begin{tight_enumerate}
  \item $|E^s| \leq \alpha \gns(n) n,$ and,
  \item with high probability, every connected component $H_{i}$ of
    $G^d = (V, E^d)$ is either a singleton or has edge expansion
    $h(H_{i}) \geq \alpha$.
  \end{tight_enumerate}
  The time taken by the algorithm is $O(m \gns(n) \log n)$.
\end{lemma}

Note in particular that having edge expansion at least $\alpha$
implies that each vertex has degree at least $\alpha$.
We can translate between edge expansion and conductance when
the graphs have bounded degrees.
This is the reason for the dependency on
$d_{\max} / d_{\min}$ in Lemma~\ref{lem:MoveEdges}.
The following simple lemma allows us to convert an edge
expansion bound into a conductance bound.
\begin{lemma}
\label{lem:EdgeExpansionToConductance}
If a graph $G$ with degree at most $d_{\max}$ has edge-expansion at
least $\alpha$, then the conductance of $G$ is at least
$\frac{\alpha}{d_{\max}}.$
\end{lemma}

\begin{proof}
  Consider any set of vertices $S$ in $G$ with $\vol{S} \leq m$. Since
  the maximum degree in $G$ is $d_{\max},$ we have
  $\abs{S} \geq \frac{\vol{S}}{d_{\max}}.$ If $|S| \leq n / 2$, then
  invoking the edge expansion bound of $G$ on $S$ we get,
\[
\abs{E\left( S, V \setminus S\right)}
\geq
\alpha \cdot \abs{S}
\geq
\frac{\alpha}{d_{\max}} \vol{S}.
\]

Otherwise, since $\vol{V \setminus S} \geq m$, we have
$\frac{n}{2} \geq \abs{V \setminus S} \geq \frac{m}{d_{\max}}.$ Thus,
\[
\abs{E\left( S, V \setminus S\right)}
\geq
\alpha \cdot \abs{V \setminus S}
\geq
\frac{\alpha}{d_{\max}} m
\geq\frac{\alpha}{d_{\max}} \vol{S}.
\]
\end{proof}

With this bound in mind, we can then given the overall routine for
moving the edges.  Its pseudocode is given in
Algorithm~\ref{alg:MoveEdges}.

\begin{algorithm}

\caption{$\MoveEdges(G, k)$}

\textbf{Input}: graph $G = (V, E)$ with $n$ vertices,
$m$ edges, and degrees in range $[d_{\min}, d_{\max}]$.\\
Vertex count reduction parameter $k$.\\
access to parameter that's expander partitioning overhead $\gns(n)$
from Lemma~\ref{lem:EdgeExpanderDecomposition}.

\textbf{Output}:
A set of at most $2 n / k$ vertices $S$
and a partial cycle decomposition of $G$ onto $S$.

\begin{tight_enumerate}
\item Set $\alpha \leftarrow d_{\min} / 4 \gns(n)$

\item $\left(E^{s}, E^{d}\right)
\leftarrow
\NSExpanderDecompose(G, \alpha)$.

\item For each connected component $H_{i}$ of $E^{s}$

\begin{tight_enumerate}

\item If $|V(H_i)| \leq k$, record the cycles generated using
  $\NaiveCycleDecomposition(H_i)$.

\item Else record the result of
  $\MoveEdgesExpander (H_i, \alpha / d_{\max}, k)$.

\end{tight_enumerate}
\item Return the unions of the subsets and partial cycle
decompositions onto them.

\end{tight_enumerate}

\label{alg:MoveEdges}
\end{algorithm}

\begin{proof}(of Lemma~\ref{lem:MoveEdges}) By guarantees of the
  expander decomposition given by
  Lemma~\ref{lem:EdgeExpanderDecomposition}, the number of edges
  contained in the expanders, $E^{d}$, is at least
\[
m - \alpha \gns\left( n \right) n
\geq
m - \frac{d_{\min}}{4} n
\geq
\frac{m}{2},
\]
where we used our choice of $\alpha = \frac{d_{\min}}{4 \gns(n)},$ and
$n d_{\min} \leq 2m$. Moreover, we know that for any $H_i$ such that
$|V(H_i)| \ge 2,$ the edge-expansion of each $H_i$ is at least
$\alpha,$ and hence the degree of each vertex in $H_i$ is at least
$\alpha.$

For each $H_i$ with at least 2 and at most $k$ vertices, we call
\NaiveCycleDecomposition, resulting in an amortized a total
running time of at most $O(k|E(H_i)|).$ In each $H_i,$ this results in
at least
\[
\abs{E\left(H_i\right)} - 2\abs{V\left(H_i\right)}
\ge
\abs{E\left(H_i\right)} - 4\alpha^{-1} \abs{E\left(H_i\right)}
\ge
 \frac{1}{2} \abs{E\left(H_i\right)},
\]
edges being incorporated into cycles of length at most $2\log n$,
which means at least $|E(H_i)| / 4 \log{n}$ cycles.

Lemma~\ref{lem:EdgeExpansionToConductance} gives that the
conductance of each $H_i$ is at least
\[
\phi
=
\frac{\alpha}{d_{\max}}
=
\frac{d_{\min}}{4 d_{\max} \gns\left( n \right)},
\]
and every vertex has degree at least
\[
\alpha
=
\frac{d_{\min}}{4 \gns\left( n \right)}
\geq
8000
\left( \frac{d_{\max}}{d_{\min}} \right)^2
\gns\left( n \right)^3
k \log n
\cdot
\frac{1}{4 \gns\left( n \right)}
=
100 \phi^{-2} k \log n.
\]
Thus, for $H_i$ with $|V(H_i)| \ge k,$
Lemma~\ref{lem:MoveEdgesExpander} gives that the number of cycles
produced in the partial decomposition is at least
\[
\left( \frac{d_{\min}}{d_{\max}} \right)^4
\frac{\abs{E\left( H_i \right)}}
{512 \cdot 10^3 \gns\left( n \right)^{4} k \log^2{n}},
\]
and the expected running time is
$O(|E(H_i)| k (d_{\max} / d_{\min})^2 \gns(n)^2 \log{n})$.
Combining these bounds over all connected components
of $H_i$ then gives the bounds on running time and
edges moved.

The bound on the lengths of the cycles in this
partial cycle decomposition also follows from
\[
25 \phi^{-2} \log n
\le
400
\left( \frac{d_{\max}}{d_{\min}} \right)^2
\gns\left( n \right)^2 \log n.
\]
\end{proof}

\subsection{Recursive Cycle Decomposition}
\label{subsec:Construction:Recursion}
Now that we have decreased the number of vertices by a large fraction
$k$, while decreasing the edge counts accordingly, we simply recurse
on this smaller subproblem.  This results in about $m / k$ edges being
placed into cycles in the original parent problem, which in turn
implies about $k$ iterations of recursion involving problems of size
about $m / k$.
Such a trade-off works well with the Master theorem,
except we need to handle the recursion explicitly due
to the overhead of $\gns(n)$ at each iteration.

In particular, Lemma~\ref{lem:Extend} gives that any cycle
on the smaller graph translates to a cycle on the original graph
that's are longer by a factor of $\poly(\gns(n), \log{n}) = n^{o(1)}$.
This means we can afford such multiplicative blowups a constant, or
sub-logarithmic number of times throughout the course of the algorithm.

However, one technical issue that we need to address first is the
reduction from arbitrary degree graphs to ones with uniform degrees.
Note that the high degree vertices may be concentrated among a small
number of vertices.  We will repeatedly peel away the low degree
vertices, from which we can obtain a graph where every vertex has
degree at least $\Delta$ assigned to it.  This graph may in turn have
vertices with high degrees: for example, a star can have $1$ edge
assigned to each vertex, but a center vertex with degree $n - 1$.  We
address this by splitting those vertices up into multiple vertices:
cycles in the resulting graph are still cycles in the original graph.
Pseudocode of this routine is in
Algorithm~\ref{alg:ExtractedBoundedDegreeGraph}.

\begin{algorithm}
\caption{$\ExtractBoundedDegreeGraph(G, \Delta)$}

\textbf{Input}:
Graph $G$ on $n$ vertices.\\
Degree threshold $\Delta$
such that each vertex of $G$ has degree at least $\Delta$.\\
Implicitly access to edges of $G$ in
linked list or binary search tree form.

\textbf{Output}: A graph $H$ along with mappings of vertices from
$H$ to $G$.

\begin{tight_enumerate}

\item Initialize $\hat{H}$ to be empty.

\item For each vertex $u$ of $G,$
  add the first $\min(\Delta, deg_G(u))$ edges from $u$
to $\hat{H}$.

\item Split any vertex $v$ in $\hat{H}$ with degree more than
  $2\Delta$ into $\left \lfloor \frac{deg(v)}{\Delta} \right \rfloor$
  vertices of degrees within $1$ of each other by distributing its
  edges arbitrarily.
\end{tight_enumerate}

\label{alg:ExtractedBoundedDegreeGraph}

\end{algorithm}

Note that as we will call this procedure repeatedly,
we will have the procedure of repeatedly removing low
degree vertices in the outer loop in order to amortize
its cost over all the removals.

The following lemma helps us to convert the graph
into one with bounded degrees.

\begin{lemma}\label{lem:DegreeReduction}
  Given a graph $G$ on $n$ vertices with every vertex having degree at
  least $\Delta$, the procedure
  $\ExtractBoundedDegreeGraph(G, \Delta)$ returns
  in $O(n \Delta)$ time a graph $H$ such that:
\begin{tight_enumerate}
  \item $H$ has at most $2n$ vertices.
  \item Every vertex in $H$ has degrees in the range $[\Delta, 2 \Delta]$.
  \item Every cycle in $H$ corresponds to a circuit in $G$.
\end{tight_enumerate}
\end{lemma}

\begin{proof}
  The assumption of minimum degree at least $\Delta$ in $G$,
  and the addition of first $\Delta$ neighbors to $H$
  ensures that $\hat{H}$ has minimum degree at least $\Delta$,
  and at most $\Delta n$ edges.
  This also means the total degree in $\hat{H}$ is at most $2 \Delta n$.

  The splitting process takes linear time, and ensures that all
  degrees in $H$ are in the range $[\Delta, 2 \Delta]$.  Combining
  this with the total degree also means that $H$ has at most $2n$
  vertices.

  The splitting process maintains a bijection between the edges of $G$
  and $H,$ and also ensures that all the edges incident on a vertex
  $u_H$ in $H$ are also incident to the original copy of that vertex
  in $G$.  Thus every cycle in $H$ can be mapped to a circuit in $G$.
\end{proof}

Repeatedly calling this and the edge moving procedure from
Section~\ref{subsec:Construction:MoveEdges} then leads to our overall
algorithm, whose pseudocode is in Algorithm~\ref{alg:CycleDecompose}.

\begin{algorithm}
\caption{$\ShortCycleAlgo(G, l, k)$}
\label{alg:CycleDecompose}

  \textbf{Input:} Graph $G(V,E)$ with $n$ vertices and $m$ edges,
  number of recursion layers $l$,
  reduction factor $k$.

  \textbf{Output:}
  A decomposition of $E$ into
  a union of cycles  $\mathcal{C}$
  and a set $E_{extra}$ of extra edges.

  \begin{tight_enumerate}
    
  \item If $l = 0,$ or $|V(G)| < k,$ return
    $\NaiveCycleDecomposition(G)$.
    \label{line:RecursiveDecomposePunt1}
    
  \item Set
    $\Delta \leftarrow \left( 64 \cdot 10^6 \gns(n)^4 \log^2{(2n)}
    \right)^{l} k,$ where $\gns(n)$ is the parameter given by Expander
    Decomposition (Lemma~\ref{lem:EdgeExpanderDecomposition}) for
    graphs with $2n$ vertices.
    
  \item Initialize $\mathcal{C}$ and $E_{extra}$ as empty.

  \item While $G$ has vertices remaining
  \begin{tight_enumerate}
  \item Repeatedly remove any vertex of $G$ with degree $< \Delta$,
    and add its incident edges to $E_{extra}$.
  \item $H \leftarrow \ExtractBoundedDegreeGraph(G, \Delta)$.
    \item $(S, \Ccal_{partial}) \leftarrow \MoveEdges(H, k)$.
    \item Create $H_{S}$ from $S$ and $\Ccal_{partial}$
        (as described in Lemma~\ref{lem:Extend})
    \item $\mathcal{C}_{H_S}
      \leftarrow \ShortCycleAlgo(H_S, l - 1, k)$.
\label{line:RecursiveDecomposePunt2}
\item Extend each cycle in $\mathcal{C}_{H_S}$ to a circuit in $H$ and
  in turn $G$ via $\Ccal_{partial}$ (according to
  Lemma~\ref{lem:Extend}). Split these circuits into cycles, add
  them to $\mathcal{C}$, and remove them from $G$.
  \end{tight_enumerate}
 
  \item Return $\mathcal{C}, E_{extra}$.
\end{tight_enumerate}
\end{algorithm}

We first give the guarantees of this algorithm
in terms of its output and the total sizes of
recursive problems that it invokes:
\begin{lemma}
\label{lem:Recursion}
For any $l \ge 0$, any $k \ge 10 \log n$, and any graph $G$ with $n$
vertices and $m$ edges, invoking $\CycleDecomposition(G, l, k)$
returns $\mathcal{C}$ and $E_{extra}$ so that the number edges in
$E_{extra}$ is at most
\[
\left( 64 \cdot 10^6 \gns\left( n \right)^4 \log^2{(2n)} \right)^{l} k n,
\]
and the length of cycles in $\mathcal{C}$ is at most
$ \left( 2000 \gns(n)^{2} \log (2n) \right)^{l+1}.$

The total running time cost is bounded by
$O(mn)$ if $l = 0$, and otherwise
\begin{tight_enumerate}
\item $O(m k \gns(n)^{2} \log n)$ in the cost of preprocessing and
  creating the $H_S$s,
\item Recursive calls to
  $\CycleDecomposition(H_S, l - 1, k)$ where the total edge
  count in $H_S$ across all steps is at most $2m,$ and each
  $H_S$ has at most $4n/k$ vertices.
\end{tight_enumerate}
\end{lemma}

\begin{proof}
  If $l = 0,$ it follows from Theorem~\ref{thm:NaiveCycleDecomposition} that
  the running time is bounded by $O(mn).$ Otherwise, if
  $|V(G)| \le k,$ again by Lemma~\ref{thm:NaiveCycleDecomposition} the
  running time is bounded by $O(mk) = O(mk \gns(n)^2 \log n)$ as
  desired.
  In either case, the size of $E_{extra}$ is bounded by $2n \le kn,$
  and the length of the cycles is at most $2 \log n,$ giving all the
  required guarantees in the case we call {\NaiveCycleDecomposition}.

  Otherwise, the size of $E_{extra}$ follows since the only edges
  added to $E_{extra}$ are from vertices with degree at most $\Delta.$

  Lemma~\ref{lem:DegreeReduction} gives that the minimum and maximum
  degrees of $H$ are within a factor of $2,$ and it has at most $2n$
  vertices.  So Lemma~\ref{lem:MoveEdges} gives that the lengths of
  the cycles in $\Ccal_{partial}$ is at most $2000 \gns(n)^2 \log (2n).$
  This and Lemma~\ref{lem:Extend} means that any cycle $\mathcal{C}_{\hat{G}}$
  corresponds to a cycle in $G$ whose   length is longer by a factor of
  $2000 \gns(n)^2 \log (2n)$.
  Conditioning on the recursion terminating, and applying induction
  with the guarantees of Lemma~\ref{thm:NaiveCycleDecomposition} as the
  base case gives the bounds on lengths of cycles returned.

  It remains to bound the running time, which we do by inductively
  showing that each recursive call makes significant call relative to
  the work done.
  The base case of $l = 0$ has been argued above.  For the inductive
  case, assume the result is true for $l - 1$.  Consider one iteration
  of the while loop, and say $\hat{n} \le n$ is the number of vertices
  remaining in $G$ at the start of this iteration.  The choice of
  $\Delta$ means that $H$ satisfies the degree requirement of
  Lemma~\ref{lem:MoveEdges} and has at most $2\hat{n}$ vertices.
  Thus, $H_S$ has at most $4\hat{n}/k$ vertices, and the number of
  cycles found is at least
\[
\frac{1}{16\cdot 10^6 k \gns\left( n \right)
\gns\left( n \right)^4 \log^2{(2n)}} \abs{E\left( H \right)}
\geq
\frac{1}{16\cdot 10^{6} k \gns\left( n \right) ^4 \log^2{(2n)}} \Delta \hat{n},
\]

Substituting in the value of $\Delta$ then gives this is at least
\[
\frac{4}{k}
\cdot
\left( 64 \cdot 10^6 \gns\left( n \right)^4 \log^2{(2n)} \right)^{l - 1} k \hat{n}.
\]

While by the inductive hypothesis, the number of extra edges returned
by the recursive call to $\CycleDecomposition(\hat{G}, l - 1, k)$
is at most
\[
\left( 64 \cdot 10^{6} \gns\left( n \right)^4 \log^2{(2n)} \right)^{l - 1} k
\frac{2\hat{n}}{k},
\]
so at least half of the edges of $H_S$ are now incorporated into the
extended cycles given by Lemma~\ref{lem:Extend}.
This means that each edge we examined can be `charged' to
the edges added to the cycles at a cost of
$2$ edges examined per edge added to $\Ccal$.
Since each edge can only be added to $\Ccal$ once,
the total sizes of the recursive calls is bounded by $2m.$

The total preprocessing costs now follow from the running time for
{\MoveEdges} given by Lemma~\ref{lem:MoveEdges}.
Thus, the inductive hypothesis holds for $l$ as well.
\end{proof}

The above lemma shows that the progress
$\ShortCycleAlgo$ makes per recursive call is close to the number of
edges in that problem.  Note that the progress made at the steps can
be highly non-uniform: a $\sqrt{n}$-sized clique in an otherwise
sparse graph will result in most of the recursive calls happening on
$\sqrt{n}$ sized graphs, with about $\sqrt{n}$ edges moved in each,
while the first few iterations may be removing about $n$ edges each.

Setting parameters in the above lemma then gives the proof of
Theorem~\ref{thm:ShortCycleDecomposition}, the main result of this
section.
\begin{proof}(Of Theorem~\ref{thm:ShortCycleDecomposition}) We
  consider setting parameters in Lemma~\ref{lem:Recursion}.  The
  lengths of the cycles constructed is
  $ ( 2000 \gns(n)^{2} \log (2n) )^{l+1}.$ Since
  $\gns(n) = \exp( O\left( \sqrt{\log{n} \log\log{n}})
  \right),$ picking $l = (\log{n})^{\nfrac{1}{4}},$ we get
  that the length of the cycles is
\[
  \exp\left( O\left(\left(\log{n} \log\log{n}\right)^{\nfrac{1}{2}}
    \right) \right) ^{\left(\log{n} \right)^{\nfrac{1}{4}} + O(1) } =
  \exp\left( O\left(\left(\log{n} \log\log{n}\right)^{\nfrac{3}{4}}
    \right) \right) = n^{o(1)}.
\]

Next, we bound the total time taken by the procedure. First, we count
the total cost of constructing $\hat{G}$'s across all $l$ levels of
recursion with size reduction factor $k.$ At each recursion level, we
invoke $\ShortCycleAlgo$ with a total of twice as many edges. The
time required for the construction at a level with $m$ input edges is
$O(mk\gns(n)^{2} \log n),$ resulting in a total construction time of  
\[
  O(2^{l} mk \gns\left( n \right)^{2} \log n).
\]

There is also the cost all the calls to \NaiveCycleDecomposition. The
calls that arise at the bottom recursion level $(l = 0)$ result in a
total running time of $O(mn(4/k)^{l} ).$ For all calls that
to {\NaiveCycleDecomposition} that arise since $|V(\hat{G})| \le k,$
the cumulative running time is bounded by $O(2^{l} m k).$

Thus, the total running time taken by $\ShortCycleAlgo(G, l, k)$
is bounded by
\[
  O\left(2^{l} mk \gns\left( n \right) ^{2} \log n + mn(4/k)^{l} \right).
\]
In order to balance the two terms, we pick $k^{l} = n,$ and thus
$k = \exp ( (\log n)^{\nfrac{3}{4}} ),$ giving a total running time of
\[
m \exp \left(O \left( (\log n)^{\nfrac{3}{4}} \right) \right)
=
  m^{1+o(1)}.
\]

Finally, the only edges not in cycles are those in $E_{extra}$ at the
top level of recursion, and the total number of such edges is bounded
by
\[
  \left( 64 \cdot 10^6 \gns\left( n \right) ^4 \log^2{(2n)} \right)^{l} k n = n \exp
  \left(O \left( (\log n)^{\nfrac{3}{4}} \right) \right) = n^{1+o(1)}.
\]

\end{proof}

%%% Local Variables:
%%% mode: latex
%%% TeX-master: "main"
%%% End:

%% file: unit-weight.tex
%!TEX root = main.tex

\section{Reduction to Unit Weight Case}
\label{sec:ReductionToUnit}

\newcommand{\almosteq}{\cong}

% \tim{It would be really really
%  great if you had a short outline of the general approach here and
%  perhaps earlier in the appendix. Moreover, sentences like "moving
%  one endpoint along $\dir{T}$" do not parse for me, and it's
%  not clear to me why we built all these lemmas around max weight
%  spanning trees, or what the important property of the max weight
%  spanning tree is. I sorta
%  stopped reading at this point, as I had too little of the high level
%  strategy to be able to effortlessly read or understand the rest of this
%  proof (though if it's helpful for me to dive into it, I can).}

We briefly describe how to reduce a general weighted graph to a sum of
graphs, each with edges of the same weight.  This reduction underlies the
constructions of spectral sparsifiers~\cite{SpielmanS08:journal}, as
well as previous results on graph sketching using expander
decompositions~\cite{AndoniCKQWZ16,JambulapatiS18}.  The main idea
stems from the observation that if $u$ is connected to a vertex $x$
by an edge of weight $1$, and $x$ in turn has a path to $y$
with each edge having weight at least $\poly(n)$,
then the edge $ux$ can be replaced by the edge $uy$
while incurring a negligible error of $1 / \poly(n)$.  Such an
observation plus bucketing of edge weights then allows one to `move'
the end points of lower weighted edges incident to some highly
weighted component to a single vertex in such a component.  For such
edges to not become self-loops, they must exit the component.  This in
turn reduces the number of connected components as we move to lower
weighted edge classes.  As a result, the number of vertices can be
bounded by the decrease in the number of connected components, which
is at most $n$.

Because our reduction only incurs $1 / \poly(n)$ error,
we can give a unified treatment of undirected and directed graphs.
This is via a definition of almost equality
that's analogous to the with high probability (w.h.p.) notation
for accumulating $1 / \poly(n)$ small failure probabilities.
Because an undirected edge is the sum of two directed edges,
one in each direction, we will define this notation in the more
general case of directed graphs from Theorem~\ref{thm:EulerianSparsifyFull}.
\begin{definition}
\label{def:AlmostEq}
Two (algorithmically generated)
(directed) graphs $\dir{G}$ and $\dir{H}$ are almost
equal if for any constant $\delta > 0$,
we can adjust constants in our algorithms so that
\[
\norm{
\LL_{G}^{\nfrac{+}{2}}
\left(\LL_{\dir{G}} - \LL_{\dir{H}} \right)
\LL_{G}^{\nfrac{+}{2}}
}
\leq
n^{-\delta}.
\]
We will denote this using $\dir{G} \almosteq \dir{H}$.
\end{definition}

Here, $\LL_G$ and $\LL_{\dir{G}}$ are defined as in
section~\ref{sec:eulerian}, with $\LL_G$ defined in
Equation~\ref{eq:Undirectification}.
Our result is a black box reduction among graphs that
keeps them almost equal.
The conditions that we keep are stronger than both degree-preserving
sparsification from Section~\ref{sec:Degree-Preserving}
and Eulerian sparsification from Section~\ref{sec:eulerian}.
We want to preserve both the in and out degrees exactly,
instead of just their differences.
\begin{theorem}
\label{thm:ReductionToUnit}
Any Eulerian directed graph $\dir{G}$ with arbitrary weights
on $n$ vertices and $m$ edges, represented either explicitly,
or implicitly as a sum of bicliques,
can be decomposed in nearly-linear
time into a sum of a graph $\dir{H}_{sparse}$ with $O(n \log{n})$ edges
plus graphs $\dir{H}_1, \dir{H}_2, \ldots \ldots$
such that
\[
\dir{G}
\almosteq
\dir{H}_{sparse} + \sum_{i} \dir{H}_i,
\]
and,
\begin{tight_enumerate}
\item The graph $\dir{H}_{sparse} + \sum \dir{H}_i$ has the exact
same in/out weighted  degrees at each vertex as in $\dir{H}$.
\item All edge weights in each $\dir{H}_{i}$ are powers of $2$.
\item The total number of vertices in $\{\dir{H}_i\}$ is $O(n \log^2{n})$.
\item The total number of edges in $\{\dir{H}_i\}$ is $O(m \log{n})$.
\item If $\dir{G}$ is represented implicitly as a sum of bicliques,
the total representation size of $\dir{H}_{i}$ returned is larger
by a factor of $O(\log^2{n})$.
\end{tight_enumerate}
\end{theorem}

Because we only work with $1/\poly(n)$ sized perturbations,
we will pick $\dir{H}_{sparse}$ from max-weighted spanning trees
of the undirected graph $G$.
This choice is convenient for aggregating the resulting changes
in in/out degrees along the tree, and fixing them with paths
along the tree.
\begin{lemma}
\label{lem:AddTree}
For any directed graph $\dir{G}$, choosing $T$ to be $1 / \poly(n)$
times any subgraph of $G$, the undirected version of $\dir{G}$, gives
\[
\dir{G} + T
\almosteq
\dir{G}.
\]
\end{lemma}

\begin{proof}
% Let $\widehat{T}$ be the max weighted spanning tree of $G$.
  Let $\widehat{T}$ be any subgraph of $G.$ Thus,
\[
\norm{\LL_{G}^{\nfrac{+}{2}}
\LL_{\widehat{T}}
\LL_{G}^{\nfrac{+}{2}}
}
\leq 1,
\]
and the result follows from the rescaling.
\end{proof}

For the case with implicitly represented bicliques, we mean that each
vertex $u$ has weight $w_u$, and an edge $uv$ has weight $w_u \cdot
w_v$ -- and thus the biclique can be represented by two sets of vertices
and vertex weights. 
Note that the max weight spanning tree of a weighted biclique
is given by the two stars centered at the vertices with
maximum weights.
So this computation still takes time nearly-linear in the
number of vertices involved in such a representation.

As we need to preserve degrees exactly, it is preferably to
only make adjustments along even length cycles.
We do so by restricting to bipartite graphs.
This is done by decompositions similar to the random bipartition
picked in the degree-preserving sparsifiers and sketches
in Algorithms~\ref{alg:SparsifyOnce} and~\ref{alg:DecomposeAndSample}.
However, in our case, we apply this as a preprocessing step
to decompose into bipartite graphs completely, instead of iteratively
removing edges from a bipartition that captures at least half
the edges at a time.
\begin{lemma}
\label{lem:DecomposeBipartite}
Any directed graph $\dir{G}$ can be written as a sum
of directed graphs $\dir{G}_1, \dir{G}_2 \ldots$
so that
\begin{tight_enumerate}
\item For each $i$, the undirected support of $\dir{G}_i$, $G_i$
is bipartite.
\item The total number of edges in $\{\dir{G}_i\}$ is $m$.
\item The the total number of vertices in $\{\dir{G}_i\}$ is
  $O(n \log{n})$.
\end{tight_enumerate}
Such a decomposition takes $O(m)$ time if $\dir{G}$
is given explicitly, or $O(n(\Kcal_{B}))$ time if it's
specified implicitly as a sum of bicliques.
Furthermore, in the latter case, the number of vertices
in the resulting cliques is $O(n(\Kcal_{B}) \log{n})$.
\end{lemma}

\begin{proof}
A greedy bipartition captures at least half of the edges,
after which we can recursively decompose the edges in
the two halves.
Since the number of remaining edges halves after each step,
we finish in $O(\log{n})$ steps, and the total bound on
the number of vertices follows from the recursive calls
on the two halves being vertex disjoint.

In the case of implicit representations as bicliques,
the bound follows from the fact that the intersection of a biclique
with a bipartition is a biclique.
Moreover, the restriction to the two sides gives two bicliques
on the boundary of the bipartition with disjoint vertex sets.
\end{proof}

Our proof then becomes modifying the graph in ways that introduce
error significantly less than the weight of this tree added.  Our most
important lemma is a statement for `locally' moving around the end point
of an edge, but by a distance of $2$ along the tree.
\begin{lemma}
  \label{lem:LocalMove} Suppose $u$, $x_1$, $x_2$, $x_3$ are four
vertices in a directed graph $\dir{G}$ such that:
  \begin{tight_itemize}
  \item There is a directed edge of weight $1$ from $u$ to $x_1$
  \item The weights of the directed edges between $x_1 x_2$, $x_2 x_3$
    in both directions exceed $\poly(n)$.
  \end{tight_itemize} Then the graphs $H$ produced by:
  \begin{tight_itemize}
  \item removing weight $1$ from $u \rightarrow x_1$ and $x_2
    \rightarrow x_3$,
  \item adding in $u \rightarrow x_3$ and $x_2 \rightarrow x_1$,
  \end{tight_itemize} has the same in/out degrees as $\dir{G}$, and
  satisfies $H \almosteq G$.
  
  The same bound also holds in the case of an edge in the other
  direction from $x_1$ to $u$, but with the operations
  \begin{tight_itemize}
  \item removing weight $1$ from $x_1 \rightarrow u$, and $x_3
    \rightarrow x_2$,
  \item adding in $x_3 \rightarrow u$ and $x_1 \rightarrow x_2$.
  \end{tight_itemize}
\end{lemma}

\begin{proof}
  The difference produced is the matrix
  \[
    \left[
      \begin{array}{cccc}
        0 & 0 & 0 & 0\\
        1 & 0 & -1 & 0 \\
        0 & 0 & 0 & 0\\
        -1 & 0 & 1 & 0
      \end{array}
    \right],
  \]
  where the vertices are arranged in the order
  $u, x_1, x_2, x_3$. This equals
  $
    \cchi_{x_1 x_3} 
    \cchi_{u x_2}^{\top},
  $
  and as the weight in $G$ between $x_1x_2$ and $x_2 x_3$
  are both $\poly(n)$, we have
  \[
    \norm{\LL_G^{\nfrac{+}{2}} \cchi_{x_1 x_2} }
    \leq
    \sqrt{ \cchi_{x_1 x_2}^{\top}       \LL_G^{+} \cchi_{x_1 x_2}}
    \leq
    1 / \poly\left( n \right).
  \]
  and similarly $\norm{ \LL_G^{\nfrac{+}{2}} \cchi_{x_2 x_3}}
  \leq 1 / \poly(n).$ Thus,  by triangle inequality,
  \[
    \norm{ \LL_G^{\nfrac{+}{2}} \cchi_{x_1 x_3}}
    \leq
    \norm{ \LL_G^{\nfrac{+}{2}} \cchi_{x_1 x_2} }
    +
    \norm{ \LL_G^{\nfrac{+}{2}} \cchi_{x_2 x_3} }
    \leq
    1 / \poly\left( n \right).
  \]
  
  Since $ux_1$ is an edge in $G$ with weight 1, we have
    \[
    \norm{\LL_G^{\nfrac{+}{2}} \cchi_{u x_1} }
    \leq
    \sqrt{ \cchi_{u x_1}^{\top}       \LL_G^{+} \cchi_{u x_1}}
    \leq 1,
  \]
  and hence by triangle inequality,
  \[
    \norm{ \LL_G^{\nfrac{+}{2}} \cchi_{u x_2}}
    \leq
    \norm{ \LL_G^{\nfrac{+}{2}} \cchi_{u x_1} }
    +
    \norm{ \LL_G^{\nfrac{+}{2}} \cchi_{x_1 x_2} }
    \leq
    1 + 1 / \poly\left( n \right).
  \]

  Combining the above bounds with the Cauchy-Schwarz inequality then gives
  \[
    \norm{\LL_G^{\nfrac{+}{2}} \cchi_{x_1 x_3} 
      \cchi_{u x_2}^{\top} \LL_G^{\nfrac{+}{2}}}
    =
    \abs{\cchi_{u x_2} \LL_G^{+} \cchi_{x_1 x_3}}
    \leq
    \norm{\LL_G^{\nfrac{+}{2}} \cchi_{u x_2}} \norm{\LL_G^{\nfrac{+}{2}} \cchi_{x_1 x_3}}
    \leq
    1 / \poly(n).
  \]
  The case with the other direction follows from the difference
  being the matrix
  \[
    \left[
      \begin{array}{cccc}
        0 & 0 & 0 & 0\\
        -1 & 0 & 1 & 0 \\
        0 & 0 & 0 & 0\\
        1 & 0 & -1 & 0
      \end{array}
    \right],
  \]
  which is exactly the negation of the above matrix.
\end{proof}

This lemma essentially allows us to `contract' higher weighted edges:
an edge of weight $w$ can be moved along paths of weight exceeding
$w\cdot \poly(n)$ while incurring negligible errors.
In particular, it plus the condition of the undirected support being bipartite
allows us to remove the trailing bits of the edge weights by canceling
them along the maximum weighted spanning tree.
\begin{lemma}
\label{lem:PowersOfTwo}
Any directed graph $\dir{G}$ whose undirected support is connected 
can be written as a sum of a directed graph $\dir{T}$ and
several directed graphs $\dir{H}_{i}$ such that:
\begin{tight_itemize}
\item The total number of edges in $\dir{T}$ and $\{\dir{H}_i\}$
is $O(m \log{n}).$
\item The undirected support of $\dir{T}$ is a tree.
\item The in/out degrees in $\dir{G}$ and
$\sum_{i} \dir{H}_{i} + \dir{T}$ are the same.
\item All edge weights in $\dir{H}_i$ are $2^{i}$.
\item $\dir{G} \almosteq \dir{T} + \sum_{i} \dir{H}$.
\end{tight_itemize}
Furthermore, such a decomposition can be generated in
$O(m\log{n})$ time if $\dir{G}$ is given explicitly,
or $O((n(\Kcal_B) + m) \log{n})$ if we're given bicliques
$\Kcal_B$ instead.
\end{lemma}

\begin{proof}
  We start by letting $\dir{T}$ be the maximum weight spanning tree in
  $G,$ the undirectification of $\dir{G},$ scaled by $1/\poly(n).$
  Observe that we still treat $\dir{T}$ as a directed graph, with
  edges of equal weight in either direction.  Lemma~\ref{lem:AddTree}
  allows us to add $\dir{T}$ to $\dir{G}$ while incurring an error of
  $1 / \poly(n)$.

  We now form the graphs $\dir{H}_i$s by taking the first $O(\log{n})$
  leading bits of each edge's weights in $\dir{G}$.  This gives
  \[
    \dir{G} = \dir{G}_{trailing} + \sum_{i} \dir{H}_i,
  \]
  where the total number of edges in $\dir{H}_i$ is $O(m \log{n})$ by
  construction. The constant factor in $O(\cdot)$ is chosen so that
  each edge in $\dir{G}_{trailing}$ has weight less than 
  $1 / \poly(n)$ times the least weight of any edge on the
  corresponding directed path in $\dir{T}.$

It remains to remove the edges in $\dir{G}_{trailing}$.  We deal with
the edges one at a time.  For a single edge in
  $\dir{G}_{trailing},$ we use Lemma~\ref{lem:LocalMove} repeatedly so
  as to locally move one of the end points of the edge closer to the
  other (in terms of the hop distance in $T$), along the unique path
  in $T$ between the end points of this edge.
  Specifically:
  \begin{itemize}
  \item By Lemma~\ref{lem:LocalMove},
  each such step incurs error that's at most $1 / \poly(n)$.  
  \item Because $G$, and hence $T$ is bipartite, we can repeat this until
  the end points are distance $1$ apart in $T$.
  \item This then coincides with an edge of $T$,
  so we can incorporate this edge in $\dir{T}$ by adjusting
  the edge weights in $\dir{T}$ appropriately,
  and discarding the original edge entirely.
  \end{itemize}
  Effectively this process `reroutes' an edge in $\dir{G}_{trailing}$ along
  the unique path in $T$ connecting its end points.

  Firstly, observe that both the local-move steps, and the tree-merge
  steps do not change the in and out degrees in the sum graph. Thus,
  the in/out degrees of all vertices are preserved.
  
  It remains to show that this removal process keeps the graph almost
  equal.  Due to the compounding of almost equality, we can analyze
  this one edge at a time.  As there are at most $m \leq n^2$ edges,
  and each edge incurs at most $n$ local-move steps, the overall
  perturbation caused to the graph is still $1 / \poly(n).$ Thus, we
  can apply this process repeatedly until all edges in
  $\dir{G}_{trailing}$ are gone.  At that point, we are left with
  $\dir{T}$ (with new, adjusted edge weights), plus the graph
  $\dir{H}_i$ whose total edge count is $O(m \log{n})$.
  
%   Then consider repeatedly applying 
% to move the end points of $e$ closer to each other (in terms of distance
% in the support of $\dir{T}$).
% Each such step incurs an

% Eventually, these end points coincide with an edge of $\dir{T}$,
% and can be incorporated into $\dir{T}$ by changing of edge weights.

  Note that due to the adjustment in weights, the tree $\dir{T}$ is no
  longer undirected (or even Eulerian). However, the weights of the
  edges in $\dir{T}$ have only been perturbed by a multiplicative
  $1 + 1 / \poly(n)$ factors.
  % them along the maximum weighted spanning tree of
  % $\dir{G}$.  We let this tree be $\dir{T}$: it starts off as
  % undirected, it will become directed due to the removal of edges from
  % $\dir{G}_{trailing}$.  Specifically, for each edge in
  % $\dir{G}_{trailing}$, we remove it, and fix the degrees of its two
  % end points by routing along the fundamental cycle between them in
  % $\dir{T}$.

In the case where all the edges are explicitly given, this procedure can
be implemented efficiently using dynamic trees.
In the case of working with implicit representations, we perform recursive
centroid decomposition on $\dir{T}$, after which the
discrepancies only need to be propagated up to root.
Then the amount propagated up to root can be calculated explicitly
at each vertex, and then accumulated using a single depth-first search.
In the case of given implicit representations as bicliques,
the changes in degrees caused by the rounding can still be calculated
explicitly, so routing the degree differences along $\dir{T}$
still remains the same.
\end{proof}

Then it remains to `shrink' the vertex set of each edge class (all
edges with the same weight that's a power of 2).  Note that the
polynomially bounded weights case corresponds to there being
$O(\log{n})$ edge classes, at which point the result follows.  We
partition these edge classes into $O(\log n)$ buckets, where the
$i^\textrm{th}$ bucket has edges of weight
$2^{c_{spread} j \log{n} + i},$ for some positive integer $j$ and a
some  constant $c_{spread}.$ Thus, we have split our graph into
$O(\log n)$ graphs such that in each of these graphs, two edges either
have the same weight, or their weights are apart by a factor more than  $n^{c_{spread}}.$

% By
% partitioning the edges by their weight classes, we can `split' the
% graph into a sum of $O(\log{n})$ graphs, each containing edges whose
% weights are of the form $2^{c_{spread} \log{n} i}$ for some constant
% $c_{spread}$.  That is, the weights of edges in each such graph are
% apart by a factor of more than $n^{c_{spread}}$.

On such a graph, Lemma~\ref{lem:LocalMove} enables us to 
`move' the end points of a lower weighted edge along paths
connected by a higher weighted class almost for free.
This enables us to `shrink' each connected component
to two vertices connected in $\dir{T}$, and move all
lower weighted edges to them.
It in turn implies that all remaining edges (that are not absorbed
into $\dir{T}$ are between components, giving a bound related
to the reduction in the number of connected components
of $\dir{T}$ as we introduce edge classes.

However, as Lemma~\ref{lem:LocalMove} only allows moving by two edges
at a time, we still need to have two representative vertices per such
connected component.  This however does not affect the overall sizes
across the edge classes.
\begin{proof}(of Theorem~\ref{thm:ReductionToUnit})
  Lemma~\ref{lem:DecomposeBipartite} allows us to reduce to the
  bipartite case with an overhead of $O(\log{n})$, and
  Lemma~\ref{lem:PowersOfTwo} allows us to work with a graph that's a
  sum of $H_i$s, where each $H_i$ contains edges with weights $2^{i}$.
  The fact that we start with at most $O(n^2)$ edges also means that
  the number of non-empty $H_i$s is at most $\poly(n)$ (assuming the
  weights are at most $\exp(\poly(n)).$
  
  We will create further separation between the edge weight classes
  by bucketing the $H_i$s.
  For a constant $\xi \geq \Theta( \log{n})$ that we will choose later,
  we will let the $j\textsuperscript{th}$ class include all indices with
  \[
  i \equiv j \pmod \xi ,
  \]
  that is, the buckets containing edge weights
  \[
  2^{l \cdot \xi + j}
  \]
  for integers $l$.
  This introduces an additional overhead of $\xi = O(\log{n})$ in the 
  number of buckets, but ensures that two edges from the same bucket
  have weights that are either the same, or apart by a factor of $poly(n)$.
  
  This $poly(n)$ separation within each bucket ensures a clear separation
  by edge weights within each bucket.
  Thus we can perform the following shrinkage procedure:
  \begin{tight_enumerate}
  \item For indices $i \equiv j \pmod \xi$ in \emph{decreasing} order
    \begin{tight_enumerate}
    \item For each connected component of $G / H_{i} / H_{i + \xi} / H_{i + 2 \xi} \ldots$,
        \begin{tight_enumerate}
          \item Pick two representative vertices, one per side of the bipartition.
          \item Move the end points of each edge in $H_{i}$ to the
                  representative vertex in its connected component of 
                  $G / H_{i} / H_{i + \xi} / H_{i + 2 \xi} \ldots$ that's on the same
                  side of its bipartition using Lemma~\ref{lem:LocalMove}.
          \item Remove any self loops.
        \end{tight_enumerate}
    \end{tight_enumerate}
  \end{tight_enumerate}

  To efficiently implement this move, observe that this algorithm is identical
  to Kruskal's algorithm for computing a maximum weighted spanning tree
  in the $j\textsuperscript{th}$ edge bucket.
  So we can build the maximum spanning tree of the $j\textsuperscript{th}$
  edge bucket, $\dir{T}^{(j)}$ before the loop.
  The changes to degrees caused by~\ref{lem:LocalMove} can then be
  computed in the same way as Lemma~\ref{lem:PowersOfTwo},
  either implicitly or explicitly.
  As the lengths of these moves is $O(n)$, Lemma~\ref{lem:LocalMove}
  and the choice of $\xi \geq \Omega( \log{n})$ gives a total error of $1 / poly(n)$.

  So it remains to bound the total number of vertices that
  each $H_i$ is moved to.
 Here the critical fact that we utilize is that a graph without self loops
 with edges incident to $\hat{n}$ vertices has at most $\hat{n} / 2$ connected
 components.
 This is because the lack of self loops means each connected component
 involves at least two vertices.
 
 Then we can bound the number of vertices that $H_i$ gets moved
 to by the total number of connected components.
 Specifically, consider the potential function:
 \[
 \Phi\left( i \right)
 :=
 \text{number of connected component in } G / H_{i} / H_{i + \xi} / H_{i + 2 \xi} \ldots.
 \]
 Now suppose the edges of $H_i$ after moving and removal of self
 loops are incident to $t(i)$ different components of
 $G / H_{i + \xi} / H_{i + 2 \xi} / H_{i + 3 \xi} \ldots$.
 Then contracting these edges decreases the number of connected components
 by at least $t / 2$, or formally
 \[
  \Phi\left( i \right)
  \leq
   \Phi\left( i + \xi \right) - t\left( i \right) / 2.
 \]
 As $\Phi(\cdot)$ is between $[1, n]$,
 this implies a bound of $O(n)$ on the sum of $t(i)$ per edge bucket,
 for a total of $O(n \log{n})$ vertices among all edge classes after the moves.
 
Putting back the overhead of $O(\log{n})$ from bucketing on the
weight classes, and another overhead of $O(\log{n})$ from
the decomposition to bipartite support then gives the overall bound.
\end{proof}

%%% Local Variables:
%%% mode: latex
%%% TeX-master: "main"
%%% End:

%% file: determinant.tex
%!TEX root = main.tex

\section{Faster Determinant Estimation
Using Faster Resistance Estimation}
\label{sec:Determinant}

We provide a short sketch showing how our results lead to the faster
determinant estimation routine from Corollary~\ref{cor:ITSOVER90PAGES}.
The algorithm utilizes the faster resistance estimation procedure from
Theorem~\ref{thm:MainER}, and uses the machinery in the algorithm
from Durfee et. al.~\cite{DurfeePPR17} with slight tweaks on the
parameters in that paper.
The detailed references that we make below are all w.r.t
version 1 of the same, which is available at
\textit{https://arxiv.org/pdf/1705.00985.pdf}.

We start by tightening the parameters of the approximate Schur
complement algorithm of~\cite{DurfeePPR17}.
Its pseudocode is given in Algorithm $4$ in that paper, and its guarantees
are in Theorem $5.3$.
The two key terms are:
\begin{enumerate}
\item $s$, the number of edges sampled.
\item $\eps_{ER}$, the accuracy to which effective resistances on the
edges need to be sampled.
\end{enumerate}
and the tradeoffs given by these bounds for a variance
of at most $\delta$ in the determinant estimation algorithm is given
in Theorem 4.1.
The guarantee on the parameters in Theorem 5.3 for approximate Schur
complements can be summarized as
\[
\frac{\eps_{ER}^2 n^2}{s}
+ O\left( \frac{n^3}{s^2} \right)
\leq \delta.
\]
Based on this tradeoff, we can tighten these parameters as below.
The procedure has one additional technical restriction.
It can only work when the vertices to be removed
are $1.1$-diagonally-dominant ($1.1$-DD): every vertex has at least
$10\%$ of its weighted degree leaving the set.
Our modifications do not affect this step, so we only state it
for completeness in the claims below.

\begin{lemma} 
\label{lem:SchurSparseBetter}
  There is a procedure {\SchurSparse} that takes a graph $G$
with $n$ vertices, $m$ edges,
a $1.1$-DD subset of vertices $V_2$,
and error $\delta > 0$, returns a graph $H^{V_1}$
with $O(n^{1.5+o(1)} \delta^{-0.5})$ edges in
$O(m^{1 + o(1)} + n^{1.875 + o(1)} \delta^{-0.875})$ expected time
such that the distribution over $H^{V_1}$  satisfies:
\[
\exp\left(-\delta\right) 
\det\left( \LL_{\textsc{SC}\left(G, V_1\right), -n} \right) 
  \leq
\expec{H^{V1}}{\det\left( \LL_{H^{V_1}, -n} \right)}
  \leq
\exp\left(\delta\right)
\det\left( \LL_{\textsc{SC}\left(G, V_1\right), -n} \right),
\]
and
\[
\expec{H^{V1}}{
  \det\left( \LL_{H^{V_1}, -n}\right)^2}
\leq \exp\left( \delta \right)
\expec{H^{V1}}{\det\left( \LL_{H^{V_1}, -n}\right)}^2.
\]
\end{lemma}
\begin{proof}
Consider the parameter tradeoffs as discussed above
just before the statement of this lemma.
We will set $s$, the number of edges sampled to
\[
 s = n^{1.5}\delta^{-0.5},
 \]
which in turn necessitates computing their effective resistances
to multiplicative accuracy
 \[
 \eps_{ER} = n^{-0.25} \delta^{0.25}.
 \]
plugging these into Theorem~\ref{thm:MainER} gives that we can compute
the effective resistance of all $t = n^{1.5}\delta^{-0.5}$ edges to this accuracy in time
\[
m ^{1 + o\left( 1 \right)} +
  n^{1.5}\delta^{-0.5} n^{o\left(1\right)}
  \left( n^{-0.25} \delta^{0.25} \right)^{-1.5}
  =
m ^{1 + o\left( 1 \right)} 
+ n^{1.875 + o\left( 1 \right) } \delta^{-0.875}.
\]
The result then follows from incorporating the $\otil{m}$
overhead in other steps from Theorem 4.1 of~\cite{DurfeePPR17}.
\end{proof}

We remark that this new choice of parameters
does not improve the runtime of the algorithm in~\cite{DurfeePPR17}
using previous algorithms for estimating effective resistances.

This change in turn propagates into the recursive algorithm
as outlined in Algorithm 6 of~\cite{DurfeePPR17}.
This algorithm repeatedly finds $1.1$-DD subsets that contain
a constant fraction of vertices, and recurses on both the Schur
complement and the matrix minor consisting of the removed
vertices.
A key difference between this algorithm and the recursive
resistance approximation routine from Section~\ref{subsec:ERReductions}
is that the overall determinant is a product of the determinants of the two
subproblems.
This means errors on them both accumulate into the overall error.
This is addressed in~\cite{DurfeePPR17} by letting the allowed variance
vary with problem sizes.
Below we outline how the better Schur complement sparsifier
from Lemma~\ref{lem:SchurSparseBetter} can be readily incorporated
to give an improved overall bound.

\begin{proof}(of Corollary~\ref{cor:ITSOVER90PAGES})
  In this proof, we use the notation found on page 30 and 31
  of~\cite{DurfeePPR17}. Consider setting the allowed variance onto a graph
Schur complemented onto vertex set $V_1(i)$ as:
\[
\delta'
=
\frac{\delta \abs{V_{1}\left(i\right)}}{n}.
\]
This is exactly the same as the second to last equation on page 31
of~\cite{DurfeePPR17}, so the overall error guarantees still hold.

Lemma~\ref{lem:SchurSparseBetter} gives that the cost of computing
this is:
\[
  O\left(\abs{V_1\left( i \right)}^{1.875 + o(1)}
   (\delta^{\prime})^{-0.875}\right)
=
  O\left(\abs{V_1\left( i \right)}^{1.875+o(1)}
   \left( \frac{\abs{V_{1}\left(i\right)}}{n} \right)^{-0.875}
   \delta^{-0.875}\right) \]

 \[ =
  \abs{V_1\left(i\right)} \cdot n^{0.875 + o(1)} \delta^{-0.875},
\]
plus an initial overhead of $m^{1 + o(1)}$.
This initial overhead term is not present in subsequent layers of the recursion
due to the vertex-based bound on edge count
  in Lemma~\ref{lem:SchurSparseBetter}.
The fact that the vertex count decreases by a constant factor
  at each step of the recursion outlined in~\cite{DurfeePPR17} page 31 means that an instance on $V_1(i)$ has edge
count bounded by
\[
\abs{V_1\left(i\right)}^{1.5} \delta^{-0.5}.
\]

So the overall cost at each level of the recursion is still bounded by
\[
\sum_{i} \abs{V_1\left(i\right)} \cdot
  n^{0.875 + o(1)} \delta^{-0.875}
  \leq
    n^{1.875 + o(1)} \delta^{-0.875}
\]
which, summed over the $O(\log{n})$ layers of recursion, gives the total.
Note that we need to set $\delta = \eps^{2}$ to obtain a $1 \pm \eps$
approximation with constant probability.
\end{proof}

%%% Local Variables:
%%% mode: latex
%%% TeX-master: "main"
%%% End: